\documentclass[11pt]{article}
\usepackage{amsmath}
\usepackage{amsfonts}
\usepackage{amssymb}
\usepackage{amsthm}
\DeclareMathOperator{\dist}{dist}

\usepackage{geometry}
\usepackage{setspace}
\usepackage{graphicx}
\usepackage{subcaption}
\usepackage{multirow}
\usepackage{multicol}
\usepackage{array}

\newcommand{\paperfigdir}{figures}
\newcommand{\paperhalfwidth}{0.48\textwidth}

\newcommand{\paperfigfull}[1]{\includegraphics[width=\textwidth]{\paperfigdir/#1}}
\newcommand{\paperfigfullnarrow}[1]{\includegraphics[width=0.9\textwidth]{\paperfigdir/#1}}
\newcommand{\paperfigfullwide}[1]{\includegraphics[width=0.9\textwidth]{\paperfigdir/#1}}
\newcommand{\paperfigfullmap}[1]{\includegraphics[width=\textwidth]{\paperfigdir/#1}}
\newcommand{\paperfighalf}[1]{\includegraphics[width=\linewidth]{\paperfigdir/#1}}
\newcommand{\paperfighalfalone}[1]{%
  \begin{minipage}{\paperhalfwidth}%
    \centering
    \includegraphics[width=\linewidth]{\paperfigdir/#1}%
  \end{minipage}%
}

\usepackage{tikz}
\usetikzlibrary{arrows.meta,calc,positioning}

\usepackage{latexsym}
\usepackage{lipsum}
\usepackage{enumerate}
\usepackage{comment}
\usepackage[toc,title,titletoc,header]{appendix}
\usepackage[bottom]{footmisc}
\usepackage{xr}
\usepackage{pifont}
\usepackage{subcaption}

\usepackage{lineno}
\setpagewiselinenumbers

\usepackage[pdfborder={0 0 0}]{hyperref}
\hypersetup{
  colorlinks=true,
  linkcolor=magenta,
  citecolor=blue,
  filecolor=magenta,
  linkbordercolor={0 1 1},
  citebordercolor={1 0 0}
}

\usepackage[authoryear]{natbib}

\numberwithin{equation}{section}
\newcommand{\myref}[2]{\hyperref[#1]{#2}}

\newtheorem{theorem}{Theorem}[section]
\newtheorem{lemma}{Lemma}[section]

\newtheorem{assumption}{Assumption}[section]

\newtheorem{remark}{Remark}[section]

\newcounter{assumptionM}
\newcounter{assumptionA}
\def\theassumptionM{M.\arabic{assumptionM}}
\def\theassumptionA{A.\arabic{assumptionA}}
\newenvironment{runningexample}[1][]{%
    \vspace{0.5\baselineskip} 
    \noindent\textbf{Running Example.}\ #1%
    \begin{itshape}%
}{%
    \end{itshape}%
    \vspace{0.5\baselineskip} 
}

\begin{document}
	\relax
	\hypersetup{pageanchor=false}
	
	\hypersetup{pageanchor=true}

\author{
Federico A. Bugni\\
Department of Economics\\
Northwestern University\\
\url{federico.bugni@northwestern.edu}
\and
Federico Crippa\\
Department of Economics\\
University of California, Berkeley \\
\url{federico.crippa@berkeley.edu}
\and
Daniel Restrepo\\
Department of Mathematics\\
University of Minnesota\\
\url{drestrep@umn.edu}
}

\bigskip
\title{Manipulation Testing in Boundary Discontinuity Designs\thanks{Corresponding author: \href{mailto:federico.bugni@northwestern.edu}{federico.bugni@northwestern.edu}. We thank Matias Cattaneo, Bruno Fava, Amilcar Velez, and seminar participants at Northwestern's reading group for helpful comments on this paper. We thank Mat{\'\i}as Mart{\'\i}nez for having shared the data.
}
}

\date{August 31, 2026}

\maketitle

\vspace{-0.3in}
\thispagestyle{empty}

\begin{spacing}{1.3}
\begin{abstract}
We propose the first manipulation test designed for boundary discontinuity designs (BDDs) with general boundary shapes. A BDD is a multidimensional extension of the regression discontinuity design (RDD) in which treatment assignment is determined by whether the multidimensional running variable crosses a lower-dimensional boundary set. The test avoids multivariate density estimation and builds on the observation that, in the absence of manipulation, observations near the boundary should be approximately evenly split between treatment and control within arbitrary groups defined by their projections onto the boundary. We test this implication using a collection of binomial balance tests on observations near the boundary, with groups formed by k-means clustering. We establish the asymptotic validity of the test under suitable regularity conditions. We also evaluate finite-sample performance through Monte Carlo simulations and illustrate the test in three empirical applications.
\end{abstract}
\end{spacing}



\thispagestyle{empty} 

\newpage
\hypersetup{pageanchor=true}
\setcounter{page}{1}
\section{Introduction}

This paper proposes the first manipulation test for boundary discontinuity designs (BDDs) with general boundary shapes. A BDD is a regression discontinuity design (RDD) in which assignment status is determined by a vector running variable. Leading examples are geographic RDDs, in which assignment status depends on location relative to a geographic boundary, with latitude and longitude serving as running variables, and multiscore RDDs, in which assignment status depends on several observable running variables, such as eligibility for a scholarship that requires both a math GPA and an English GPA above given thresholds. BDDs have become increasingly common in empirical work because georeferenced data are more widely available and treatment-assignment rules are growing more complex. See \cite{cattaneo2026boundary} for a recent survey of developments in the theory and practice of BDDs.

The RDD, in both one- and multidimensional settings, is widely used to estimate treatment effects, with numerous applications in economics and other social sciences. In the one-dimensional setting, assignment status for unit $i$ is determined by whether a scalar running variable $Z_i$ crosses a known cutoff. Under the assumption that the conditional expectations of the potential outcomes are continuous at the cutoff, the RDD identifies the conditional average treatment effect at the cutoff; see \cite{hahn/todd/vanderklaauw:2001}, as well as \cite{lee/lemieux:2010} and \cite{imbens/lemieux:2008} for surveys. In a BDD, the vector running variable $Z_i \in \mathbb{R}^d$ has dimension $d>1$, and the cutoff is replaced by a boundary in $\mathbb{R}^d$, which we denote by $\mathcal{B}$. By a natural extension of the one-dimensional argument, the BDD identifies the conditional average treatment effect at each point on the boundary.

While the assumptions required for nonparametric identification in an RDD are fundamentally untestable, empirical researchers routinely assess RDD validity by examining testable implications of stronger identification assumptions. One such implication, formalized in \cite{lee:2008}, is that units have only limited control over the running variable. In the one-dimensional RDD, this implication requires the running variable density to be continuous at the cutoff. Conversely, a discontinuity in the density at the cutoff may indicate that units can manipulate the running variable to affect treatment assignment, potentially undermining the design's validity. This motivates the manipulation test of \cite{mccrary:2008}, which estimates the density of the running variable on both sides of the cutoff and tests whether the two limits are equal. More recently, \cite{bugni/canay:2021} proposed an alternative test based on $g$-order statistics that is valid under weaker assumptions. As we explain below, our methodology is closer in spirit to this latter approach.

The rationale behind the \cite{mccrary:2008} manipulation test extends naturally to the BDD setting; see Appendix \ref{sec:leeExtension} for a formal derivation. In a BDD, the corresponding testable implication is that the joint density of the vector running variable is continuous across the treatment boundary at every point $z \in \mathcal B$. Although this diagnostic is important for empirical applications, no general test is currently available for the joint-density-continuity null in BDDs with general boundaries. This paper develops such a test.

A natural approach to construct such a test would be to extend the methodology of \cite{mccrary:2008} to the multidimensional setting. This would involve estimating the joint density of the vector running variable in the treatment and control regions near each boundary point, and then testing whether the two one-sided limits are equal. Such a task poses a challenging econometric problem. It requires approximating the distribution of a test statistic based on a continuum of multidimensional nonparametric estimators. It also requires reliable density estimation near each point on a general boundary, a task that is severely hindered by the curse of dimensionality. The test proposed in this paper avoids these difficulties by not requiring consistent density estimation at each boundary point.

Our manipulation test for BDDs builds on the $g$-order statistics framework introduced by \citet{bugni/canay:2021}. The procedure depends on an integer tuning parameter, $q$, the number of observations closest to the boundary used by the statistic. The parameter $q$ controls proximity to the boundary: as the sample size grows, the selected observations become increasingly close to it. Our main formal result shows that, under the joint density continuity null, their side indicators behave asymptotically as independent Bernoulli random variables with success probability $1/2$, independently of their projected locations on $\mathcal B$. This result forms the basis of the test: we project the selected observations onto the boundary and partition the projected locations into clusters using $k$-means with several values of $ k$. Under the joint density continuity null, the number of observations in the treatment region should behave like a binomial random variable, both overall and within each cluster. Under manipulation, some clusters may depart from this behavior. By forming clusters of observations with nearby projected locations, the $k$-means step localizes the boundary and makes such departures easier to detect, especially when manipulation is localized and varies smoothly along $\mathcal B$.

The main contribution of the paper is to establish that the proposed density-free manipulation test controls size asymptotically for fixed $q$ in BDDs with general boundaries. To prove this result, we bring tools from geometric measure theory to the analysis of manipulation testing. These tools guide the regularity conditions imposed on the boundary and allow us to characterize the behavior of observations selected by distance and then projected onto $\mathcal B$. This characterization yields the random-labeling property on which the test is based and may be useful beyond the specific procedure considered here. We also use Monte Carlo simulations to study the test's finite-sample behavior, including its ability to detect localized and offsetting manipulation.

The tuning parameter $q$ specifies how many observations closest to the boundary are selected for our test. In our asymptotic framework, we treat $q$ as fixed as the sample size $n$ grows. Developing a fully data-dependent choice of $q$ is beyond the scope of this paper, although \citet{bugni2026rates} may provide useful guidance for future work. In our empirical applications, we report results over a range of $q$ values to make their sensitivity transparent. The applications to \cite{DaiEtAl2022Hypertension} and \cite{keele/titiunik:2015} yield conclusions that are stable over a broad range of values of $q$, while the evidence against the null in \cite{elacqua/hincapie/martinez:2024} becomes stronger as $q$ increases.

The remainder of the paper is organized as follows. Section \ref{sec:lit_review} relates our contribution to the existing literature. Section \ref{sec:setting} introduces the BDD setup and presents our manipulation test. Section \ref{sec:results} establishes the theoretical properties of the test. Section \ref{sec:simulations} reports simulation evidence. Section \ref{sec:empirical} presents three empirical illustrations based on data from \cite{DaiEtAl2022Hypertension}, \cite{keele/titiunik:2015}, and \cite{elacqua/hincapie/martinez:2024}. Section \ref{sec:conclusions} concludes. Proofs of all results are collected in Appendix \ref{sec:proofs}.

\subsection{Related Literature}\label{sec:lit_review}

Our paper is related to the growing literature on BDDs. Existing work develops identification, estimation, and inference methods for multidimensional regression discontinuity designs, including both multiscore and geographic designs; see, among others, \cite{imbens/zajonc:2009}, \cite{papay/willett/murnane:2011}, \cite{reardon/robinson:2012}, \cite{cattaneo/idrobo/titiunik:2024}, and the recent survey by \cite{cattaneo2026boundary}. We focus on a different but complementary aspect of the design: implementing a \citet{mccrary:2008}-style manipulation test within a BDD. As already explained, the corresponding testable implication in this setting is that the joint density of the vector running variable is continuous across the boundary at every point $z\in\mathcal B$.

To the best of our knowledge, no existing manipulation test applies to BDDs with general boundaries. The closest paper is \citet{crippa:2025}, which studies manipulation testing in multiscore RDDs with hyper-rectangular assignment rules using density estimation in the treatment and control regions. Our procedure differs in both scope and construction. First, our test applies to general boundaries, including geographic boundaries with complex shapes, and nests multiscore RDDs with hyper-rectangular assignment rules as a special case. Second, our test relies on the asymptotic properties of observations closest to the boundary as the sample size grows, avoiding the need to consistently estimate the density functions.

Despite the lack, until recently, of clear theoretical guidance on manipulation testing in BDDs, empirical applications use several diagnostic strategies. In multiscore RDDs, researchers often apply componentwise manipulation tests. In geographic RDDs, manipulation tests are less commonly reported; when they are reported, researchers reduce the vector running variable to signed distance from the boundary and then apply a one-dimensional density test. Table \ref{tab:manipulation-tests} documents these practices across BDD applications, drawing on the surveys and reviews in \cite{cattaneo2026boundary}, \cite{lehner:2021}, and \cite{caicedo2021historical}. The table classifies studies according to whether they report no manipulation test, componentwise manipulation tests, or one-dimensional signed-distance tests.

Table \ref{tab:manipulation-tests} highlights two patterns. Manipulation tests are relatively common in multiscore applications and much less common in geographic applications, where boundaries often have complex shapes and no standard manipulation test for the joint-density-continuity null has been available. Among applications that report a test, the most common approaches are one-dimensional. First, some studies apply a manipulation test to each component of the vector running variable. We call this approach ``componentwise tests''. Second, other studies transform the vector running variable into its signed distance from the boundary and apply a one-dimensional manipulation test to the resulting scalar variable. We call this approach the ``signed-distance test''.

Both componentwise tests and signed-distance tests have important limitations. We begin with componentwise tests. This approach applies only to BDDs with hyper-rectangular assignment rules and does not readily extend to boundaries with more complex shapes, such as those arising in geographic RDDs. In addition, applying a separate test to each component of the vector running variable creates a multiple-testing problem: the componentwise tests must be combined to obtain a test with the desired overall size. Although standard methods exist, simple procedures such as the Bonferroni correction may reduce power. Moreover, this multiple-testing problem becomes more severe as the dimension of the vector running variable, and hence the number of componentwise tests, increases.

We next discuss signed-distance tests, which apply a one-dimensional manipulation test after reducing the vector running variable to its signed distance from the boundary. This approach has two key limitations, on which we expand in Section \ref{sec:signed_distance_test} of the appendix. First, continuity of the signed-distance density is only an implication of $H_0$, not an equivalent condition: discontinuities at different boundary points may average out, leaving the signed-distance density continuous. Appendix \ref{app:counterexample_averaging} provides an example in which $H_0$ fails at almost every boundary point, but a signed-distance test is powerless. Second, even the implication from $H_0$ to the continuity of the signed-distance density requires regularity conditions on the boundary. Appendix \ref{app:counterexamples} provides an example in which $H_0$ holds but the signed-distance density is discontinuous at zero, causing a signed-distance test to incorrectly reject $H_0$.

\begin{table}[htbp]
  \centering
  \scriptsize

  \begin{tabular}{|p{0.12\textwidth}|p{0.38\textwidth}|p{0.38\textwidth}|}
    
    \hline
    
    & Multiscore RD design & Geographic RD design \\
    
    \hline
    
    No Test &
    \cite{Ou2010ExitExam, snider2015barriers, frey:2019, DaiEtAl2022Hypertension, KampfenMosca2024BPScreening} & \cite{black:1999,kane2003school, kane2006school,bayer2007unified,lalive:2008, dell:2010, GroutJaegerPlantinga2011LandUsePortland,  EugsterEtAl2011SocialInsuranceCulture, basten2013beyond, FerwerdaMiller2014DevolutionResistance,michalopoulos2014national, turner2014land,BaroneDAcuntoNarciso2015Telecracy,keele/titiunik:2015,MacDonaldKlickGrunwald2016PrivatePolice, michalopoulos2016long,DeKadtLarreguy2018AgentsRegime,dell2018historical,ehrlich2018persistent,Kumar2018TexasHomeEquity,SpenkuchToniatti2018PoliticalAds,VelezNewman2019EthnicTV,becker2020forced,dell2020development,Ito2020CleanAir,LetsaWilfahrtMechanisms,MosconaNunnRobinson2020SegmentaryLineage, Schafer2020Time,Appau2021VietnamAgriculture, gonzalez2020cell,Lowes2021Concessions,michaels2021planning,Dehdari2022CommonIdentity,fontana2022historical,JonesEtAl2022AidLoss,Mangonnet2022ProtectedAreas,mendez2022multinationals,Sides2022TVAds,Zheng2022SchoolQuality,Baragwanath2023CollectiveRights,Henn2023TraditionalAuthorities,Paulsen2023NewDemocracy,prillaman2023strength, bjerre-nielsen_gandil_2024_attendance,cox_fiva_king_2024_bound_by_borders,doucette_2024_parliamentary_constraints,doucette_2024_pre_modern_institutions,grasse_2024_state_terror,jardim/long/plotnick/vigdor/wiles:2024,wuepper_et_al_2024_public_policies_forest,boix_2025_political_emancipation,loumeau_2025_regional_borders,mueller-crepon_2025_building_tribes, ring_2025_wealth_taxation, yamagishi_sato_2025_persistent_stigma} \\
    
    \hline

    Signed-distance Test  &
    \cite{clark2014signaling, cohodes2014merit, becht2016does, CastroEsposito2022Bonuses, MurphyJohnson2023DualIdentification, elacqua/hincapie/martinez:2024} &
    \cite{clinton_sances_2018_politics, dell2018nation, Dupraz2019ColonialLegaciesEducation, giuntella_mazzonna_2019_sunset, he_wang_zhang_2020_watering, ambrus2021loss, Albertus2020LandReform,Laliberte2021ContextualEffects,mcalexander_2023_borders, Woller2023CostOfVoting} \\
    
    \hline
    
    Componentwise Tests &
    \cite{matsudaira:2008, Robinson2011ELReclassification, hinnerich2014democracy, egger2015impact, elacqua2016short, evans2017smart, smith2017giving, Johnson2019ELClassification, londono2020upstream, JonesEtAl2022AidLoss, moussa2022impact, salti2022impact, LarsenValant2024GradeRetention} & \\
    
    \hline
    
  \end{tabular}

  \caption{\small Classification of BDD studies according to whether they report a manipulation test and, if so, the type of test used. ``Componentwise tests'' involve applying a one-dimensional manipulation test to each component of the vector running variable. ``Signed-distance test'' is a one-dimensional manipulation test applied to the signed distance from the boundary.}
  \label{tab:manipulation-tests}
\end{table}

Our testing approach is closest in spirit to \citet{bugni/canay:2021}. As in that paper, we avoid direct density estimation and instead exploit the behavior of observations closest to the cutoff. In a BDD, however, the cutoff is a boundary embedded in $\mathbb R^d$ rather than a single point. This distinction creates several new challenges: selected observations must be projected onto the boundary, local neighborhoods must be constructed along it, and its geometry must be incorporated into the asymptotic analysis.

We address these challenges using tools from geometric measure theory. These tools let us formulate appropriate regularity conditions on the boundary, characterize the behavior of observations selected by their distance from it, and establish the random-labeling approximation underlying our test. In this respect, our analysis is related to \cite{cattaneo2026boundary}, which studies estimation and inference in BDDs, and to work on submanifold and level-set integrals, including \cite{chen2026setsequallythinminimax} and \cite{qiao:2021}. Like our paper, this literature highlights the role of lower-dimensional geometry in econometric estimation and inference.

Our statistic also connects to the statistics literature on random labeling and spatial clustering. Under the joint-density-continuity null, the side indicators of the observations closest to the boundary behave asymptotically as independent Bernoulli labels, conditional on their projected locations. See, for example, \citet{LotwickSilverman1982} and \citet[Section~4.5]{Diggle2013} for related random-labeling problems. Existing methods for detecting departures from random labeling include scan statistics for localized excesses \citep{Kulldorff1997, AbolhassaniPrates2021} and nearest-neighbor methods for short-range segregation \citep{CuzickEdwards1990, Dixon2014}. We use $k$-means clustering to construct data-driven neighborhoods at several resolutions along the boundary and test for localized imbalances within them. This choice is motivated by the forms of manipulation we consider most relevant in economic applications: we expect manipulation to generate excess mass in localized neighborhoods of the boundary rather than across arbitrary, nonlocal subsets of boundary points.

Our focus is on density continuity. In empirical one-dimensional RDD applications, density-based manipulation tests are often complemented by balance tests that assess continuity of predetermined covariates at the cutoff; see, for example, \citet{lee/lemieux:2010}. Developing analogous balance tests for BDDs that assess covariate continuity locally along the boundary is beyond the scope of this paper. We are currently developing these methods in a companion paper.

\section{Our Test}\label{sec:setting}

This section presents our procedure for testing manipulation in BDDs. Section \ref{sec:setup} introduces the econometric framework and states the hypothesis testing problem. Section \ref{sec:our_test} describes the testing procedure.

\subsection{Setup} \label{sec:setup}

We consider a BDD in which assignment status is determined by a $d$-dimensional vector running variable. For unit $i$, let $Z_i\in\mathbb R^d$ denote the observed vector running variable, with generic random vector $Z$ and joint density $f_Z$.

By definition, a BDD assigns units to treatment and control regions according to a known rule $T:\mathbb R^d\to\{0,1\}$, where $T(z)=1$ if $z$ belongs to the treatment region and $T(z)=0$ if $z$ belongs to the control region. In a sharp BDD, $T(z)=1$ indicates the region in which treatment is assigned. In a fuzzy BDD, it indicates the region in which the treatment probability is discontinuously higher. Our procedure applies to both cases without modification.

We define the {\it assignment boundary} $\mathcal{B}$ as the common boundary of the treatment and control regions:
\begin{align} \label{eq:defnB}
    \mathcal{B} ~\equiv~ bd(\{z\in \mathbb{R}^{d}~:~T(z)=t\}) \qquad \text{for } t=0,1,
\end{align}
where $bd(S)$ denotes the topological boundary of a set $S$.\footnote{The topological boundary of $S$ is the intersection of the closure of $S$ and the closure of its complement. Recall that the closure of a set consists of all points in the set together with all of its limit points.}

Our goal in this paper is to develop a manipulation test for BDDs in the spirit of \cite{mccrary:2008}. Intuitively, in the absence of manipulation, the density of the vector running variable $Z$ should be continuous across the boundary at every boundary point $b\in\mathcal B$; see Appendix \ref{sec:leeExtension} for a formal justification. Throughout the paper, we use manipulation to mean a discontinuity in this density across the boundary, that is, a violation of the joint density continuity null. Our assumptions below require the density to have well-defined one-sided limits at every boundary point. Accordingly, we consider the following hypothesis testing problem:
\begin{equation} \label{eq:H0}
H_0: \lim_{z \to b, T(z) = 1} f_Z(z) = \lim_{z \to b, T(z) = 0} f_Z(z)~\text{ for all } b \in \mathcal{B}
~~~\text{vs.}~~~H_1: H_0 \text{ does not hold}.
\end{equation}
It is useful to note that the regularity conditions below require the density $f_Z$ to be well defined only near the boundary $\mathcal B$, and impose no restrictions on the distribution of $Z$ away from it.

\begin{remark}[Interpretation of the test]
As with other \cite{mccrary:2008}-type diagnostics, the manipulation test is a test for continuity of the joint density of $Z$ across the boundary. A rejection provides evidence of manipulation as defined in this paper, under the maintained assumptions. It does not, by itself, identify the mechanism generating the density discontinuity, nor does it test all assumptions required for identification in a BDD. Failure to reject should therefore not be read as proof that the design is valid, and rejection should not automatically be interpreted as evidence of manipulation in the behavioral sense.
\end{remark}

\begin{runningexample} 
We use the empirical application in \cite{DaiEtAl2022Hypertension} as a running example throughout the paper. We briefly introduce the setup here and leave further details to Section \ref{sec:daiapp}. The units are participants in the China Health and Nutrition Survey. Their systolic and diastolic blood pressure were measured and used to determine whether they had hypertension. Participants classified as hypertensive were then informed of the diagnosis, and the study examines whether this information changed their behavior.

The two blood pressure measures serve as the running variables. Let $Z_1$ denote systolic blood pressure and let $Z_2$ denote diastolic blood pressure, both measured in millimeters of mercury (mmHg). The diagnostic thresholds are $140$ and $90$, respectively. When at least one blood pressure measure is at or above its threshold, an individual is diagnosed with hypertension and therefore belongs to the treatment region:
\begin{equation*}
T(z_1,z_2) ~=~ \mathbf 1\{z_1 \geq 140\;\text{or}\; z_2 \geq 90\}.
\end{equation*}
Thus, the control region is the southwest portion of the score space, where both blood-pressure measures are below their thresholds, while the treatment region consists of points for which at least one threshold is reached.

The boundary $\mathcal B$ is the set of points at which a small change in one of the running variables can change assignment status. It is the union of two threshold rays:
\begin{equation*}
\mathcal B~=~\{(z_1,z_2): z_1=140,~ z_2\leq 90\}~\cup~\{(z_1,z_2): z_1\leq 140,~ z_2=90\}.
\end{equation*}
Figure \ref{fig:step0} shows the treatment region, the control region, and the boundary in this example.

\begin{figure}[htbp]
  \centering
  \paperfighalfalone{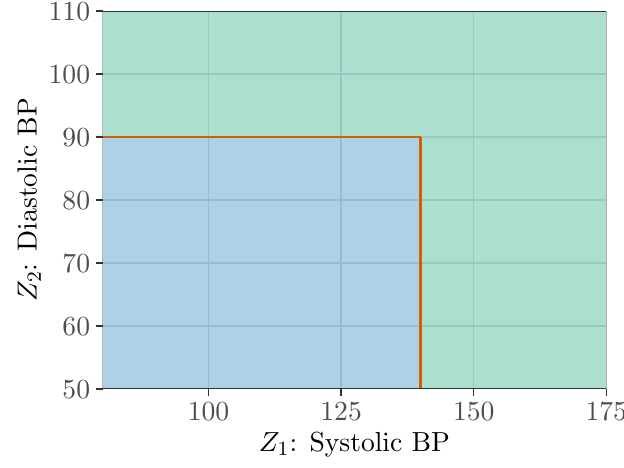}
  \caption{Regions induced by the assignment rule in the running example. The treatment region is shown in green, the control region in blue, and the boundary in orange.
  }
  \label{fig:step0}
\end{figure}
\end{runningexample}

\subsection{Testing procedure} \label{sec:our_test}

Let $\{Z_i\}_{i=1}^{n}$ denote a sample of $n$ observations of the vector running variable. Our analysis assumes that the observations are i.i.d.\ and that both the distribution of $Z_i$ near the boundary and the boundary $\mathcal B$ satisfy suitable regularity conditions. We defer the precise statement of these conditions to Section \ref{sec:results}.

Our test relies on a simple implication of the joint-density continuity null. If the density of $Z$ is continuous across the boundary, then observations sufficiently close to the boundary should be approximately balanced between the treatment and control regions, regardless of their projected location along the boundary. Thus, after selecting observations close to $\mathcal B$, the fraction of selected observations in the treatment region should be close to $1/2$ within localized regions of the boundary.\footnote{This intuition is simple, but its formal justification for general, piecewise-smooth boundaries is established under the geometric regularity conditions introduced in Section \ref{sec:results}.} Manipulation would instead create excess mass in one region relative to the other, producing localized regions in which this fraction differs from $1/2$.

Given a significance level $\alpha\in(0,1)$, implementation requires choosing $q\in\{1,\ldots,n\}$, the number of observations closest to the boundary used by the statistic. Our asymptotic framework treats $q$ as fixed as $n$ diverges. Given $(\alpha,q)$, the test proceeds in three steps.

\subsubsection*{Step 1: Selection}

This step selects the $q$ sample observations that are used for our test. For each observation $i=1,\dots,n$, we compute its distance to the boundary,
\begin{equation*}
    \dist(Z_i,\mathcal B)
    ~=~
    \inf_{b\in\mathcal B}\|Z_i-b\|.
\end{equation*}
Let $i_1<i_2<\cdots<i_q$ denote the indices of the $q$ observations closest to the boundary, listed in their original sample order. Any ties in determining the $q$ selected observations are resolved using a fixed deterministic rule. Set
\begin{equation*}
    \tilde Z_j ~\equiv~ Z_{i_j}, \qquad j=1,\dots,q.
\end{equation*}
We refer to $\{\tilde Z_j\}_{j=1}^{q}$ as the {\it selected observations}. Although selection is based on distance from the boundary, we label these observations in their original sample order. The remaining $n-q$ sample observations are not used in our test.

For each selected observation, we record its side indicator,
\begin{equation*}
    \tilde T_j ~\equiv~ T(\tilde Z_j), \qquad j=1,\dots,q,
\end{equation*}
and its projected location,
\begin{equation*}
    \tilde B_j ~\equiv~ B(\tilde Z_j), \qquad j=1,\dots,q,
\end{equation*}
where $B(z)$ denotes a closest point on the boundary, i.e.,
\begin{equation*}
    B(z)
    ~\in~
    \underset{b\in\mathcal B}{\arg\min}~\Vert z-b\Vert.
\end{equation*}
If the projection is not unique, we select one using a fixed measurable deterministic tie-breaking rule. Under our assumptions, the set of points with non-unique projections is negligible under the relevant local distribution, so the choice of tie-breaking rule is asymptotically immaterial.

\begin{runningexample}
Figure \ref{fig:step1} illustrates the selection step in the running example. The left panel shows a simulated sample of $n=100$ individuals together with the boundary. Green points lie in the treatment region, $T(Z_i)=1$, and blue points lie in the control region, $T(Z_i)=0$. The right panel highlights the $q=30$ observations closest to $\mathcal B$, colored by their side indicator, with projected locations shown as black squares.

\begin{figure}[htbp]
  \centering
    \begin{subfigure}[b]{\paperhalfwidth}
    \centering
    \paperfighalf{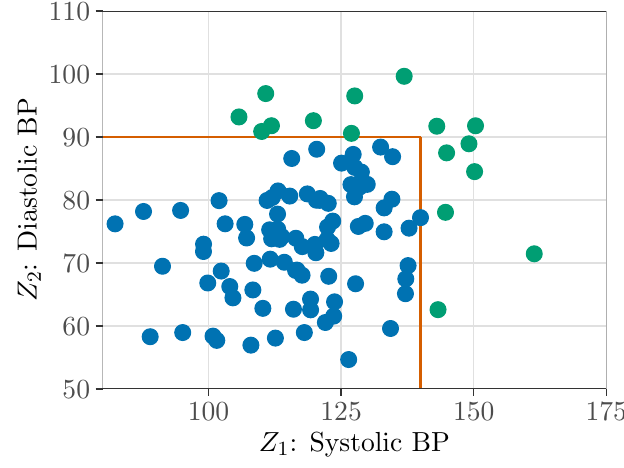}
  \end{subfigure}
  \hfill
  \begin{subfigure}[b]{\paperhalfwidth}
    \centering
    \paperfighalf{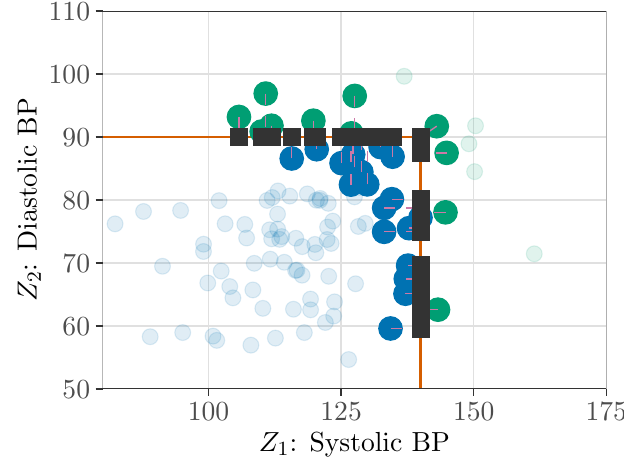}
      \end{subfigure}
  \caption{Step 1 in the running example. The left panel shows the full sample ($n=100$). The right panel highlights the $q=30$ observations closest to the boundary, colored by their side indicator, with projected locations shown as black squares.}
  \label{fig:step1}
\end{figure}
\end{runningexample}

\subsubsection*{Step 2: Clustering}\label{sec:clustering}

This step partitions the projected locations into clusters. Because each projected location corresponds to one selected observation, the resulting partition also groups selected observations whose projected locations are close together along the boundary. Under the joint density continuity null, the side indicators should be balanced not only overall but also within localized regions of $\mathcal B$. Clustering therefore allows the researcher to detect localized manipulation. Our focus on localized clusters is motivated by the forms of manipulation we consider most relevant in economic applications: we expect manipulation to generate excess mass in particular neighborhoods of the boundary rather than across arbitrary, nonlocal subsets of boundary points.

We partition the projected locations $\tilde B=\{\tilde {B_j}\}_{j=1}^{q}$ using the $k$-means algorithm. Specifically, for each $k\in\mathcal K=\{1,2,\dots,q\}$, we compute a $k$-cluster partition of the projected locations, resolving any non-uniqueness using a fixed measurable deterministic rule. For each selected observation $j=1,\dots,q$ and each $k\in\mathcal K$, let $\tilde\pi_{j,k}$ denote the cluster assignment of its projected location $\tilde B_j$. Thus, $\tilde\pi_{j,k}=a$ means that $\tilde B_j$ belongs to cluster $a\in\{1,\dots,k\}$; equivalently, selected observation $j$ inherits the cluster assignment of its projected location. For each $k$, $(\tilde\pi_{1,k},\tilde\pi_{2,k},\dots,\tilde\pi_{q,k})$ is the cluster assignment vector associated with the corresponding partition. Let $\tilde\pi={(\tilde\pi_{1,k},\tilde\pi_{2,k},\dots,\tilde\pi_{q,k})}_{k\in\mathcal K}$ denote the resulting collection of cluster assignment vectors.

The clusters can be interpreted as data-driven neighborhoods along the boundary. The extreme choices have simple interpretations: when $k=1$, all projected locations are assigned to a single cluster, while when $k=q$, each projected location forms its own cluster. Intermediate values of $k$ produce partitions at different resolutions. Smaller values of $k$ produce relatively coarse clusters along the boundary, while larger values of $k$ produce more localized clusters. The collection $\tilde{\pi}$ is random because it is computed from the projected locations $\tilde B$.

\begin{runningexample}
Figure \ref{fig:step2b} continues Figure \ref{fig:step1} and illustrates how $k$-means partitions the projected locations in the running example. Different values of $k$ produce different data-driven neighborhoods along the boundary. As explained earlier, these clusters are not intended to identify substantively meaningful categories of individuals; rather, they partition the projected locations at different resolutions.

\begin{figure}[htbp]
  \centering
   \paperfigfullwide{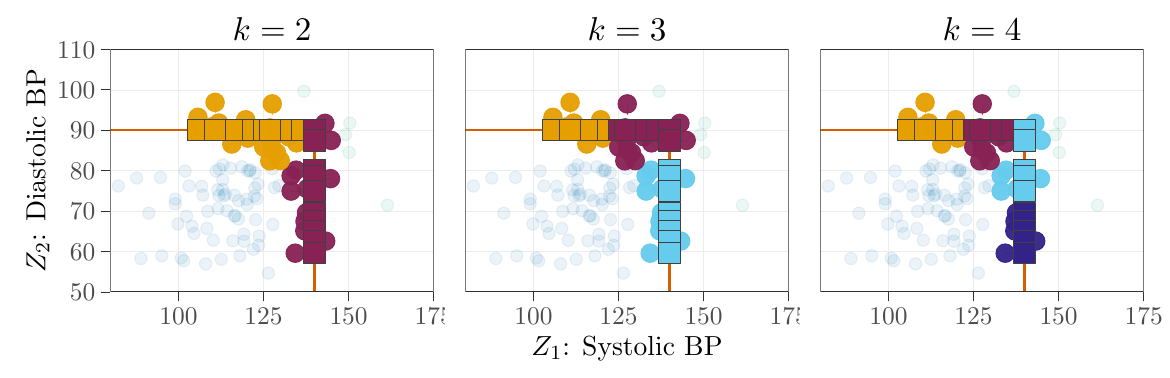}
  \caption{Step 2 in the running example with $n=100$, $q=30$, and $k \in \{2,3,4\}$}
  \label{fig:step2b}
\end{figure}
\end{runningexample}

\begin{remark}[Number of clusters] \label{remark:k_grid}
    The default implementation described in this section uses the full cluster-count grid $\mathcal{K} =\{1,2,\dots,q\}$. This choice is not essential. More generally, one may consider
    \begin{equation*}
        \mathcal{K}~\equiv~\{k_{1},\ldots,k_{J}\}~\subseteq~ \{1,2,\ldots,q\},
    \end{equation*}
    which denotes the cluster-count grid used in the $k$-means algorithm. The procedure and results extend immediately to this case, after replacing $\mathcal{K} = \{1,2,\ldots,q\}$ with $\mathcal{K}=\{k_{1},\ldots,k_{J}\}$. The formal arguments in the paper are derived for this latter, more general case.
\end{remark}

\subsubsection*{Step 3: Testing} 

Given $\tilde{T} = \{\tilde{T}_i\}_{i=1}^{q}$ from step 1 and $\tilde{\pi} = \{(\tilde{\pi}_{1,k},\tilde{\pi}_{2,k},\dots,\tilde{\pi}_{q,k})\}_{k \in \mathcal{K}}$ from step 2, compute the test statistic $S(\tilde{T}, \tilde{\pi})$ defined as follows:
\begin{align}
    S(\tilde{T}, \tilde{\pi}) ~=~ \max_{k \in \mathcal{K}} \max_{a=1,\ldots,k}
    \sqrt{ \sum_{i=1}^q \mathbf 1\{\tilde{\pi}_{i,k} = a\} }
    \left| \frac{ \sum_{i=1}^q \tilde{T}_i \mathbf 1\{\tilde{\pi}_{i,k} = a\} }{ \sum_{i=1}^q \mathbf 1\{\tilde{\pi}_{i,k} = a\} } - \frac{1}{2} \right| .
    \label{eq:test_function}
\end{align}
The last factor on the right-hand side of Equation \eqref{eq:test_function} computes the absolute difference between the observed fraction of observations in the treatment region within each cluster and $1/2$, scaling this absolute difference by the square root of the cluster size. This absolute difference measures the imbalance within the cluster. The test statistic $S(\tilde{T}, \tilde{\pi})$ is then the maximum standardized imbalance across all clusters and all partitions considered by the procedure.

The critical value is obtained by simulating the limiting null distribution of $S(\tilde{T},\tilde{\pi})$. According to our results, conditional on $\tilde\pi$, the side indicators are asymptotically governed by a conditional Bernoulli distribution. We approximate this distribution by repeatedly replacing the observed side indicators $\tilde{T} = \{\tilde{T}_i\}_{i=1}^{q}$ with an i.i.d.\ sample $T^{*} =\{T_i^{*}\}_{i=1}^{q}$ drawn from $\mathrm{Bernoulli}(1/2)$ and recomputing the statistic as $S(T^*,\tilde{\pi})$. Repeating this procedure approximates the conditional distribution $P(S(T^*,\tilde{\pi})\leq s\mid\tilde{\pi})$ to arbitrary accuracy. We define the critical value $c(\alpha,\tilde{\pi})$ as the corresponding $(1-\alpha)$ quantile:
\begin{align}
    c(\alpha,\tilde{\pi}) ~=~ \inf \left\{ ~s\in\mathbb R: ~P(S(T^*, \tilde{\pi})\leq s\mid \tilde{\pi}) \geq 1-\alpha~ \right\}.
        \label{eq:test_cv}
\end{align}

Our non-randomized test rejects $H_0$ in \eqref{eq:H0} when our statistic exceeds the critical value:
\begin{align}
    \phi_n^{\mathrm{nr}}(q,\alpha) ~=~ \mathbf 1 \{ S(\tilde{T}, \tilde{\pi})> c(\alpha,\tilde{\pi})\}.
    \label{eq:test_NR_version}
\end{align}
The non-randomized test may be asymptotically conservative due to the discreteness of the Bernoulli distribution. To obtain an asymptotically exact test, we propose a randomized test, which rejects $H_0$ in Equation \eqref{eq:H0} according to the following rule:
\begin{equation}
    \phi_n^{\mathrm{r}}(q,\alpha) ~=~\mathbf 1 \{  S(\tilde{T}, \tilde{\pi})> c(\alpha,\tilde{\pi}) \} + \mathbf 1 \{ S(\tilde{T}, \tilde{\pi}) = c(\alpha,\tilde{\pi}) \} \mathbf 1 \{  U \leq u(\alpha,\tilde{\pi}) \}, 
    \label{eq:test_R_version}
\end{equation}
where $U \sim \mathrm{Unif}(0,1)$ is independent of the data and
\begin{align}\label{eq:u_defn}
    u(\alpha,\tilde{\pi}) ~=~\Bigg\{ 
    \begin{array}{cl}
    \dfrac{ \alpha - P(S(T^*, \tilde{\pi})> c(\alpha,\tilde{\pi}) \mid \tilde{\pi}) }{P(S(T^*, \tilde{\pi})= c(\alpha,\tilde{\pi}) \mid \tilde{\pi}) } & \text{ if }P(S(T^*, \tilde{\pi})= c(\alpha,\tilde{\pi}) \mid \tilde{\pi})>0, \\
       0  & \text{ if }P(S(T^*, \tilde{\pi})= c(\alpha,\tilde{\pi}) \mid \tilde{\pi})=0. 
    \end{array} 
\end{align}

\begin{runningexample}
Figure \ref{fig:step3} illustrates the testing step in the running example. For each partition generated in Step 2, the statistic computes the share of selected observations in each cluster that lie in the treatment region. Under the joint density continuity null, observations sufficiently close to the boundary should be approximately balanced between the treatment and control regions. A cluster with a large excess of observations in either region contributes a large value to the statistic.

The histogram shows the simulated null distribution based on $10{,}000$ Bernoulli relabeling draws. In each draw, the observed side indicators $\tilde T$ are replaced by independent $\mathrm{Bernoulli}(1/2)$ labels, while the projected locations and collection of cluster-assignment vectors $\tilde\pi$ are held fixed. The solid vertical line marks the observed statistic $S(\tilde T,\tilde\pi)$, and the dashed vertical line marks the critical value $c(\alpha,\tilde\pi)$. The test rejects when the observed localized imbalance is too large relative to this conditional Bernoulli distribution.
\end{runningexample}

\begin{figure}[htbp]
  \centering
    \paperfighalfalone{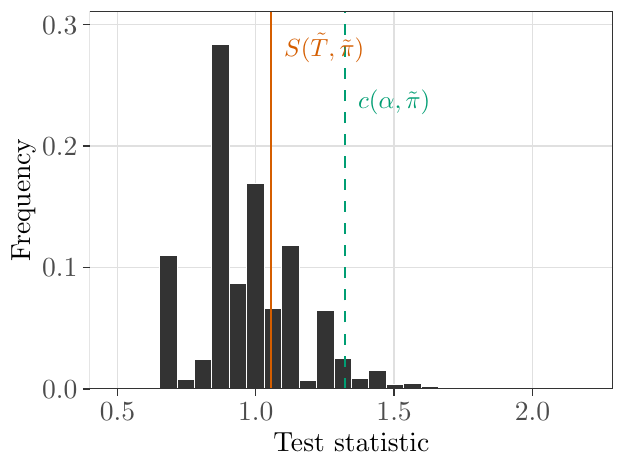}
 \caption{Illustration of step 3 with $n=100$, $q=30$, $\mathcal{K} = \{1,2,\ldots,q\}$, and $\alpha=0.05$. The histogram shows the simulated null distribution of $S(T^*,\tilde{\pi})$ based on $10{,}000$ Bernoulli relabeling draws, the solid vertical line marks the observed statistic $S(\tilde{T},\tilde{\pi})$, and the dashed vertical line marks the critical value $c(\alpha,\tilde{\pi})$.}  
  \label{fig:step3}
\end{figure}

\begin{remark}[Test statistic]
    The default implementation defines the test statistic $S(\tilde{T}, \tilde{\pi})$ in \eqref{eq:test_function} as the maximum standardized imbalance across all clusters and all partitions considered by the procedure. However, our results do not rely on this particular choice of statistic. They extend to any arbitrary function of $\tilde{T}$ and $\tilde{\pi}$ whose critical value is computed from the corresponding randomization distribution obtained by repeatedly replacing $\tilde{T}$ with the i.i.d.\ $\mathrm{Bernoulli}(1/2)$ sample $T^*$. Thus, our findings also apply to other aggregations of the same standardized imbalances. For example, a natural alternative is a Cram\'er--von Mises-type statistic, which sums the squared imbalances over $k \in \mathcal{K}$ and $a=1,\ldots,k$, weighted by the corresponding cluster size. That is,
    \begin{align}
        S^{\rm cvm}(\tilde{T}, \tilde{\pi}) ~=~ \sum_{k \in \mathcal{K}}\sum_{a=1}^{k}
      \left( { \sum_{i=1}^q \mathbf 1\{\tilde{\pi}_{i,k} = a\} }\right)
        \left( \frac{ \sum_{i=1}^q \tilde{T}_i \mathbf 1\{\tilde{\pi}_{i,k} = a\} }{ \sum_{i=1}^q \mathbf 1\{\tilde{\pi}_{i,k} = a\} } - \frac{1}{2} \right)^2 .
        \label{eq:test_function2}
    \end{align}
\end{remark}

\section{Assumptions and main results}\label{sec:results}

In this section, we establish the validity of the testing procedure described in Section \ref{sec:our_test}. Section \ref{sec:assumption} states the regularity conditions on the distribution of the running variables and on the geometry of the boundary. Section \ref{sec:validity} then derives the formal validity result. Finally, Section \ref{sec:discussion} discusses several features of the proposed test relative to alternative approaches.

\subsection{Assumptions}\label{sec:assumption}

We now state the assumptions used to establish the asymptotic validity of our test. The intuition behind the procedure is simple: under the joint density continuity null in \eqref{eq:H0}, observations sufficiently close to the boundary should be approximately balanced between the treatment and control regions. Formalizing this intuition for a general boundary, however, requires regularity conditions on both the density $f_Z$ and the boundary $\mathcal B$.

The assumptions on $f_Z$ are multidimensional analogs of those commonly used in manipulation tests with a scalar running variable. The assumptions on $\mathcal B$ are specific to the BDD setting. In a standard RDD, the boundary is a single threshold. In a BDD with a $d$-dimensional vector running variable, it may instead be a curve, a surface, or a more complicated object. We therefore impose geometric conditions that rule out pathological boundary behavior and ensure that local approximation and projection arguments are well-behaved. To state these restrictions precisely, we use basic concepts from geometric measure theory. Background references include \cite{federer:1996,evans/gariepy:1992,morgan:1998,maggi:2012}.

Before stating the assumptions, we introduce some notation. We use $\mathcal L$ to denote $d$-dimensional Lebesgue measure, the usual notion of volume in $\mathbb R^d$. For $m\in\mathbb N$, we use $\mathcal H^m$ to denote $m$-dimensional Hausdorff measure, which provides an $m$-dimensional notion of size for subsets of $\mathbb R^d$. When $m=d$, $\mathcal H^m$ coincides with $\mathcal L$. When $m<d$, it distinguishes among lower-dimensional sets that all have zero $d$-dimensional Lebesgue measure. For example, when $d=2$, line segments have zero two-dimensional Lebesgue measure, but their lengths are measured by $\mathcal H^1$. We also use $C^r$ to denote $r$-times continuously differentiable regularity, and $\Vert \cdot \Vert_{C^r}$ for the corresponding $C^r$ norm, as defined in \citet[p.~12]{hirsch2012differential} and \citet[p.~35]{hirsch2012differential}, respectively. Appendix~\ref{sec:GMTstuff} provides the formal definition of Hausdorff measure and a brief discussion of $C^r$ manifolds and norms. For additional background on Hausdorff measure, see \citet[Chapter 2]{federer:1996} and \citet[Chapter 3]{maggi:2012}.

With this notation in place, Assumption \ref{ass:assumption} states the conditions under which our testing procedure is asymptotically valid.

\begin{assumption}\label{ass:assumption}
Let $\{Z_i\}_{i=1}^n$ be an i.i.d.\ sample of $d$-dimensional random vectors, where $d \geq 2$. Let the boundary $\mathcal{B}\subset \mathbb{R}^d$ be closed and nonempty.
\begin{enumerate}[{\rm (a)}]
\item For some $\delta>0$, the distribution of $Z_i$ restricted to
\begin{equation*}
    \mathcal{B}^{\delta}
    \equiv
    \big\{z\in\mathbb{R}^{d}:\inf_{a\in\mathcal{B}}\Vert z-a\Vert\leq\delta\big\}
\end{equation*}
is absolutely continuous with respect to Lebesgue measure and admits a version of its density $f_Z$ such that, for every $b\in\mathcal B$ and $t\in\{0,1\}$, $\lim_{z\to b,\,T(z)=t}f_Z(z)$ exists. Moreover, for some $b_0\in\mathcal B$ and $t_0\in\{0,1\}$,
\begin{equation*}
    0
    ~<~
    \lim_{z\to b_0,\,T(z)=t_0}f_Z(z)
    ~\leq~
    \sup_{z\in\mathcal B^\delta}f_Z(z)
    ~<~
    \infty.
\end{equation*}

\item We have $\mathcal{H}^{d-1}(\mathcal{B})>0$ and, for every $r>0$,
\begin{equation*}
    \mathcal{H}^{d-1}\big(\{b \in \mathcal{B} : \Vert b\Vert \leq r\}\big) ~<~ \infty .
\end{equation*}

\item There exists a constant $\gamma>0$ and $\bar r>0$ such that, for all $t \in \{0,1\}$, all $b \in \mathcal{B}$, and all $r\in(0,\bar r)$,
\begin{equation*}
    \frac{
    \mathcal{L}(\{z \in \mathbb{R}^d : \Vert z-b\Vert \leq r\} \cap \{T(z)=t\})
    }{
    \mathcal{L}(\{z \in \mathbb{R}^d : \Vert z-b\Vert \leq r\})
    }
    ~\geq~ \gamma .
\end{equation*}

\item There exists a closed set $\Sigma\subseteq\mathcal B$ such that $\mathcal H^{d-1}(\Sigma)=0$ and $\mathcal B\setminus\Sigma$ is an embedded $C^{\max\{2,d-1\}}$, $(d-1)$-dimensional manifold without boundary.

\item We have
\begin{align*}
    \lim_{r\rightarrow \infty }~\underset{u\downarrow 0}{\lim \sup}~
    P\big(\{ \dist (Z_{i},\mathcal{B})  \leq u\}\cap \{\Vert Z_{i} \Vert>r\}\big)/u~=~0,\\
    \lim_{r\to\infty}
\int_{{b\in\mathcal B:\Vert b\Vert>r}} f_Z(b)~d\mathcal H^{d-1}(b)
~=~0 .
\end{align*}
\end{enumerate}
\end{assumption}

Assumption \ref{ass:assumption} has five components, which we discuss in detail next. 

Assumption \ref{ass:assumption}(a) requires the distribution of $Z_i$ to be absolutely continuous in a fixed neighborhood of the boundary. Its density $f_Z$ must be bounded above throughout this neighborhood and strictly positive at least one boundary point. Throughout the paper, $f_Z$ denotes this density on $\mathcal B^\delta$. In particular, this is the density used to formulate the null hypothesis in \eqref{eq:H0}. Under $H_0$, for every $b\in\mathcal B$, $f_Z(b)$ is the common limiting value of the density as $b$ is approached from the treatment and control regions. Positivity at some $b_0\in\mathcal B$ rules out the degenerate case in which the density vanishes along the entire boundary, while boundedness provides the domination needed for our local probability calculations. Both restrictions are local to the boundary. Outside $\mathcal B^\delta$, the distribution of $Z_i$ is unrestricted and need not admit a density.

Assumption \ref{ass:assumption}(b) requires $\mathcal B$ to have positive and locally finite $(d-1)$-dimensional size. The condition $\mathcal H^{d-1}(\mathcal B)>0$ ensures that the boundary is not negligible according to the measure appropriate for a hypersurface in $\mathbb R^d$. The second condition requires the portion of $\mathcal B$ contained in any bounded ball to have finite $(d-1)$-dimensional Hausdorff measure. When $d=2$, $\mathcal H^1$ measures length, so this condition requires the boundary to have finite length inside every bounded ball. When $d=3$, $\mathcal H^2$ measures surface area, so it requires finite surface area inside every bounded ball.

Assumption \ref{ass:assumption}(c) is a uniform thickness condition on the treatment and control regions. It requires that, below a common radius $\bar r$, each region occupy at least a fixed fraction $\gamma$ of the Lebesgue volume of every ball centered at a boundary point. Thus, neither the treatment region nor the control region can become arbitrarily thin relative to the surrounding neighborhood, uniformly over $b\in\mathcal B$. This condition rules out, for example, cusp-like configurations in which one region's volume share converges to zero near some boundary points.

Assumption \ref{ass:assumption}(d) requires $\mathcal B$ to be a regular hypersurface except on a negligible closed set $\Sigma$. More precisely, after removing a set with zero $\mathcal H^{d-1}$-measure, the remaining set $\mathcal B\setminus\Sigma$ must be an embedded $C^{\max\{2,d-1\}}$, $(d-1)$-dimensional manifold without boundary. The term ``without boundary'' refers to the boundary of $\mathcal B\setminus\Sigma$ as a manifold. It does not preclude $\mathcal B$ from having corners, vertices, edges, or endpoints, provided that these features are included in $\Sigma$ and thus have zero $\mathcal H^{d-1}$-measure. The $C^2$ regularity is used for the local tubular representation, while the $C^{d-1}$ regularity permits the application of the finite-differentiability version of the Morse--Sard theorem.

Assumption \ref{ass:assumption}(e) imposes two tail restrictions. The first requires the probability of observing $Z_i$ both close to $\mathcal B$ and far from the origin to become negligible, after applying the local normalization by $u$, as the radius of the ball tends to infinity. The second requires the density-weighted Hausdorff measure of the portion of $\mathcal B$ far from the origin to vanish in the same limit. Thus, neither the local probability mass near $\mathcal B$ nor its limiting boundary integral can be driven by increasingly remote portions of the boundary. Both restrictions hold automatically if the support of $Z_i$ is bounded or if $\mathcal B$ is bounded.

We conclude this section by returning to the running example based on \cite{DaiEtAl2022Hypertension} and verifying that Assumption \ref{ass:assumption} holds in that setting under mild regularity conditions.

\begin{runningexample}
We now verify Assumption \ref{ass:assumption} in our running example. Since the running variable is a vector with two components, $d=2$. The boundary $\mathcal B$ is L-shaped, as depicted in Figure \ref{fig:step0}. It is the union of the systolic and diastolic blood-pressure threshold rays, which meet at the corner $(140,90)$.

Part (a) requires the joint distribution of systolic and diastolic blood pressure to admit a well-behaved density in a neighborhood of the boundary, with a density that is bounded above in that neighborhood and positive at some boundary point. Away from the boundary, the distribution of the vector running variable is unrestricted and need not admit a density.

Because $d=2$, the relevant boundary measure is $\mathcal H^1$, which measures length. The L-shaped boundary has positive and finite length inside any bounded ball, so part (b) holds. Part (c) also holds. Away from the corner, every sufficiently small ball centered on the boundary is divided into treatment and control regions of equal Lebesgue volume. At the corner, the control region occupies one quadrant and the treatment region occupies the remaining three. Therefore, the lower bound in part (c) holds with any $\gamma\leq1/4$.

To verify part (d), set $\Sigma:=\{(140,90)\}$. The set $\Sigma$ is closed and satisfies $\mathcal H^1(\Sigma)=0$. Moreover, $\mathcal B\setminus\Sigma$ is the disjoint union of two open rays. Each ray is a smooth embedded one-dimensional manifold without boundary, and their disjoint union has the same property. Since $\max\{2,d-1\}=2$, it follows that $\mathcal B\setminus\Sigma$ is, in particular, an embedded $C^2$, one-dimensional manifold without boundary. Thus, part (d) allows the corner because it belongs to a set that is negligible under the relevant boundary measure.

Finally, part (e) holds if the density of the vector running variable decays sufficiently fast near remote portions of the boundary. In particular, it is immediate if the support of the vector running variable is bounded. More generally, they hold under sufficiently fast decay of its density along the two threshold rays. Thus, under mild distributional assumptions, the running example satisfies Assumption \ref{ass:assumption}.
\end{runningexample}

\subsection{Properties of our test} \label{sec:validity}

This section describes the asymptotic properties of the hypothesis testing procedure described in Section \ref{sec:our_test}. Throughout the asymptotic analysis, we fix the number $q$ of selected observations as $n\to\infty$. Thus, the section studies the limiting experiment generated by the $q$ observations closest to the boundary.

Our main result in this section is Theorem \ref{thm:validity}, which establishes the asymptotic validity of our proposed hypothesis test. Before stating this theorem, we present several intermediate results that provide intuition. The intermediate results follow the three steps in our procedure. Theorem \ref{thm:keyPart1} describes the asymptotic behavior of the selected sample and establishes the random-labeling property of the side indicators relative to their projected locations. Theorem \ref{thm:clustering} shows that this random-labeling property is preserved after the clustering step. Together, these results justify the conditional randomization approximation used to construct the critical values in Theorem \ref{thm:validity}.

\subsubsection*{Step 1: Selection}

Step 1 selects the $q$ observations closest to the boundary $\mathcal B$, which are labeled in their original sample order. Recall that, for each selected observation $\tilde Z_i$, we define its projected location and side indicator as $\tilde B_i=B(\tilde Z_i)$ and $\tilde T_i=T(\tilde Z_i)$, respectively. The first result describes the joint limiting distribution of the selected observations, their projected locations, and their side indicators. Under this labeling, the limiting observations form an i.i.d.\ sample whose common distribution is derived from the boundary-local law.

\begin{theorem}[Convergence and limiting distribution]\label{thm:keyPart1}
Under Assumption \ref{ass:assumption} and $H_0$ in \eqref{eq:H0},
\begin{equation*}
\{(\tilde Z_i,\tilde B_i,\tilde T_i)\}_{i=1}^q ~\overset{d}{\to}~ \{(Z_i^*,B_i^*,T_i^*)\}_{i=1}^q,
\end{equation*}
where $\{(Z_i^*,B_i^*,T_i^*)\}_{i=1}^{q}$ is an i.i.d.\ sample. Moreover, 
\begin{enumerate}[(a)]
    \item $T_i^*$ is independent of $(Z_i^*,B_i^*)$,
    \item $T_i^* \sim \mathrm{Bernoulli}(1/2)$,
    \item $Z_i^*=B_i^*$ almost surely.
\end{enumerate}
\end{theorem}

Theorem \ref{thm:keyPart1} describes the asymptotic distribution of the selected sample when the selected observations are labeled in their original sample order. It shows that, under $H_0$ and the regularity conditions in Assumption \ref{ass:assumption}, the $q$ observations closest to the boundary behave asymptotically as an i.i.d.\ sample from the boundary-local limiting distribution. This reduces the original order-statistic selection problem to a simpler boundary-local experiment. The weak convergence result is the multidimensional analog of \citet[Lemma B.3]{bugni/canay:2021}, which in turn builds on \citet[Theorem 1]{kauffman/reiss:1992}.

The result also characterizes the limiting distribution that underlies all subsequent results. The limiting side indicators are i.i.d.\ Bernoulli$(1/2)$ and are independent of the limiting locations $(Z^*,B^*)$. Moreover, $Z_i^*=B_i^*$ almost surely, so, in the limit, the original location carries no additional information once its projected location is known. Thus, asymptotically, the relevant information in the selected sample is contained in the projected locations $\tilde B$ and the side indicators $\tilde T$.

The intuition is that, under $H_0$, observations sufficiently close to the boundary are approximately balanced between the treatment and control regions. This balance holds locally at regular boundary points: although the density level may vary along the boundary, the local share of observations on each side is asymptotically one-half. Hence, the projected location along the boundary carries no asymptotic information about the side indicator. This is the key simplification used below: after projecting onto the boundary, the selected side indicators can be treated, asymptotically, as independent Bernoulli$(1/2)$ draws, which forms the basis of the testing procedure.

\subsubsection*{Step 2: Clustering}

Step 2 of our test clusters the projected locations $\tilde B$ using $k$-means. The next result derives the limiting distribution of the side indicators $\tilde T$ conditional on the collection of cluster-assignment vectors $\tilde\pi$.

\begin{theorem}[Side indicators conditional on clustering]\label{thm:clustering}
Under Assumption \ref{ass:assumption} and $H_0$ in \eqref{eq:H0}, for every
$t\in\{0,1\}^q$,
\begin{equation*}
    P(\tilde T=t\mid \tilde\pi)
    ~\overset{p}{\to}~
    2^{-q}.
\end{equation*}
\end{theorem}

Theorem \ref{thm:clustering} implies that $\tilde T$ and $\tilde\pi$ are asymptotically independent, with $\tilde T$ converging to an i.i.d.\ $\mathrm{Bernoulli}(1/2)$ vector. The limiting distribution of $\tilde T$ is already established in Theorem \ref{thm:keyPart1}. Thus, the novelty of Theorem \ref{thm:clustering} is the independence between the selected side indicators and the collection of cluster-assignment vectors. In the limit, the information used to form the clusters carries no information about the side indicators. This result justifies approximating the null distribution of the test statistic $S(\tilde T,\tilde\pi)$ by redrawing side indicators from an i.i.d.\ $\mathrm{Bernoulli}(1/2)$ distribution conditional on the realized collection of cluster-assignment vectors.

\subsubsection*{Step 3: Testing}

The following theorem is the paper's main result. It establishes the asymptotic validity of the randomized and non-randomized tests defined in Section \ref{sec:our_test}.

\begin{theorem}[Asymptotic validity]\label{thm:validity}
Fix $q\in\mathbb N$ and $\alpha\in(0,1)$. Under Assumption \ref{ass:assumption} and $H_0$ in \eqref{eq:H0},
\begin{equation*}
    \lim_{n\to\infty} E\left[\phi_n^{\mathrm r}(q,\alpha)\right] ~=~ \alpha,
\end{equation*}
and
\begin{equation*}
    \limsup_{n\to\infty} E\left[\phi_n^{\mathrm{nr}}(q,\alpha)\right] ~\le~ \alpha.
\end{equation*}
\end{theorem}

Theorem \ref{thm:validity} shows that the randomized test is asymptotically exact, while the non-randomized test is asymptotically valid. In general, the non-randomized test may be asymptotically conservative because the test statistic is discrete, reflecting the underlying $\mathrm{Bernoulli}(1/2)$ distribution of the side indicators.

The rationale behind Theorem \ref{thm:validity} is Theorem \ref{thm:clustering}. That theorem shows that the selected side indicators $\tilde T$ are asymptotically distributed as an i.i.d.\ $\mathrm{Bernoulli}(1/2)$ sample and are asymptotically independent of the collection of cluster-assignment vectors $\tilde \pi$. This finding formally justifies an asymptotic approximation to the distribution of $S(\tilde{T},\tilde{\pi})$ obtained by replacing $\tilde{T}$ with an i.i.d.\ $\mathrm{Bernoulli}(1/2)$ sample while holding $\tilde{\pi}$ fixed. This approximation is easy to compute via simulation, leading to the critical value $c(\alpha,\tilde{\pi})$.

\subsection{Discussion}\label{sec:discussion}

This section discusses how the validity result in Theorem \ref{thm:validity} should be interpreted and implemented. We focus on two implementation choices: the number of selected observations $q$ and the grid of cluster sizes $\mathcal{K}$. We then compare the proposed clustered procedure with two alternatives that either avoid or substantially modify the clustering step.

We first discuss the role of $q$. Theorem \ref{thm:validity} establishes size control under $H_0$ for fixed $q$ as the sample size diverges. The parameter $q$ determines the number of observations closest to the boundary that enter the statistic and therefore controls proximity to the boundary. Smaller values of $q$ use observations closer to the boundary, making the random-labeling approximation more accurate, but they also leave less information for detecting manipulation. Larger values of $q$ use more observations, but at the cost of relying on observations farther from the boundary, which may generate size distortions. The choice of $q$ therefore reflects a tradeoff between the accuracy of size control and power.

Given that $q$ is the number of observations used by the statistic, treating it as fixed means that the effective sample size of the test does not grow with $n$. As a result, our test is not consistent against fixed alternatives. We therefore study its power through simulations. We leave extending the analysis to data-dependent or diverging values of $q$ for future work; \citet{bugni2026rates} may serve as a useful starting point. For this reason, in applications, we report the test's sensitivity over a range of values of $q$.

Next, consider the choice of the cluster-count grid $\mathcal K$ in Step 2. Recall that the default implementation uses $\mathcal K=\{1,2,\dots,q\}$, but this choice is not essential: Theorem \ref{thm:validity} applies to any subset $\mathcal K\subseteq \{1,2,\dots,q\}$. Conditional on $q$ and on the chosen distance, the full grid $\mathcal K=\{1,2,\dots,q\}$ has the advantage of avoiding the choice of a single clustering resolution and of searching thoroughly over possible localized imbalances. At the same time, it can be computationally demanding, especially in repeated calculations such as Monte Carlo simulations. In practice, one may instead use a smaller grid $\mathcal K=\{k_1,\ldots,k_J\}$ that includes both small and large values of $k$. Small values of $k$ detect widespread manipulation, while larger values detect more localized manipulation.

The clustering operation in Step 2 is central to the proposed test because it preserves information about where selected observations lie along the boundary. It also adds a technical complication: the $k$-means clusters are random because they are computed from the projected locations. Theorem \ref{thm:validity} justifies the resulting randomization procedure by showing that, under the joint density continuity null, the side indicators of the selected observations behave asymptotically as independent Bernoulli labels, conditional on the projected locations and hence on the clusters formed from them. Given this complication, it is natural to ask whether one can construct useful tests that avoid the random clustering step. We now discuss two such alternatives.

The first alternative is the one-cluster version of our test. Formally, this corresponds to replacing the default cluster-count grid $\mathcal K=\{1,2,\dots,q\}$ with $\mathcal K=\{1\}$, so that all selected observations form a single cluster. This version applies the logic of \cite{bugni/canay:2021} to the signed distance from the boundary, where distance is measured from $\mathcal B$, and the sign is determined by whether the observation lies in the treatment or control region. Under our assumptions, the resulting test is asymptotically valid by Theorem \ref{thm:validity}. This validity, however, does not follow mechanically from the one-dimensional result in \cite{bugni/canay:2021}. The null hypothesis in \eqref{eq:H0} concerns the joint density along a general boundary, and geometric regularity is needed to ensure that its continuity implies continuity of the signed distance density at zero. 
Appendix \ref{app:counterexamples} shows that, without such regularity, the signed-distance density can be discontinuous even when $H_0$ in \eqref{eq:H0} holds, causing the test to incorrectly reject. Even when the required regularity conditions hold, the test may have limited power. Appendix \ref{app:counterexample_averaging} shows that discontinuities at different boundary points may average out, leaving the signed-distance density continuous even when $H_0$ in \eqref{eq:H0} fails. Thus, although the one-cluster version is asymptotically valid under our assumptions, it may be powerless against offsetting boundary manipulation.

Beyond this validity issue, the one-cluster version can be less informative than the proposed clustered test. Because it aggregates all selected observations into a single cluster, it is designed to detect widespread imbalances between the treatment and control regions. It can therefore be insensitive to manipulation that is localized along $\mathcal B$. It can also miss offsetting manipulation, in which excess mass in the treatment region near one side of the boundary is offset by excess mass in the control region near another side. The clustering step is designed precisely to avoid this information loss. By clustering projected locations that are close to one another along the boundary, the proposed statistic can detect localized imbalances that would cancel out or be diluted in the aggregate. The Monte Carlo simulations below illustrate this tradeoff: the one-cluster version performs well for widespread manipulation, while clustering is important for localized and offsetting manipulation.

A second alternative would be to consider all clusters induced by all possible partitions of the projected locations into $k$ clusters, for each $k\in\mathcal K$. Such a procedure would include the clusters generated by $k$-means, as well as every cluster generated by every other admissible partition. Relative to the proposed test, this alternative has the advantage that the collection of candidate clusters is not tied to a particular clustering algorithm. However, it has two important drawbacks. First, the number of candidate clusters is extremely large, making the resulting procedure computationally unattractive when $q$ is large or when $\mathcal K$ contains many values. Second, searching over such a large collection can reduce power because the critical value must account for the maximum over many clusters, including many that do not correspond to localized neighborhoods of the boundary. The $k$-means step provides a practical compromise: it searches for imbalances at several resolutions while restricting attention to clusters of projected locations that are nearby under the chosen distance. We do not claim that $k$-means is optimal; rather, it provides a simple and computationally convenient way to construct such clusters.\footnote{A procedure that searches over all possible partitions may have power against manipulation occurring over nonlocal subsets of boundary points. We, however, consider such manipulation less relevant in many economic applications. We expect manipulation to generate excess mass in localized neighborhoods of the boundary rather than across arbitrary subsets of boundary points.}

Taken together, these choices determine how the proposed procedure translates the joint density continuity null into a finite-sample random-labeling problem. The parameter $q$ determines the selected local sample, while the cluster-count grid $\mathcal K$ determines the resolution at which the test searches for imbalances along the boundary. Our procedure provides a practical and flexible way to implement this idea in BDDs with general boundaries. In particular, the clusters formed by $k$-means provide a natural collection of neighborhoods when manipulation is expected to arise in localized regions of $\mathcal B$.

\section{Monte Carlo simulations}\label{sec:simulations}

This section studies the finite-sample behavior of our proposed manipulation test. The simulations are designed to evaluate the size of the test under $H_0$ in \eqref{eq:H0}, its power against different forms of departures from $H_0$, and the role played by the number $q$ of selected observations.

The simulation designs use the L-shaped boundary from the running example, but we recenter the two cutoffs at $(0,0)$ for simplicity. That is, we consider a two-dimensional running variable,
\begin{equation*}
Z_i~=~(Z_{i1},Z_{i2})~\in~\mathbb{R}^2,
\end{equation*}
and assign treatment according to
\begin{equation*}
T(z_1,z_2) ~=~ \mathbf 1\{z_1 \geq 0\;\text{or}\; z_2 \geq 0\}.
\end{equation*}
The boundary is therefore the $L$-shaped set
\begin{equation*}
\mathcal B~=~\{(z_1,z_2): z_1=0,~ z_2\leq 0\}~\cup~\{(z_1,z_2): z_1\leq 0,~ z_2=0\}.
\end{equation*}
As explained in Section \ref{sec:assumption}, this boundary consists of two smooth rays that meet at a corner.

\subsection{Simulation designs}\label{sec:sim_designs}

We generate data from four designs. The baseline component of the running variable is
\begin{equation*}
Z_i\sim \mathcal{N}(\mu,\Sigma), \qquad \mu=(-1,-1), \qquad \Sigma=\mathrm{diag}(9,9).
\end{equation*}
Design 1 draws all $n$ observations from this baseline distribution, which satisfies $H_0$ in \eqref{eq:H0}. In Designs 2--4, we draw $n-m$ observations from the baseline distribution and $m$ observations from a contamination distribution, generating departures from $H_0$ in \eqref{eq:H0}.

We choose $\mu=(-1,-1)$ so that the baseline distribution is centered inside the control region, as in the running example after recentering the cutoffs. Centering the distribution at the corner $(0,0)$ would introduce additional symmetry of the density along the boundary. At the same time, the covariance matrix $\Sigma=\mathrm{diag}(9,9)$ is sufficiently dispersed to generate a substantial number of observations in both the treatment and control regions.

\paragraph{Design 1: $H_0$.}
All observations are generated i.i.d.\ from the baseline distribution:
\begin{equation*}
Z_i\sim \mathcal{N}(\mu,\Sigma).
\end{equation*}
Figure \ref{fig:designs}(a) depicts a typical draw with $n=500$. This density is continuous at the boundary, so $H_0$ in \eqref{eq:H0} holds. We use this design to evaluate the finite-sample size of the proposed test.

\paragraph{Design 2: Localized violation.}
In this design, $n-m$ observations are drawn i.i.d.\ from the baseline distribution, and the remaining $m$ observations are drawn i.i.d.\ from a small rectangle near one ray of the boundary:
\begin{equation*}
Z_i\sim \mathrm{Unif}([-2.5,-2]\times[-0.5,0]), \qquad i=1,\ldots,m.
\end{equation*}
Figure \ref{fig:designs}(b) illustrates the support of the contamination distribution. This design represents a situation in which manipulation occurs near a localized portion of the boundary. When $m$ is small relative to $n$, this type of manipulation may be difficult to detect with a test that only uses distance to the boundary, because the manipulation signal is averaged with observations projected onto unaffected portions of the boundary. The clustering step in our procedure improves power against localized manipulation by allowing the test to focus on smaller portions of the boundary.

\paragraph{Design 3: Widespread violation.}
In this design, $n-m$ observations are drawn i.i.d.\ from the baseline distribution, and the $m$ contaminated observations are drawn independently from two strips near the boundary. Half of the contaminated observations are drawn from a strip near the horizontal ray,
\begin{equation*}
Z_i\sim \mathrm{Unif}([-3,0]\times[-0.5,0]), \qquad i=1,\ldots,m/2,
\end{equation*}
and half are drawn from a strip near the vertical ray,
\begin{equation*}
Z_i\sim \mathrm{Unif}([-0.5,0]\times[-3,0]), \qquad i=m/2+1,\ldots,m.
\end{equation*}
Figure \ref{fig:designs}(c) depicts the support of the contamination distributions. This design represents a situation in which manipulation occurs along widespread portions of both boundary rays. Because the manipulation signal has the same direction along both rays, a test that pools the selected observations should detect it. Our test should also have power against this alternative. However, in this design, clustering may be less beneficial because it can split a common manipulation signal across smaller clusters rather than aggregating it over the full selected sample.

\paragraph{Design 4: Offsetting violation.}
In this design, $n-m$ observations are drawn i.i.d.\ from the baseline distribution, and the $m$ contaminated observations are drawn independently from two strips near the boundary. Half of the contaminated observations are control units near the horizontal ray,
\begin{equation*}
Z_i\sim \mathrm{Unif}([-3,0]\times[-0.5,0]), \qquad i=1,\ldots,m/2,
\end{equation*}
while the remaining half are treated units near the vertical ray,
\begin{equation*}
Z_i\sim \mathrm{Unif}([0,0.5]\times[-3,0]), \qquad i=m/2+1,\ldots,m.
\end{equation*}
Figure \ref{fig:designs}(d) illustrates the support of the contamination distributions. This design generates opposite-sign manipulation on different portions of the boundary. As a result, a test that pools all selected observations into a single cluster should have no power, because the excess mass in the treatment region near one ray is offset by the excess mass in the control region near another ray. In contrast, the clustering step in our procedure is designed to detect this type of manipulation because it can detect imbalances localized to different parts of the boundary.

\begin{figure}[htbp]
\centering
\paperfigfull{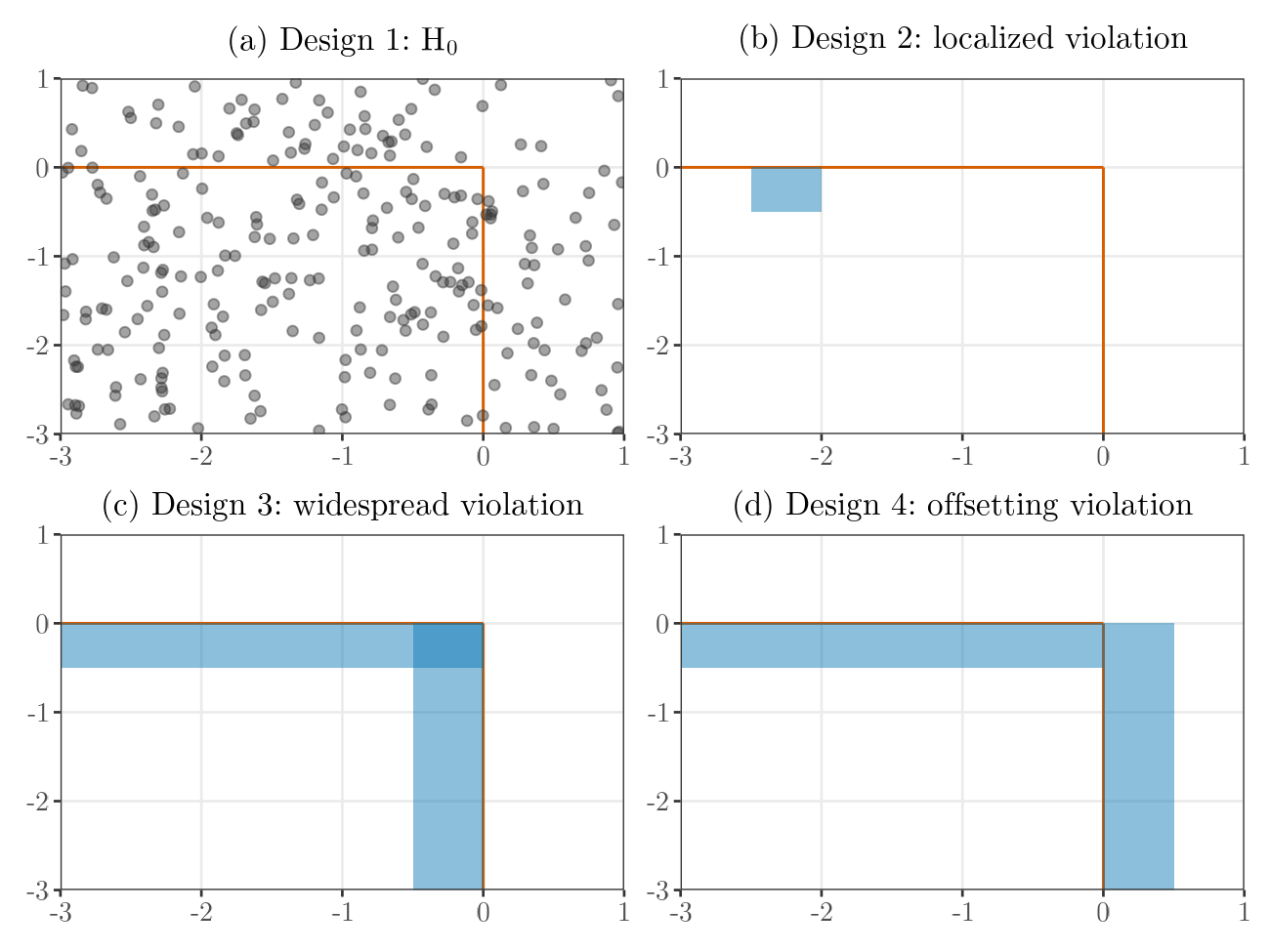}
\caption{\small Simulation designs. Panel (a) depicts the null design. Panels (b)--(d) show the locations of manipulated units under localized, widespread, and offsetting manipulation, respectively.}
\label{fig:designs}
\end{figure}

\subsection{Implementation}\label{sec:sim_implementation}

We set $\alpha=10\%$, $n\in\{500,~1{,}000,~1{,}500,~2{,}000,~2{,}500,~3{,}000,~3{,}500,~4{,}000,~4{,}500,~5{,}000\}$, and $m =0.04n$. We choose this value of $m$ so that departures from $H_0$ are noticeable in the simulations while remaining a small fraction of the full sample. For each Monte Carlo replication and each design, we draw a sample, select the $q\in\{50,~100,~150,~200\}$ observations closest to the boundary, project them onto $\mathcal B$, and compute the statistic $S(\tilde{T},\tilde{\pi})$ described in Section \ref{sec:our_test}.

For each value of $q$, we implement our preferred version of the test with $\mathcal K=\{1,2,\ldots,q\}$. To illustrate the role of clustering, we also consider the one-cluster version of our test, which uses $\mathcal K=\{1\}$. Comparing the preferred implementation with the one-cluster version allows us to assess the role of clustering in detecting localized and offsetting manipulation.

For each test evaluation, we compute the critical value in \eqref{eq:test_cv} using $R=5{,}000$ independent Bernoulli relabeling draws from the conditional Bernoulli distribution. We then compute empirical rejection rates from $MC=2{,}000$ independent Monte Carlo replications for each combination of design, sample size, value of $q$, and cluster-count grid.

\subsection{Results}\label{sec:sim_results}

Figure \ref{fig:mc-bdd-mt} reports empirical rejection frequencies for the non-randomized version of the test in \eqref{eq:test_NR_version}. Panel (a) corresponds to the null design, while Panels (b)--(d) correspond to the localized, widespread, and offsetting alternatives. Figure \ref{fig:mc-bdd-mt_r} shows the analogous results for the randomized version of the test in \eqref{eq:test_R_version}. The randomized and non-randomized versions are virtually indistinguishable, so we focus our discussion on the non-randomized results.

\begin{figure}[htbp]
  \centering
    \paperfigfullnarrow{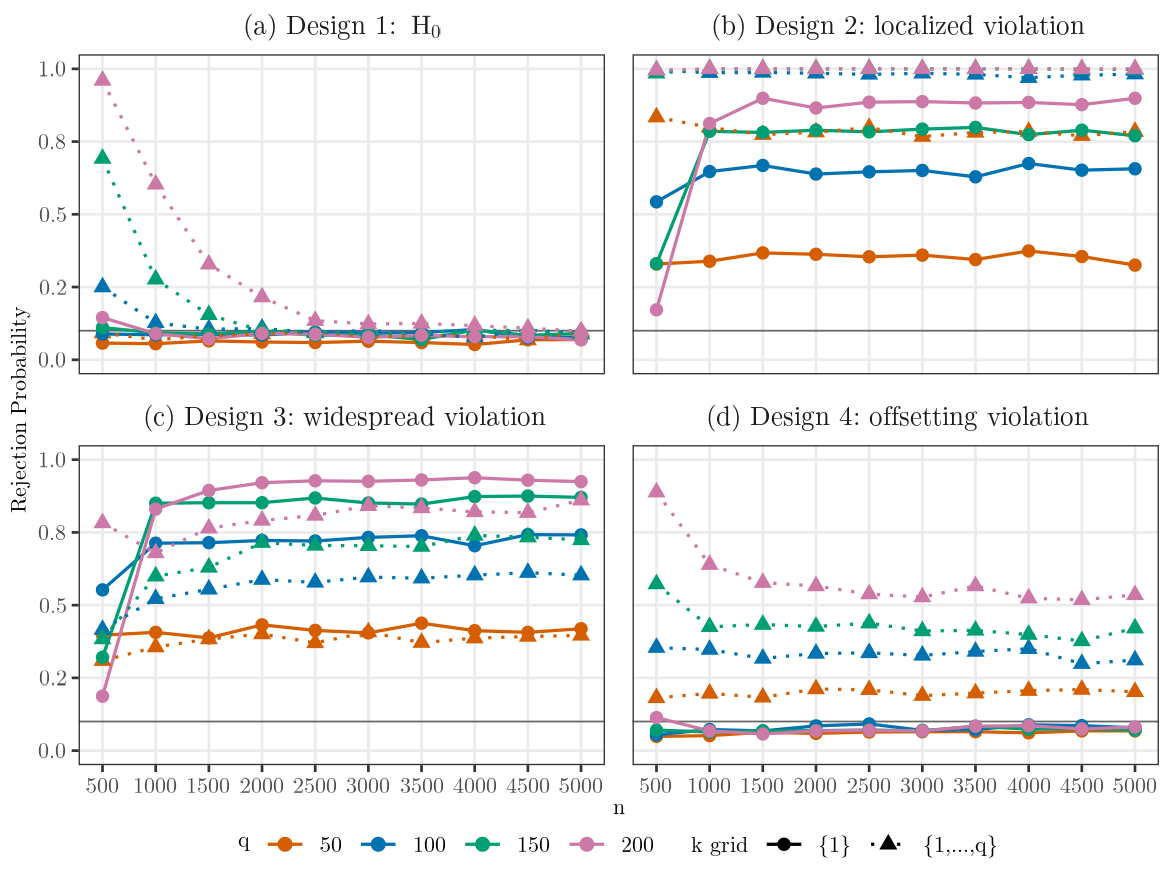}
  \caption{\small Empirical rejection probabilities for the non-randomized version of the test at nominal $\alpha=10\%$. The horizontal axis is sample size $n$; colors denote $q$; line type distinguishes $\mathcal K=\{1,2,\dots,q\}$ and $\mathcal{K}= \{1\}$.}
  \label{fig:mc-bdd-mt}
\end{figure}

\begin{figure}[htbp]
  \centering
    \paperfigfullnarrow{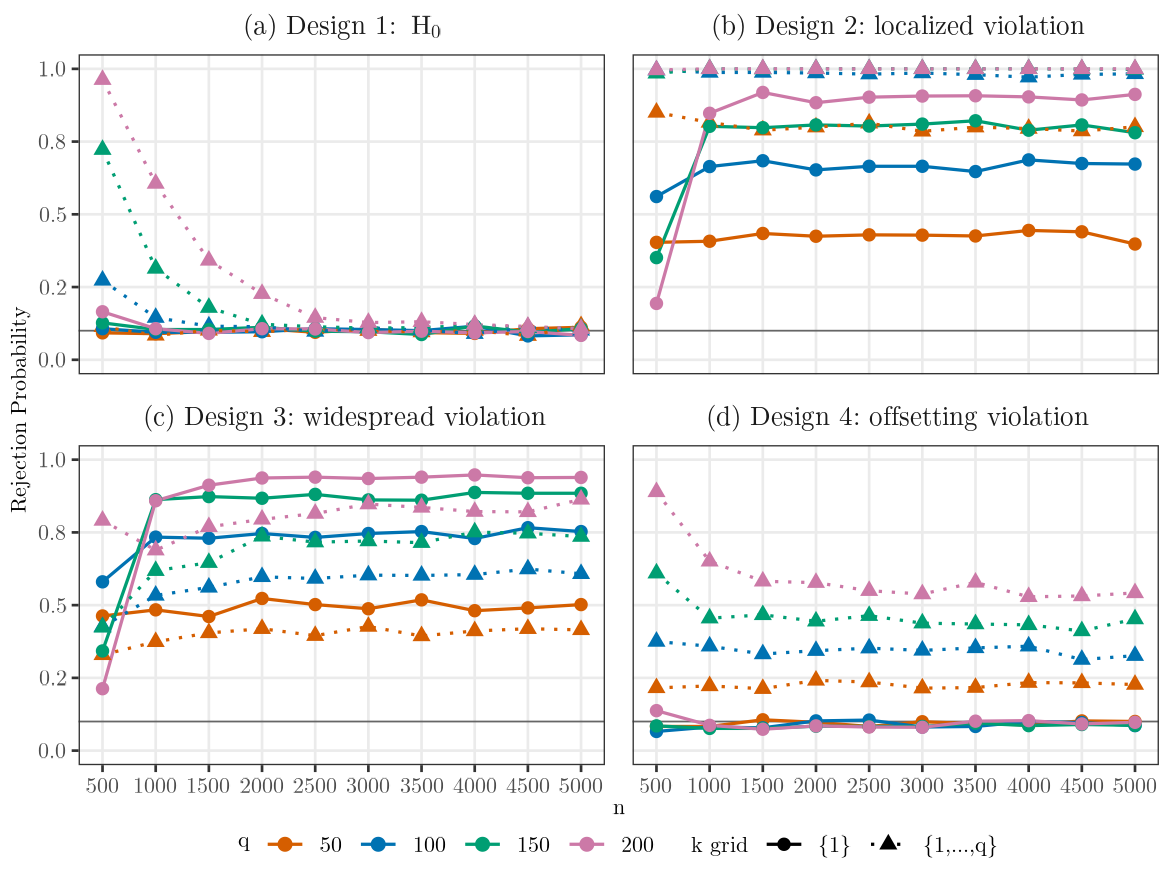}
  \caption{\small Empirical rejection probabilities for the randomized version of the test at nominal $\alpha=10\%$. The horizontal axis is sample size $n$; colors denote $q$; line type distinguishes $\mathcal K=\{1,2,\dots,q\}$ and $\mathcal{K}= \{1\}$.}
  \label{fig:mc-bdd-mt_r}
\end{figure}

We begin by describing the results under $H_0$ in panel (a). Under the joint density continuity null, the rejection frequencies approach the nominal level as the sample size increases. This pattern is consistent with the fixed-$q$ asymptotic approximation in Theorem \ref{thm:validity}. For each fixed value of $q$, as $n$ grows, the $q$ closest observations lie increasingly close to the boundary $\mathcal B$. As a result, the Bernoulli approximation for observations close to the boundary becomes more accurate, and the finite-sample size distortions disappear. Once the sample size is sufficiently large relative to $q$, size control is very good across all values of $q$ considered.

Panel (a) also shows that the asymptotic approximation breaks down when $q$ represents a large fraction of the sample size $n$. This is especially visible when $q=200$ and $n=500$, so that our test is implemented with $q/n=40\%$ of the sample. In this case, the selected observations are no longer sufficiently close to the boundary, and smooth variation in the baseline density can generate finite-sample imbalances along different portions of the boundary. The one-cluster version with $\mathcal K=\{1\}$ averages these imbalances over the entire selected sample, while our preferred implementation with $\mathcal K=\{1,2,\ldots,q\}$ can isolate the portions where they are strongest, such as near $(-1,0)$, $(0,-1)$, and $(0,0)$. Thus, the size distortions reflect a finite-sample failure of the approximation based on proximity to the boundary when $q$ is too large relative to $n$, rather than a failure of the test in the setting covered by the asymptotic theory.

Panel (b) considers localized manipulation. In this design, our proposed test exhibits very high power for all sample sizes. This highlights the benefit of the richer cluster-count grid $\mathcal K=\{1,2,\ldots,q\}$, which lets the test focus on partitions that isolate the localized manipulation. In contrast, the one-cluster version has relatively lower power. This is because, when all selected observations are pooled together, the excess mass generated by the contamination is averaged with observations projected onto unaffected portions of the boundary.

Panel (c) reports results for widespread manipulation, where the manipulation signal has the same direction along both rays of the boundary. Our test exhibits reasonably high rejection rates, while the one-cluster version performs slightly better. This is expected, because pooling observations into a single cluster reinforces the manipulation signal. Together with panel (b), these results illustrate the main power tradeoff introduced by clustering: more localized clusters improve sensitivity to localized manipulation, but they may sacrifice power against widespread manipulation aligned across the boundary.

Finally, panel (d) considers offsetting manipulation, where the manipulation signal has opposite signs along the two rays of the boundary. In this case, our preferred implementation exhibits good rejection rates, while the one-cluster version has no power. This is again expected. Our preferred implementation benefits from clustering because some partitions can separate the two boundary rays and detect the imbalance within each ray. In contrast, the one-cluster version pools the two rays together, effectively canceling the manipulation signals. This design illustrates the main benefit of clustering when manipulation is heterogeneous across the boundary.

We can summarize these findings as follows. First, under $H_0$, provided that the sample size is sufficiently large relative to $q$, the rejection rates are close to the nominal size and relatively stable across values of $q$. This is consistent with our asymptotic results. Second, the Monte Carlo simulations illustrate both the benefits and the costs of the clustering step in the proposed test. The benefit is that clustering makes the test sensitive to localized and offsetting manipulation along the boundary. The cost is that when the manipulation signal is common across the boundary in both magnitude and sign, clustering can dilute it by splitting it into smaller clusters. Given that our test achieves nontrivial rejection rates across all alternative designs, we view the benefits of clustering as outweighing its costs in these simulations. This supports using the fine cluster-count grid $\mathcal K=\{1,2,\ldots,q\}$ as our preferred implementation.

\section{Empirical applications}\label{sec:empirical}

We illustrate our proposed test with three applications that span several settings in which BDDs arise in practice. \citet{DaiEtAl2022Hypertension} studies a two-score design with an L-shaped boundary induced by two diagnostic thresholds, which motivated the running example in the previous sections. \citet{keele/titiunik:2015} considers a geographic RDD, where assignment status is determined by location relative to a known border. \citet{elacqua/hincapie/martinez:2024} analyzes an RDD in which assignment status depends on seven running variables and a nontrivial institutional rule. In each application, we test $H_0$ in \eqref{eq:H0}, namely, that the joint density of the vector running variable is continuous across the boundary at every boundary point.

We use the same implementation across the three applications. We focus on the non-randomized version of the test in Section~\ref{sec:our_test}, using our preferred cluster-count grid $\mathcal K=\{1,2,\ldots,q\}$. The sample sizes are $n=33{,}495$ for \citet{DaiEtAl2022Hypertension}, $n=7{,}523$ for \citet{keele/titiunik:2015}, and $n=2{,}039$ for \citet{elacqua/hincapie/martinez:2024}. For each integer $q\in\{20,21,\ldots,200\}$, we compute the simulated $p$-value by approximating the conditional Bernoulli distribution using $R=20{,}000$ independent Bernoulli relabeling draws. We start the grid at $q=20$ because, for very small values of $q$, the test is based on too few observations to be informative. We stop the grid at $q=200$ because, given the sample sizes in these applications, larger values of $q$ select observations that are relatively far from the boundary, making the approximation based on proximity to the boundary less appropriate. The purpose of presenting the $p$-value path is purely descriptive: it reports the results for all values of $q$ considered and makes their sensitivity to the choice of $q$ transparent. The paper does not claim simultaneous validity over $q$, and the path itself does not constitute an additional test.

\subsection{\cite{DaiEtAl2022Hypertension}} \label{sec:daiapp}

\citet{DaiEtAl2022Hypertension} studies whether being informed of a hypertension diagnosis leads individuals to change subsequent health behavior, such as smoking and dietary fat intake. The analysis uses data from the China Health and Nutrition Survey (CHNS), a longitudinal household survey that collects detailed health and nutrition measures for individuals in China.

The empirical design exploits the clinical definition of hypertension: an individual is classified as hypertensive when systolic blood pressure is at least 140 mmHg or diastolic blood pressure is at least 90 mmHg. Survey examiners measure blood pressure and inform participants of their results. Because classification depends on two thresholds, the design naturally forms a two-dimensional BDD with an L-shaped boundary. We construct the sample from the 1997--2006 survey waves, retaining respondents with three valid readings of each blood pressure measure who report no previous hypertension diagnosis. We use the average of the three systolic readings and the average of the three diastolic readings as the running variables and measure distances in raw mmHg.

Figure~\ref{fig:maps_closest_d} displays the observations in the treatment and control regions together with the diagnostic boundary. Individuals in the treatment region are shown in green and individuals in the control region in blue. The L-shaped boundary is defined by the systolic and diastolic blood-pressure thresholds. Panel (a) displays all the data, while panels (b) and (c) display the selected observations and the projected locations for $q=100$.

\begin{figure}[htbp]
\centering
\paperfigfullmap{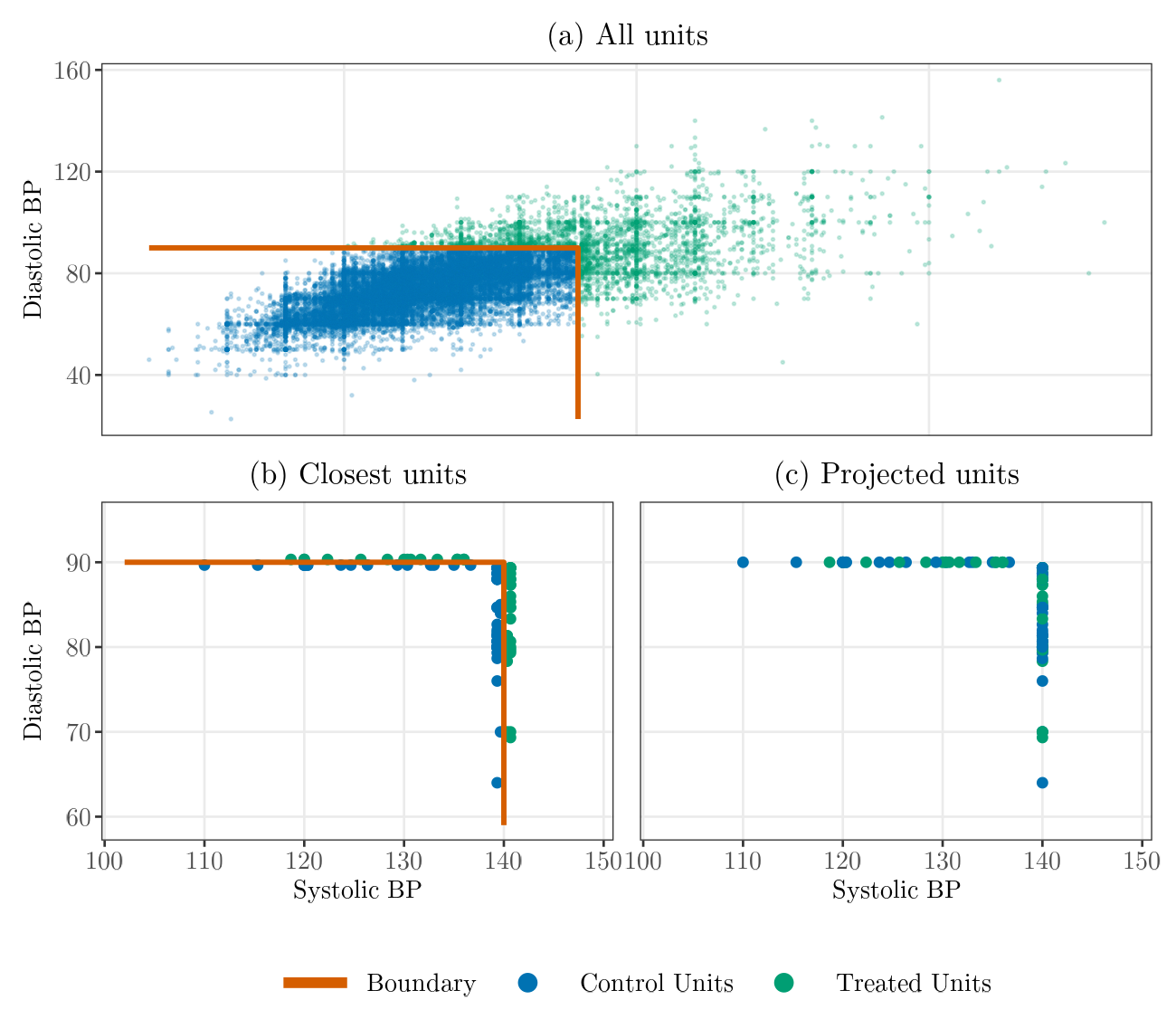}
\caption{Data from \cite{DaiEtAl2022Hypertension}. Individuals in the treatment region are shown in green, individuals in the control region in blue, and the boundary in orange. Panel (a) shows all observations. Panels (b) and (c) focus on the selected $q=100$ observations closest to the boundary: panel (b) displays these observations in the space of the running variables, while panel (c) displays the projected locations.}
\label{fig:maps_closest_d}
\end{figure}

Figure \ref{fig:dai} reports the empirical results of our test for this application. Panel (a) shows that the simulated $p$-values remain above conventional significance levels throughout the range of values of $q$ considered. The smallest $p$-value is approximately $0.29$, so the test does not reject $H_0$ for any $q\in\{20,21,\ldots,200\}$. Panel (b) shows the same stable pattern: the observed test statistic remains below the simulated $95\%$ critical value over the full grid.
Thus, the application provides no evidence against continuity of the joint density of systolic and diastolic blood pressure across the diagnostic boundary. This conclusion is stable as the selected sample expands to include observations farther from the boundary.

\begin{figure}[htbp]
\centering
\begin{subfigure}{\paperhalfwidth}
    \centering
    \paperfighalf{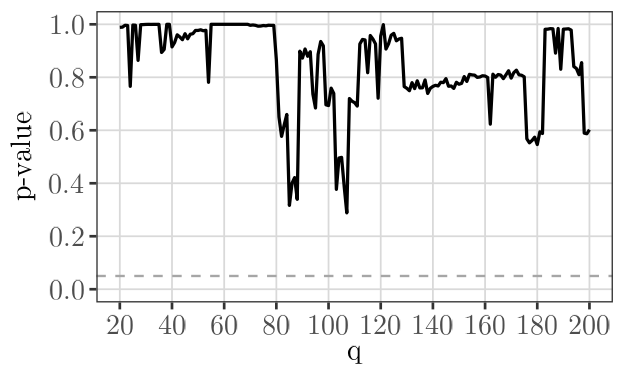}
    \caption{$p$-value as a function of $q$.}
    \label{fig:dai-pvalue}
\end{subfigure}
\hfill
\begin{subfigure}{\paperhalfwidth}
    \centering
    \paperfighalf{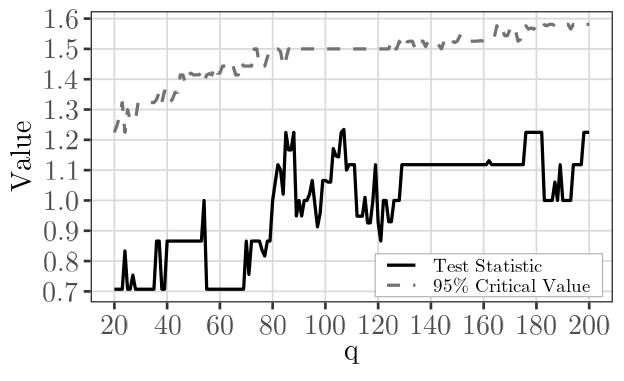}
    \caption{Test statistic and $95\%$ CV as a function of $q$.}
    \label{fig:dai-statistic}
\end{subfigure}
\caption{Results of our test applied to the data in \citet{DaiEtAl2022Hypertension}. Panel (a) reports the simulated $p$-value for each value of $q$. Panel (b) reports the corresponding value of the test statistic and $95\%$ critical value.}
\label{fig:dai}
\end{figure}

\subsection{\cite{keele/titiunik:2015}} \label{sec:geo_application}

\cite{keele/titiunik:2015} studies the effect of presidential campaign advertising on voter turnout using a geographic regression discontinuity design. The source of discontinuous treatment is the media-market boundary separating the Philadelphia and New York City television markets within the West Windsor-Plainsboro school district. In the 2008 election, voters in the Philadelphia media market were exposed to a large volume of presidential advertising, whereas voters in the New York City media market were not.

The running variables are geographic coordinates, and assignment status changes discretely at the media-market border. Individual-level voter-file data allow us to locate each voter relative to that boundary. In this setting, the joint density continuity null requires that the spatial density of voters be continuous across the border at every boundary point. A discontinuity would indicate that the voter distribution changes abruptly across the boundary, potentially undermining comparisons between units located immediately across it.

Figure \ref{fig:map_units_kt} displays the geographic location of observations in the treatment and control regions together with the boundary. Panel (a) displays all observations, while panels (b) and (c) display, respectively, the selected observations and their projected locations for $q=100$. The projected locations display the type of localized imbalance in side indicators that the proposed test is designed to detect.

\begin{figure}[htbp]
\centering
\paperfigfullmap{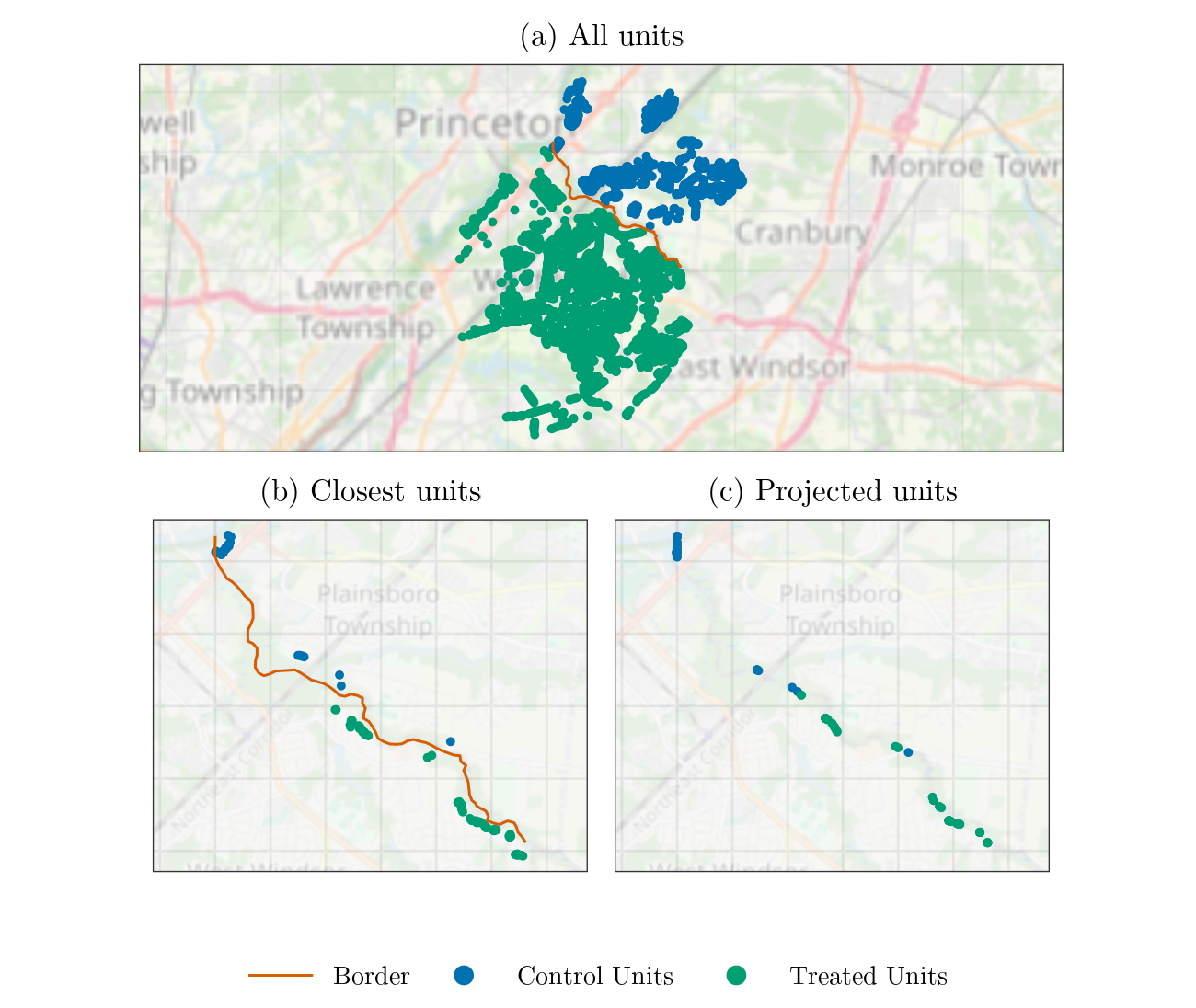}
\caption{Data from \cite{keele/titiunik:2015}. Observations in the treatment region are shown in green, observations in the control region in blue, and the boundary in orange. Panel (a) shows all observations. Panels (b) and (c) focus on the selected $q=100$ observations closest to the boundary: panel (b) displays these observations in the space of the running variables, while panel (c) displays the projected locations.}
\label{fig:map_units_kt}
\end{figure}

Figure \ref{fig:keele-titiunik} reports the empirical results of our test for this application. Panel (a) shows that the simulated $p$-values are very close to zero over most of the grid, so our test rejects $H_0$ throughout the range of values of $q$ considered. Panel (b) shows that the observed test statistic is well above the simulated $95\%$ critical value for all values of $q$, with the gap between the two curves widening over most of the range considered. Thus, the evidence against $H_0$ is not driven by a narrow choice of the tuning parameter. Instead, the results indicate a persistent imbalance in the side indicators of observations near the geographic boundary, providing strong evidence against $H_0$ in this application.

\begin{figure}[htbp]
\centering
\begin{subfigure}{\paperhalfwidth}
    \centering
    \paperfighalf{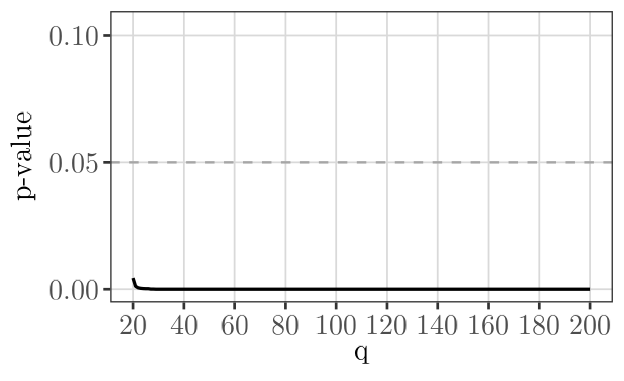}
    \caption{$p$-value as a function of $q$.}
\end{subfigure}
\hfill
\begin{subfigure}{\paperhalfwidth}
    \centering
    \paperfighalf{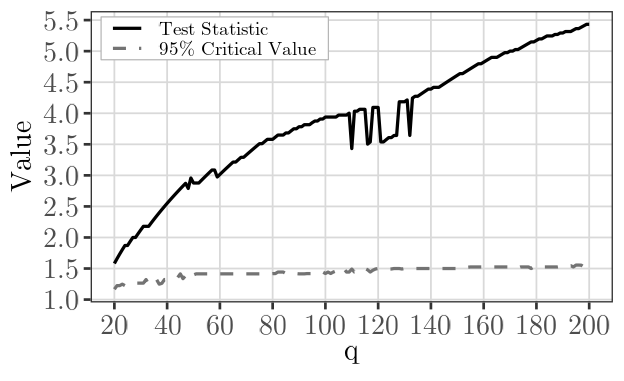}
    \caption{Test statistic and $95\%$ CV as a function of $q$.}
\end{subfigure}
\caption{Results of our test applied to the data in \cite{keele/titiunik:2015}. Panel (a) reports the simulated $p$-value for each value of $q$. Panel (b) reports the corresponding value of the test statistic and $95\%$ critical value.}
\label{fig:keele-titiunik}
\end{figure}

\subsection{\cite{elacqua/hincapie/martinez:2024}} \label{sec:elacquaapp}

\citet{elacqua/hincapie/martinez:2024} examines the effects of a school accountability policy on teacher mobility in Chile. The policy targets low-performing schools, and the paper uses the discontinuity induced by the assignment rule to estimate the policy effect. The main finding is that teachers in low-performing schools are more likely to leave, and movers' destinations differ across public and private schools.

The assignment mechanism is substantially more complex than in the previous two applications. It depends on SIMCE (\textit{Sistema de Medici\'{o}n de la Calidad de la Educaci\'{o}n}) scores and related school-quality measures. In a given year, a school is classified as low performing if it satisfies the participation requirements and, in at least two of the previous three evaluations, has both an average fourth-grade SIMCE below 220 and less than 20\% of students scoring above 250 on the average test score. Independently of these criteria, we also assign schools below the tenth percentile of the school quality index to treatment. The resulting boundary depends on $d=7$ running variables: three lagged SIMCE averages, three lagged shares of students scoring above 250, and the current school quality index. Before computing distances and projected locations, each running variable is standardized as $(Z_j-c_j)/\widehat{\sigma}_j$, where $c_j$ is the relevant institutional cutoff and $\widehat{\sigma}_j$ is the sample standard deviation\footnote{Following the original paper, we standardize the running variables using standard deviations computed from the full sample. This introduces dependence across standardized observations and therefore does not strictly preserve the i.i.d. structure assumed in our theoretical analysis. As a robustness check, we use sample splitting: we estimate the standard deviations from a subsample of 100 observations and use them to standardize the remaining observations, which constitute the sample used to implement the test and are i.i.d. conditional on the auxiliary subsample. The conclusions are unchanged.
}.

This application illustrates an important practical feature of our proposed test. Although the boundary is embedded in $\mathbb R^7$ and is challenging to visualize, implementing the test remains unchanged. The researcher only needs to compute distances to the boundary, project observations onto it, and apply the same conditional randomization procedure used in the lower-dimensional applications. At the same time, the high-dimensional nature of the assignment rule makes finite-sample proximity to the boundary more delicate, so robustness across values of $q$ is especially informative.

Figure \ref{fig:elaqua_et_al} reports the empirical results of our test for this application. The conclusion depends somewhat on the smallest values of $q$, but becomes clear as the selected sample expands to include observations farther from the boundary. For small values of $q$, the simulated $p$-values are mostly above conventional significance levels and fluctuate considerably. This is consistent with panel (b), where the observed statistic is close to, and often below, the simulated $95\%$ critical value. As $q$ increases, however, the observed statistic rises steadily and eventually exceeds the critical value by a large margin. Correspondingly, the $p$-values fall toward zero and remain very small for larger values of $q$, indicating an imbalance in the side indicators of observations near the boundary. Overall, the results provide strong evidence against $H_0$ for this application, although the evidence is less apparent when the test uses only the observations closest to the boundary.

\begin{figure}[htbp]
\centering
\begin{subfigure}{\paperhalfwidth}
    \centering
    \paperfighalf{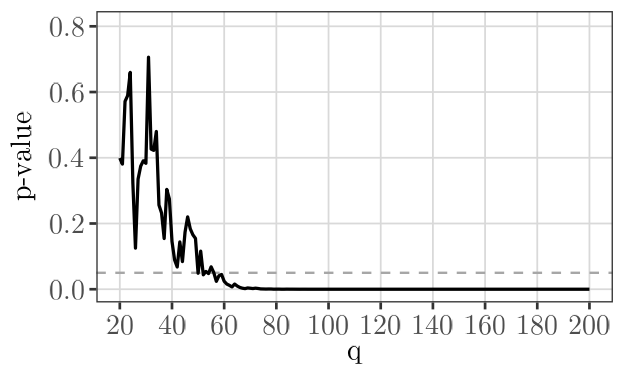}
    \caption{$p$-value as a function of $q$.}
\end{subfigure}
\hfill
\begin{subfigure}{\paperhalfwidth}
    \centering
    \paperfighalf{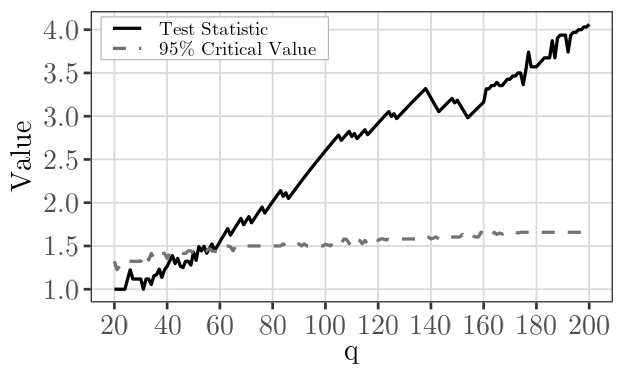}
    \caption{Test statistic and $95\%$ CV as a function of $q$.}
\end{subfigure}
\caption{Results of our test applied to the data in \cite{elacqua/hincapie/martinez:2024}. Panel (a) reports the simulated $p$-value for each value of $q$. Panel (b) reports the corresponding value of the test statistic and $95\%$ critical value.}
\label{fig:elaqua_et_al}
\end{figure}

\section{Conclusions and future research}\label{sec:conclusions}

This paper proposed the first manipulation test designed for BDDs with general boundaries. The motivation follows the logic of \citet{lee:2008} for one-dimensional RDDs: in the absence of manipulation, the density of the running variable should be continuous at the cutoff. In a BDD, however, assignment status is determined by whether a vector running variable crosses a general boundary, so the relevant null is the joint density continuity null across the boundary at every boundary point. Existing one-dimensional density tests and signed-distance approaches do not directly test this null without additional assumptions.

Our test avoids multivariate density estimation by exploiting a simple implication of density continuity. Observations sufficiently close to the boundary should be approximately balanced between the treatment and control regions, both overall and within localized regions along the boundary. To operationalize this idea, we select the $q$ observations closest to the boundary, project them onto the boundary, and form clusters based on their projected locations using $k$-means. We then test whether the side indicators are balanced within these clusters. The test statistic combines binomial balance statistics across different cluster-count choices, allowing it to detect widespread, localized, and offsetting manipulation. We compute the critical values using Bernoulli relabeling draws justified by our asymptotic results.

We establish the asymptotic validity of the procedure under regularity conditions that allow for general boundaries. For fixed $q$, under the joint density continuity null, the selected side indicators are asymptotically distributed as independent $\mathrm{Bernoulli}(1/2)$ random variables, conditional on the projected locations. This random-labeling property provides the basis for simulated critical values from the conditional Bernoulli distribution and yields asymptotic size control for the non-randomized version of the test. The asymptotic framework, therefore, justifies a simple implementation that preserves the geometry of the assignment rule while avoiding the curse of dimensionality associated with estimating a multivariate density.

The Monte Carlo simulations show that our proposed test performs well in finite samples. 
Under the joint density continuity null, our test exhibits good size control in our simulations when the selected observations are sufficiently close to the boundary. It also exhibits power against localized, widespread, and offsetting discontinuities across different parts of the boundary. We also apply the test to three empirical settings: a two-score BDD with an L-shaped boundary induced by two threshold rules, a geographic RDD, and a higher-dimensional design with seven running variables and a nontrivial institutional assignment rule. These applications illustrate that the test is easy to implement and can provide informative diagnostics in BDDs with different boundary geometries and different dimensions of the vector running variable.

Beyond the density continuity null studied here, empirical practice in one-dimensional RDDs typically pairs the manipulation test with tests of continuity of predetermined covariates at the cutoff, often called balance tests; see, for example, \citet{lee/lemieux:2010}. Extending such tests to BDDs to assess continuity of predetermined covariates across the boundary at every boundary point is a natural next step that we are currently developing. The geometric and random labeling tools developed in this paper, in particular selecting observations closest to the boundary, projecting them onto $\mathcal B$, and localizing along the boundary, provide a foundation for that work.

\appendices
\begin{small}

\section{Proofs} \label{sec:proofs}

This appendix uses the following notation. For any $x\in \mathbb{R} ^{d}$ and $r\geq 0$, $B( x,r) \equiv \{ z\in \mathbb{R} ^{d}:\Vert z-x\Vert < r \} $ denotes the open ball centered at $x$ and radius $r$. For any $t \in \{0,1\}$, we use $\{T=t\}\equiv \{ z\in \mathbb{R} ^{d}:T(z)=t\} $. For any set $A$, we use $bd(A)$ to denote the topological boundary of $A$, as in Section \ref{sec:setup}.
 
Given a collection of boundary points $\{b_i\}_{i=1}^q \in \mathcal{B}^q$ and a number of clusters $k \in \{1,2,\dots,q\}$, let
\begin{equation}
    \Upsilon(\{b_i\}_{i=1}^q,k) ~\in~  \Pi_k ~\subseteq~ \{1,2,\dots,k\}^q
    \label{eq:pi_function}
\end{equation}
denote the cluster-assignment vector returned by the implementation of the $k$-means algorithm described in Section \ref{sec:clustering}, including the ordering of the cluster labels and the resolution of ties. For example, if  $q=3$ and $k=2$, $\Upsilon(\{b_i\}_{i=1}^3,2)$ can only take values $(1,1,2)$, $(1,2,1)$, and $(2,1,1)$, and so $\Pi_2 = \{(1,1,2),(1,2,1),(2,1,1)\} \subseteq \{1,2\}^3$. More generally, the range of $\Upsilon(\{b_i\}_{i=1}^q,k)$ is a strict subset of $\{1,2,\dots,k\}^q$, which we denote by $\Pi_k$. 

Let $\mathcal K={k_1,k_2,\dots,k_J}$ denote the cluster-count grid used in the $k$-means algorithm in Section \ref{sec:clustering}. In the default implementation, $\mathcal K=\{1,2,\dots,q\}$ and $J=q$. Let $\Pi=\Pi_{k_1}\times \Pi_{k_2}\cdots \times \Pi_{k_J}$ denote the range of $\{\Upsilon(\{b_i\}_{i=1}^q,k_j)\}_{j=1}^J$. By definition, we have that $\tilde{\pi}=\{\Upsilon(\tilde B,k_j)\}_{j=1}^J$, where $\tilde B=\{B(\tilde Z_i)\}_{i=1}^q$.

\subsection{Proof of theorems}

For any $z \in \mathbb{R}^d$, let $g(z)$ denote the distance from $z$ to $\mathcal{B}$, i.e.,
\begin{equation*}
    g(z) ~\equiv~ \dist(z,\mathcal{B}) ~\equiv ~ \inf_{b \in \mathcal{B}} \Vert z-b \Vert.
\end{equation*}
Step 1 of our testing procedure selects the $q$ observations in the sample with the smallest values of $g(Z_i)$. The goal of Theorem \ref{thm:keyPart1} is to show that these selected observations behave asymptotically like $q$ i.i.d. draws from the boundary-local distribution $P^*$.

The proof uses an auxiliary sample to guarantee the uniqueness of the $g$-ordering. Under Assumption \ref{ass:assumption}(a), the original sample $\{Z_i\}_{i=1}^{n}$ is continuously distributed only in a $\delta$-neighborhood of the boundary, whereas the order-statistic result of \citet[Theorem 1]{kauffman/reiss:1992} is stated for continuously distributed observations. For this reason, for $i=1, \dots, n$, define
\begin{equation*}
    Z_{i}^{\mathrm{cts}} ~\equiv~ Z_{i} + v_{i} { \bf 1}\{\dist(Z_i,\mathcal{B}) > \delta \},
\end{equation*}
where the perturbations $v_i$ are i.i.d.\ with $v_i \sim \mathrm{Unif}(B({\bf 0}_{d},\delta/2))$, and are independent of the data. This implies that $Z_i^{\mathrm{cts}}$ is absolutely continuously distributed. Hence, observations already within distance $\delta$ of the boundary are left unchanged, while observations farther away are only perturbed slightly. Importantly, this change is only required within the proof of Theorem \ref{thm:keyPart1}; no changes to the data are required in practice.

For any $z \in \mathbb{R}^d$, let $V(z)= (z,B(z),T(z))$. As in \citet[page 67]{kauffman/reiss:1992}, define $V(z_1) \leq_{g} V(z_2)$ if and only if $g(z_1) \leq g(z_2)$. Rearrange $\{Z_{i}^{\mathrm{cts}}\}_{i=1}^{n}$ in terms of their $g$-order:
\begin{equation}
    \{V(Z_{i}^{\mathrm{cts}})\}_{\left[ 1:n\right] }~\leq _{g}~\{V(Z_{i}^{\mathrm{cts}})\}_{\left[ 2:n\right] }~\leq _{g}~\dots~\leq _{g}~\{V(Z_{i}^{\mathrm{cts}})\}_{\left[ n:n\right] }.
\label{eq:concomitants}
\end{equation}
Because $Z_i^{\mathrm{cts}}$ is absolutely continuously distributed and every level set of $g(z)=\dist(z,\mathcal B)$ has Lebesgue measure zero, $g(Z_i^{\mathrm{cts}})$ has an atomless distribution. Consequently, the values $\{g(Z_i^{\mathrm{cts}})\}_{i=1}^{n}$ are almost surely distinct, so this ordering is almost surely unique.

For a fixed $q \in \mathbb{N}$, Step 1 chooses the first $q$ concomitants $\{\{V(Z_{i}^{\mathrm{cts}})\}_{\left[ 1:n\right] },\{V(Z_{i}^{\mathrm{cts}})\}_{\left[ 2:n\right] },\dots,\{V(Z_{i}^{\mathrm{cts}})\}_{\left[ q:n\right] }\}$. Let $\{V(\tilde{Z}_{i}^{\mathrm{cts}})\}_{i=1}^{q}$ denote these concomitants, relabeled in original sample order.

\begin{proof}[Proof of Theorem \ref{thm:keyPart1}] We divide this proof into two parts. The first part proves convergence in distribution to an i.i.d.\ sample. The second part proves the desired properties of the (common) limiting distribution.

\noindent\underline{Part 1.} We show that the common limiting distribution is given by $P^*$ in \eqref{eq:limitP}, defined in Lemma \ref{lem:boundaryLocalLimit}.

For any $z \in \mathbb{R}^d$, let $V(z)= (z,B(z),T(z))$. Let $V_i^*=(Z_i^*,B_i^*,T_i^*)$ for $i=1,\ldots,q$. Let $S \subseteq (\mathbb{R}^d \times \mathbb{R}^d \times \{0,1\})^q$ be an arbitrary continuity set of $(P^*)^{\otimes q}$ and fix $\varepsilon>0$ arbitrarily. It suffices to show that for all sufficiently large $n$, 
\begin{equation}
 |P(\{V(\tilde Z_i)\}_{i=1}^q \in S)-P^*(\{V^*_i\}_{i=1}^q\in S)| ~<~ \varepsilon.  
  \label{eq:keyPart1_2}
\end{equation}

Let $\{V_{i,u}\}_{i=1}^{q}$ be i.i.d.\ with common law $P_u$. Since $P_u$ converges weakly to $P^*$ and $q$ is fixed, $P_u^{\otimes q}$ converges weakly to $(P^*)^{\otimes q}$. By Lemma \ref{lem:boundaryLocalLimit} and the fact that $S$ is a continuity set of $(P^*)^{\otimes q}$, there exists $\delta_1 \in (0, \delta/2)$ such that $P_u$ is well defined for every $u\in(0,\delta_1)$ and, for every $u \in (0,\delta_1)$, 
\begin{equation}
    |P(\{V_{i,u}\}_{i=1}^{q}\in S) - P^*(\{V^*_i\}_{i=1}^q\in S)| ~<~ \varepsilon /3. \label{eq:keyPart1_3}
\end{equation}

As an auxiliary step, we derive a multidimensional analog of \citet[Lemma B.3]{bugni/canay:2021}. Let $D^{\mathrm{cts}}_{q+1:n}$ denote the $(q+1)$th order statistic of $\{g(Z_i^{\mathrm{cts}})\}_{i=1}^{n}$, where $g(z)=\dist(z,\mathcal{B})$. Set $N_{n} = \sum_{i=1}^{n} { \bf 1}\{ \dist(Z_i , \mathcal{B} ) \leq \delta_1 \}$. Since $\{Z_i\}_{i=1}^n$ are i.i.d., $N_{n} \sim \mathrm{Bi}(n, P( \dist(Z_i , \mathcal{B} ) \leq \delta_1 ) )$. Because $P_{\delta_1}$ is well defined, $P( \dist(Z_i , \mathcal{B} ) \leq \delta_1 ) > 0$. Therefore, by the Strong Law of Large Numbers, ${N_n}/{n} \overset{a.s.}{\to} P( \dist(Z_i , \mathcal{B} ) \leq \delta_1 ) > 0$. Thus, $N_n \geq q+1$ for all $n$ large enough with probability 1. Since $\delta_1<\delta/2$, observations with $\dist(Z_i,\mathcal B)\leq\delta_1$ are unchanged by the auxiliary perturbation. Thus, if at least $q+1$ observations lie within distance $\delta_1$ of the boundary, then the $(q+1)$th closest observation in the auxiliary sample must also lie within distance $\delta_1$. Hence
\begin{equation} \label{eq:keyPart1_4}
    \lim_{n \to \infty} P( D^{\mathrm{cts}}_{q+1:n} > \delta_1 ) ~=~ 0.  
\end{equation}

Now decompose
\begin{equation*}
    P(\{V(\tilde Z_i)\}_{i=1}^q \in S) - P^*(\{V^*_i\}_{i=1}^q\in S)~=~ M_{1,n} + M_{2,n} + M_{3,n},
\end{equation*}
where 
\begin{align*}
    M_{1,n}& ~=~ \int_{0}^{\delta_1 } P(\{V(\tilde Z_i)\}_{i=1}^q \in S|D^{\mathrm{cts}}_{q+1:n}=u) dP_{D^{\mathrm{cts}}_{q+1:n}}(u) - P^*(\{V^*_i\}_{i=1}^q\in S) P(D^{\mathrm{cts}}_{q+1:n} \leq \delta_1), \\
    M_{2,n}& ~=~ \int_{\delta_1}^{\infty } P(\{V(\tilde Z_i)\}_{i=1}^q \in S|D^{\mathrm{cts}}_{q+1:n}=u) dP_{D^{\mathrm{cts}}_{q+1:n}}(u), \\
    M_{3,n}& ~=~ -P^*(\{V^*_i\}_{i=1}^q\in S) P(D^{\mathrm{cts}}_{q+1:n} > \delta_1 ).
\end{align*}

Next, note that 
\begin{equation}
    |M_{2,n}| ~\leq~ P(D^{\mathrm{cts}}_{q+1:n} > \delta_1 )~~~~~~~\text{and}~~~~~~~ 
    |M_{3,n}| ~\leq~ P(D^{\mathrm{cts}}_{q+1:n} > \delta_1 ). \label{eq:keyPart1_6}
\end{equation}
By \eqref{eq:keyPart1_4} and \eqref{eq:keyPart1_6}, for all $n$ large enough, we have 
\begin{equation}
|M_{2,n}| < \varepsilon/3 ~~~~~~~\text{and}~~~~~~~ |M_{3,n}| < \varepsilon /3.  \label{eq:keyPart1_8}
\end{equation}
Also, 
\begin{align}
|M_{1,n}|& ~=~\Big| \int_{0}^{\delta_1 }P(\{V(\tilde Z_i)\}_{i=1}^q \in S|D^{\mathrm{cts}}_{q+1:n}=u) dP_{D^{\mathrm{cts}}_{q+1:n}}(u) - P^*(\{V^*_i\}_{i=1}^q\in S)P(D^{\mathrm{cts}}_{q+1:n}\leq \delta_1 ) \Big|  \notag\\
& ~\overset{(1)}{=}~\Big|\int_{0}^{\delta_1 }P(\{V(\tilde Z_i^{\mathrm{cts}})\}_{i=1}^q \in S|D^{\mathrm{cts}}_{q+1:n}=u) dP_{D^{\mathrm{cts}}_{q+1:n}}(u) - P^*(\{V^*_i\}_{i=1}^q\in S)P(D^{\mathrm{cts}}_{q+1:n}\leq \delta_1 )\Big|  \notag\\
&~ \overset{(2)}{=}~\Big|\int_{0}^{\delta_1 }P(\{V_{i,u}\}_{i=1}^{q}\in S)dP_{D^{\mathrm{cts}}_{q+1:n}}(u)-P^*(\{V^*_{i}\}_{i=1}^{q}\in S)P(D^{\mathrm{cts}}_{q+1:n}\leq \delta_1 )\Big| \notag\\
& ~\leq ~\int_{0}^{\delta_1 }\left|P(\{V_{i,u}\}_{i=1}^{q}\in S)-P^*(\{V^*_{i}\}_{i=1}^{q}\in S)\right|dP_{D^{\mathrm{cts}}_{q+1:n}}(u)  \notag \\
& ~\overset{(3)}{<}~\varepsilon /3,  \label{eq:keyPart1_9}
\end{align}
where (1) holds because, on $\{D^{\mathrm{cts}}_{q+1:n}\leq \delta_1 <\delta /2\}$, the selected observations from the auxiliary sample coincide with the selected observations from the original sample, i.e., $\{\tilde{Z}_{i}^{\mathrm{cts}}\}_{i=1}^q=\{\tilde{Z}_{i}\}_{i=1}^q$, (2) holds by \citet[Theorem 1, equation (6)]{kauffman/reiss:1992}, which shows that $\{\{V(\tilde Z_i^{\mathrm{cts}})\}_{i=1}^{q}|D^{\mathrm{cts}}_{q+1:n}=u\}$ is i.i.d.\ with distribution $P_{u}$, and (3) by \eqref{eq:keyPart1_3}. Therefore, for all $n$ large enough, \eqref{eq:keyPart1_8} and \eqref{eq:keyPart1_9} imply that \eqref{eq:keyPart1_2} holds. Since $S$ was an arbitrary continuity set of $(P^*)^{\otimes q}$, the desired weak convergence follows.

\underline{Part 2.} We now prove the properties of the limiting distribution. Let $A\subseteq\mathbb R^d\times\mathbb R^d$ be Borel-measurable, and let $t\in\{0,1\}$. By \eqref{eq:limitP},
\begin{equation*}
P^*(\{(Z_i^*,B_i^*)\in A\} \cap \{T_i^*=t\})
~=~
\frac{
\int_{\mathcal B}\mathbf 1\{(b,b)\in A\}f_Z(b)d\mathcal H^{d-1}(b)
}{
2\int_{\mathcal B}f_Z(b)d\mathcal H^{d-1}(b)
}.
\end{equation*}
By computing marginals from this expression, we obtain
\begin{align*}
P^*((Z_i^*,B_i^*)\in A)
~&=~
\frac{
\int_{\mathcal B}\mathbf 1\{(b,b)\in A\}f_Z(b)d\mathcal H^{d-1}(b)
}{
\int_{\mathcal B}f_Z(b)d\mathcal H^{d-1}(b)
},\\
P^*(T_i^*=t)~&=~1/2.
\end{align*}
Hence, $P^*(\{(Z_i^*,B_i^*)\in A\} \cap \{T_i^*=t\})=
P^*((Z_i^*,B_i^*)\in A)P^*(T_i^*=t)$, i.e.,  $T_i^*$ is independent of $(Z_i^*,B_i^*)$, and $T_i^*\sim \mathrm{Bernoulli}(1/2)$. This proves parts (a) and (b). Also, \eqref{eq:limitP} assigns probability one to triples of the form $(b,b,t)$ with $b\in\mathcal B$ and $t\in\{0,1\}$. Therefore, $P^*(Z_i^*=B_i^*)=1$, which proves part (c).
\end{proof}

\begin{proof}[Proof of Theorem \ref{thm:clustering}]
Let $\mathcal{K}=\{k_{1},k_{2},\dots ,k_{J}\}$ denote the collection of $k$ parameters used in the k-means algorithm. (In the default implementation, $\mathcal{K}=\{1,2,\dots ,q\}$ and $J=q$.) For any $b=\{b_{i}\}_{i=1}^{q}\in \mathcal{B}^{q}$ and $k\in \{1,2,\dots ,q\}$, let $\Upsilon (b,k)\in \Pi _{k}$ be as in \eqref{eq:pi_function}, and recall that $\tilde{\pi}=\{\Upsilon (\tilde{B},k_{j})\}_{j=1}^{J}$, where $\tilde{B}=\{\tilde{B}_{i}\}_{i=1}^{q}$. 

For each $\pi \in \Pi $, define 
\begin{equation*}
E(\pi )~\equiv ~\left\{ b\in \mathcal{B}^{q}:\{\Upsilon (b,k_{j})\}_{j=1}^{J}=\pi \right\} .
\end{equation*}
By construction, $\{\tilde{\pi}=\pi \}=\{\tilde{B}\in E(\pi )\}$. Because the k-means procedure breaks ties measurably, $\{\Upsilon (b,k_{j})\}_{j=1}^{J}$ is measurable and hence $E(\pi )\in Bo(\mathcal{B} ^{q})$.

Fix $\varepsilon >0$ arbitrarily. It suffices to show that
\begin{equation}
\lim_{n\rightarrow \infty }P\Big( \sup_{t\in \{0,1\}^{q}}\Big\vert P( \tilde{T}=t\mid \tilde{\pi})-2^{-q}\Big\vert >\varepsilon \Big) ~=~0.
\label{eq:clustering_2}
\end{equation}
Note that
\begin{align}
P\Big( \sup_{t\in \{0,1\}^{q}}\Big\vert P( \tilde{T}=t\mid \tilde{\pi})-2^{-q}\Big\vert >\varepsilon \Big)
~&\overset{(1)}{\leq }~ E\Big[ \sup_{t\in \{0,1\}^{q}}\Big\vert P(\tilde{T} =t\mid \tilde{\pi})-2^{-q}\Big\vert \Big]/\varepsilon \notag \\
&=~\sum_{\pi \in \Pi }\sup_{t\in \{0,1\}^{q}}\vert P(\tilde{T}=t,\tilde{\pi}=\pi )-2^{-q}P( \tilde{\pi}=\pi ) \vert/\varepsilon \notag \\
&\overset{(2)}{=}~\sum_{\pi \in \Pi }\sup_{t\in \{0,1\}^{q}}\vert P(\tilde{T}=t,\tilde{B}\in E(\pi ))-2^{-q}P( \tilde{B}\in E(\pi )) \vert/\varepsilon \notag \\
&\overset{(3)}{\leq}~ \sum_{\pi \in \Pi }\sup_{S\in Bo(\mathcal{B} ^{q}),t\in \{0,1\}^{q}}\vert P(\tilde{T}=t,\tilde{B}\in S)-2^{-q}P( \tilde{B}\in S) \vert /\varepsilon\label{eq:clustering_4}
\end{align}
where (1) holds by Markov's inequality, (2) by definition of $E(\pi )$, and (3) by $E(\pi )\in Bo(\mathcal{B}^{q})$. By \eqref{eq:clustering_4} and the fact that $\Pi $ is finite, \eqref{eq:clustering_2} follows from showing that
\begin{equation}
\lim_{n\rightarrow \infty }\sup_{S\in Bo(\mathcal{B}^{q}),t\in \{0,1\}^{q}}\vert P(\tilde{T}=t,\tilde{B}\in S)-2^{-q}P( \tilde{B}\in S) \vert ~=~0.
\end{equation}
Fix $\mu >0$ arbitrarily. It suffices to find $N=N( \mu ) $ such that for all $n\geq N$,
\begin{equation} 
\sup_{S\in Bo(\mathcal{B}^{q}),t\in \{0,1\}^{q}}\vert P(\tilde{T}=t,\tilde{B}\in S)-2^{-q}P( \tilde{B}\in S) \vert~ \leq~ \mu .
\label{eq:clustering_5}
\end{equation}
We show this in the remainder of this proof.

Let $P_{u}$, $\delta _{1}\in (0,\delta /2)$, $D_{q+1:n}^{\mathrm{cts}}$, $\{Z_{i}^{\mathrm{cts}}\}_{i=1}^{n}$, $\{V(Z_{i}^{\mathrm{cts}})\}_{i=1}^{n}$, $\{g(Z_{i}^{\mathrm{cts}})\}_{i=1}^{n}$, $\{V(\tilde{Z}_{j}^{\mathrm{cts} })\}_{j=1}^{q}$ be as defined in the proof of Theorem \ref{thm:keyPart1}. We denote $\tilde{Z}_{j}^{\mathrm{cts}}=V_{1}(\tilde{Z}_{j}^{\mathrm{cts}})$, $ \tilde{B}_{j}^{\mathrm{cts}}=V_{2}(\tilde{Z}_{j}^{\mathrm{cts}})$ and $ \tilde{T}_{j}^{\mathrm{cts}}=V_{3}(\tilde{Z}_{j}^{\mathrm{cts}})$. Recall that $\{D_{q+1:n}^{\mathrm{cts}}\leq \delta _{1}<\delta /2\}$ implies that $ \{V(\tilde{Z}_{j}^{\mathrm{cts}})\}_{j=1}^{q}=\{V(\tilde{Z}_{j})\}_{j=1}^{q}$ and, conditional on $D_{q+1:n}^{\mathrm{cts}}=u$, $\{V(\tilde{Z}_{j}^{ \mathrm{cts}})\}_{j=1}^{q}$ are i.i.d.\ with common distribution $P_{u}$ for all sufficiently small $u>0$.  Therefore, $\{D_{q+1:n}^{\mathrm{cts}}\leq \delta _{1}<\delta /2\}$ implies that $\{(\tilde{B}_{j}^{\mathrm{cts}}, \tilde{T}_{j}^{\mathrm{cts}})\}_{j=1}^{q}=\{(\tilde{B}_{j},\tilde{T} _{j})\}_{j=1}^{q}$ and, conditional on $D_{q+1:n}^{\mathrm{cts}}=u$, $\{( \tilde{B}_{j}^{\mathrm{cts}},\tilde{T}_{j}^{\mathrm{cts}})\}_{j=1}^{q}$ are i.i.d.\ with common distribution $Q_{u}$ for every sufficiently small $u>0$, and $Q_{u,B}$ to denote its marginal distribution over the $B_i$ component. 

We begin with some preliminary results. For any $u>0$, let 
\begin{equation*}
\delta (u)~\equiv ~\underset{{ S\in Bo(\mathcal{B}), a\in \{0,1\}}}{\sup} \big\vert P( B(Z_{i})\in S,T(Z_{i})=a\mid \dist(Z_{i},\mathcal{B})\leq u) -\frac{1}{2}P( B(Z_{i})\in S\mid \dist(Z_{i},\mathcal{B})\leq u) \big\vert ,
\end{equation*}
where conditional probabilities are well defined by Lemma \ref{lem:BoundaryHasProb}. Note that
\begin{equation}
\lim_{u\downarrow 0}\frac{P(\dist(Z_{i},\mathcal{B})\leq u)}{u}~\overset{(1)} {=}~2\int_{\mathcal{B}}f_{Z}(b)d\mathcal{H}^{d-1}(b)~\overset{(2)}{>}~0,
\label{eq:clustering_6}
\end{equation}
where (1) follows from Lemma \ref{lem:LimitDen} and (2) follows from \eqref{eq:boundaryLocalLimit_2} in Lemma \ref{lem:boundaryLocalLimit}. Also, for every $S\in Bo(\mathcal{B})$ and $a\in \{0,1\}$, 
\begin{align}
& \vert P( B(Z_{i})\in S,T(Z_{i})=a\mid \dist(Z_{i},\mathcal{B} )\leq u) -\frac{1}{2}P( B(Z_{i})\in S\mid \dist(Z_{i},\mathcal{B} )\leq u) \vert  \notag \\
& =~\frac{\left\vert 
\begin{array}{c}
P( B(Z_{i})\in S,T(Z_{i})=a,\dist(Z_{i},\mathcal{B})\leq u) /u-\int_{S}f_{Z}(b)d\mathcal{H}^{d-1}(b) \\ 
-P( B(Z_{i})\in S,\dist(Z_{i},\mathcal{B})\leq u) /(2u)+\int_{S}f_{Z}(b)d\mathcal{H}^{d-1}(b)
\end{array}
\right\vert }{P(\dist(Z_{i},\mathcal{B})\leq u)/u}. \label{eq:clustering_8}
\end{align}
By \eqref{eq:clustering_6}, \eqref{eq:clustering_8}, and the uniform conclusion of Lemma \ref{lem:LimitDen}, we deduce that $\lim_{u\downarrow0}\delta(u)=0$. Therefore, $\exists \delta _{2}\in (0,\delta _{1})$ so that, for all $u\in (0,\delta _{2})$,
\begin{equation}
\delta (u)~\leq~ \mu /(4q). \label{eq:clustering_10}
\end{equation}
By repeating the arguments leading to \eqref{eq:keyPart1_4} with $\delta_2$ in place of $\delta_1$, we can find $N=N(\mu)$ such that, for all $n\geq N$,
\begin{equation}
P(D_{q+1:n}^{\mathrm{cts}}>\delta _{2})~\leq~ \mu /4. \label{eq:clustering_12}
\end{equation}

For any $u\in (0,\delta _{2})$, we then have
\begin{align}
&\Vert Q_{u}^{\otimes q}-\left\{ Q_{u,B}\otimes \mathrm{Bernoulli} (1/2)\right\} ^{\otimes q}\Vert _{\mathrm{TV}} \notag \\
&\overset{(1)}{\leq } ~q\left\Vert Q_{u}-Q_{u,B}\otimes \mathrm{Bernoulli}(1/2)\right\Vert _{ \mathrm{TV}} \notag \\
&\overset{(2)}{=}~q\sup_{A\in Bo(\mathcal{B}\times \{0,1\})}\left\vert Q_{u}(A)-\{Q_{u,B}\otimes \mathrm{Bernoulli}(1/2)\}(A)\right\vert   \notag \\
&\overset{(3)}{=}~q\sup_{S_{0},S_{1}\in Bo(\mathcal{B})}\left\vert
Q_{u}\left( (S_{0}\times \{0\})\cup (S_{1}\times \{1\})\right) -Q_{u,B}(S_{0})/2-Q_{u,B}(S_{1})/2\right\vert   \notag \\
&\overset{(4)}{=}~q\sup_{S_{0},S_{1}\in Bo(\mathcal{B})}\left\vert Q_{u}(S_{0}\times \{0\})+Q_{u}(S_{1}\times \{1\})-Q_{u,B}(S_{0})/2-Q_{u,B}(S_{1})/2\right\vert   \notag \\
&\leq ~q\sup_{S_{0},S_{1}\in Bo(\mathcal{B})}\left\{ \left\vert Q_{u}(S_{0}\times \{0\})-Q_{u,B}(S_{0})/2\right\vert +\left\vert Q_{u}(S_{1}\times \{1\})-Q_{u,B}(S_{1})/2\right\vert \right\}  \notag \\
&\overset{(5)}{\leq }~\mu /2,
\label{eq:clustering_14}
\end{align}
where (1) holds by the tensorization inequality for total variation, (2) by the definition of total variation, (3) by the fact that every $A\in Bo( \mathcal{B}\times \{0,1\})$ can be written as $A=(S_{0}\times \{0\})\cup (S_{1}\times \{1\})$ for some $S_{0},S_{1}\in Bo(\mathcal{B})$, (4) by $ (S_{0}\times \{0\})\cap (S_{1}\times \{1\})=\emptyset $, and (5) by \eqref{eq:clustering_10}.

Next, for any $S\in Bo(\mathcal{B}^{q})$ and $t\in \{0,1\}^{q}$, we have
\begin{align}
&\vert P(\tilde{B}\in S,\tilde{T}=t)-2^{-q}P\left( \tilde{B}\in S\right) \vert  \notag \\
&=~\left\vert 
\begin{array}{c}
\int_{0}^{\delta _{2}}P(\tilde{B}\in S,\tilde{T}=t\mid D_{q+1:n}^{\mathrm{cts }}=u)dP_{D_{q+1:n}^{\mathrm{cts}}}(u)-2^{-q}\int_{0}^{\delta _{2}}P(\tilde{B} \in S\mid D_{q+1:n}^{\mathrm{cts}}=u)dP_{D_{q+1:n}^{\mathrm{cts}}}(u) \\ 
+\int_{\delta _{2}}^{\infty }P(\tilde{B}\in S,\tilde{T}=t\mid D_{q+1:n}^{ \mathrm{cts}}=u)dP_{D_{q+1:n}^{\mathrm{cts}}}(u)-2^{-q}\int_{\delta _{2}}^{\infty }P(\tilde{B}\in S\mid D_{q+1:n}^{\mathrm{cts}
}=u)dP_{D_{q+1:n}^{\mathrm{cts}}}(u)
\end{array}
\right\vert   \notag \\
&\leq~
\Big\vert \int_{0}^{\delta _{2}}\left( P(\tilde{B}\in S,\tilde{T}=t\mid D_{q+1:n}^{\mathrm{cts}}=u)-2^{-q}P(\tilde{B}\in S\mid D_{q+1:n}^{\mathrm{cts }}=u)\right) dP_{D_{q+1:n}^{\mathrm{cts}}}(u)\Big\vert  
+2P(D_{q+1:n}^{\mathrm{cts}}>\delta _{2}) \notag \\
&\overset{(1)}{=}~\left\{ 
\begin{array}{c}
\vert \int_{0}^{\delta _{2}}( P(\tilde{B}^{\mathrm{cts}}\in S, \tilde{T}^{\mathrm{cts}}=t\mid D_{q+1:n}^{\mathrm{cts}}=u)-2^{-q}P(\tilde{B} ^{\mathrm{cts}}\in S\mid D_{q+1:n}^{\mathrm{cts}}=u)) dP_{D_{q+1:n}^{
\mathrm{cts}}}(u)\vert  \\ 
+2P(D_{q+1:n}^{\mathrm{cts}}>\delta _{2})
\end{array}
\right\}   \notag \\
&\overset{(2)}{=}~
\big\vert \int_{0}^{\delta _{2}}( Q_{u}^{\otimes q}(S\times \{t\})-\left\{ Q_{u,B}\otimes \mathrm{Bernoulli}(1/2)\right\} ^{\otimes q}(S\times \{t\})) dP_{D_{q+1:n}^{\mathrm{cts}}}(u)\big\vert +2P(D_{q+1:n}^{\mathrm{cts}}>\delta _{2})
  \notag \\
& \overset{(3)}{\leq }~\mu /2+2P(D_{q+1:n}^{\mathrm{cts}} >\delta _{2}),
\label{eq:clustering_16}
\end{align}
where (1) holds by $\delta _{2}\leq \delta _{1}$ and $\{D_{q+1:n}^{\mathrm{ cts}}\leq \delta _{1}<\delta /2\}$ implies that $\{(\tilde{B}_{j}^{\mathrm{ cts}},\tilde{T}_{j}^{\mathrm{cts}})\}_{j=1}^{q}=\{(\tilde{B}_{j},\tilde{T}_{j})\}_{j=1}^{q}$, (2) by the fact that, conditional on $D_{q+1:n}^{\mathrm{ cts}}=u$, $\{(\tilde B_j^{\mathrm{cts}},\tilde T_j^{\mathrm{cts}})\}_{j=1}^q$ are i.i.d.\ with common distribution $Q_u$, and (3) by $\delta _{1}\geq \delta _{2}$ and \eqref{eq:clustering_14}. 

Taking the supremum in \eqref{eq:clustering_16} and considering $n\geq N( \mu ) $, we obtain
\begin{equation}
\sup_{S\in Bo(\mathcal{B}^{q}),t\in \{0,1\}^{q}}\vert P(\tilde{T}=t,\tilde{B}\in S)-2^{-q}P( \tilde{B}\in S) \vert ~\overset{(1)}{\leq }~\mu /2+2P(D_{q+1:n}^{\mathrm{cts}}>\delta _{2})~\overset{(2)}{\leq }~\mu ,
\label{eq:clustering_18}
\end{equation}
where (1) holds by \eqref{eq:clustering_16} and (2) by \eqref{eq:clustering_12} and $n\geq N( \mu ) $. This proves \eqref{eq:clustering_5} and completes the proof.
\end{proof}

\begin{proof}[Proof of Theorem \ref{thm:validity}]
Let $\mathcal{K}=\{k_{1},k_{2},\dots,k_{J}\}$ denote the collection of $k$ parameters used in the k-means algorithm. (In the default implementation of the algorithm, $\mathcal{K}=\{1,2,\dots,q\}$ and $J=q$.) For any $\{b_{i}\}_{i=1}^{q}\in\mathcal{B}^{q}$ and $k\in\{1,2,\dots,q\}$, let $\Upsilon(\{b_{i}\}_{i=1}^{q},k)\in\Pi_{k}$ be the function defined in \eqref{eq:pi_function}. As explained after introducing \eqref{eq:pi_function}, $\tilde\pi=\{\Upsilon(\tilde B,k_{j})\}_{j=1}^{J}\in\Pi$, where $\tilde B=\{B(\tilde Z_{i})\}_{i=1}^{q}$.

For any $\pi\in\Pi$, define
\begin{equation}
\Delta_{n}(\pi)~\equiv~
\sup_{t\in\{0,1\}^{q}}
|P(\tilde T=t\mid\tilde\pi=\pi)-2^{-q}|.
\label{eq:validity_0}
\end{equation}
By Theorem \ref{thm:clustering},
\begin{equation}
\Delta_{n}(\tilde\pi)~\overset{p}{\to}~0.
\label{eq:validity_0a}
\end{equation}
Because $0\leq\Delta_{n}(\tilde\pi)\leq1$, \eqref{eq:validity_0a} implies that
\begin{equation}
E[\Delta_{n}(\tilde\pi)]~\to~0.
\label{eq:validity_0b}
\end{equation}

We first prove the result for the randomized version. For any $\pi\in\Pi$ and $u\in[0,1]$, define
\begin{equation}
\Psi_{\rm r}(\pi,u)~\equiv~
\{t\in\{0,1\}^{q}:S(t,\pi)>c(\alpha,\pi)\}
\cup
\{t\in\{0,1\}^{q}:S(t,\pi)=c(\alpha,\pi),\,u\leq u(\alpha,\pi)\}.
\label{eq:validity_1}
\end{equation}
By definition, the randomized test rejects $H_{0}$ if and only if $\tilde T\in\Psi_{\rm r}(\tilde\pi,U)$, where $U\sim\mathrm{Unif}(0,1)$ is independent of the data.

For any $\pi\in\Pi$, we have
\begin{align}
\int_{0}^{1}\sum_{t\in\{0,1\}^{q}}
{\bf 1}\{t\in\Psi_{\rm r}(\pi,u)\}\,du\,2^{-q}
~\overset{(1)}{=}~
P(S(T^{\ast},\pi)>c(\alpha,\pi)) +
P(S(T^{\ast},\pi)=c(\alpha,\pi))u(\alpha,\pi) ~\overset{(2)}{=}~\alpha,
\label{eq:validity_3}
\end{align}
where (1) holds by \eqref{eq:validity_1}, the fact that $T^{\ast}=\{T_{i}^{\ast}\}_{i=1}^{q}$ is i.i.d.\ $\mathrm{Bernoulli}(1/2)$, and $T^{\ast}\perp U$, and (2) by the definition of $u(\alpha,\pi)$ in \eqref{eq:u_defn}.

Next, consider the following argument:
\begin{align*}
|E[\phi_{n}^{\mathrm r}(q,\alpha)]-\alpha|
&~\overset{(1)}{=}~
\Big|
E\Big[
\int_{0}^{1}\sum_{t\in\{0,1\}^{q}}
{\bf 1}\{t\in\Psi_{\rm r}(\tilde\pi,u)\}du~
(P(\tilde T=t\mid\tilde\pi)-2^{-q})
\Big]
\Big|\\
&~{\leq}~
E\Big[
\int_{0}^{1}\sum_{t\in\{0,1\}^{q}}
{\bf 1}\{t\in\Psi_{\rm r}(\tilde\pi,u)\}du~
|P(\tilde T=t\mid\tilde\pi)-2^{-q}|
\Big]\\
&~\overset{(2)}{\leq}~
E\Big[
\int_{0}^{1}\sum_{t\in\{0,1\}^{q}}
{\bf 1}\{t\in\Psi_{\rm r}(\tilde\pi,u)\}du ~
\Delta_{n}(\tilde\pi)
\Big]\\
&~\overset{(3)}{=}~
\alpha 2^{q}E[\Delta_{n}(\tilde\pi)]
~\overset{(4)}{\to}~0,
\end{align*}
where (1) holds by the definition of $\Psi_{\rm r}(\tilde\pi,u)$, that $U\sim\mathrm{Unif}(0,1)$ is independent of the data, and \eqref{eq:validity_3}, (2) by \eqref{eq:validity_0}, (3) by \eqref{eq:validity_3}, and (4) by \eqref{eq:validity_0b}. This proves that the randomized test has asymptotic rejection probability $\alpha$ under $H_{0}$.

We now prove the result for the non-randomized version. For any $\pi\in\Pi$, define the following set:
\begin{equation}
\Psi_{\rm nr}(\pi)~\equiv~
\{t\in\{0,1\}^{q}:S(t,\pi)>c(\alpha,\pi)\}.
\label{eq:validity_5}
\end{equation}
By definition, the non-randomized test rejects $H_{0}$ if and only if $\tilde T\in\Psi_{\rm nr}(\tilde\pi)$.

For any $\pi\in\Pi$, note that
\begin{equation}
\sum_{t\in\{0,1\}^{q}}
{\bf 1}\{t\in\Psi_{\rm nr}(\pi)\}\,2^{-q}
~\overset{(1)}{=}~
P(S(T^{\ast},\pi)>c(\alpha,\pi))
~\overset{(2)}{\leq}~\alpha,
\label{eq:validityPf_2}
\end{equation}
where (1) holds by \eqref{eq:validity_5} and that $T^{\ast}=\{T_{i}^{\ast}\}_{i=1}^{q}$ is i.i.d.\ $\mathrm{Bernoulli}(1/2)$, and (2) by definition of $c(\alpha,\pi)$. This inequality can be strict because the distribution of $S(T^{\ast},\pi)$ is discrete.

Next, consider the following argument:
\begin{align}
E[\phi_{n}^{\mathrm{nr}}(q,\alpha)]
&~\overset{(1)}{=}~
E\Big[
\sum_{t\in\{0,1\}^{q}}
{\bf 1}\{t\in\Psi_{\rm nr}(\tilde\pi)\}
P(\tilde T=t\mid\tilde\pi)
\Big]\notag\\
&~\leq~
E\Big[
\sum_{t\in\{0,1\}^{q}}
{\bf 1}\{t\in\Psi_{\rm nr}(\tilde\pi)\}2^{-q}
\Big]+
E\Big[
\sum_{t\in\{0,1\}^{q}}
{\bf 1}\{t\in\Psi_{\rm nr}(\tilde\pi)\}
|P(\tilde T=t\mid\tilde\pi)-2^{-q}|
\Big]\notag\\
&~\overset{(2)}{\leq}~
\alpha+\alpha 2^{q}E[\Delta_{n}(\tilde\pi)],\label{eq:validityPf_3}
\end{align}
where (1) holds by definition of $\Psi_{\rm nr}(\tilde\pi)$, (2) by \eqref{eq:validity_0} and \eqref{eq:validityPf_2}. By \eqref{eq:validity_0b} and \eqref{eq:validityPf_3}, we conclude that
\begin{equation*}
\limsup_{n\to\infty}E[\phi_{n}^{\mathrm{nr}}(q,\alpha)]~\leq~\alpha.
\end{equation*}
This proves that the non-randomized test has asymptotic rejection probability no larger than $\alpha$ under $H_{0}$. The inequality in the last display can be strict because the distribution of $S(T^{\ast},\pi)$ is discrete.
\end{proof}

\subsection{Auxiliary results}

\begin{lemma}\label{lem:boundaryLocalLimit}
Let Assumption \ref{ass:assumption} and $H_0$ in \eqref{eq:H0} hold. For any $u>0$, let
\begin{equation}
P_u(\cdot)
~\equiv~
P\left(
    (Z_i,B(Z_i),T(Z_i))\in \cdot
    \mid \dist(Z_i,\mathcal B)\leq u
\right).
\label{eq:localLaw}
\end{equation}
Then, the laws $P_u$ converge weakly as $u\downarrow0$ to a probability measure $P^*$ on $\mathbb{R}^d\times \mathbb{R}^d\times\{0,1\}$, defined by
\begin{equation}
P^*(C)~=~\sum_{t \in\{0,1\}}
\frac{\int_{\mathcal{B}}\mathbf{1}\{(b,b,t)\in C\} f_Z(b)d\mathcal{H}^{d-1}(b)}{2\int_{\mathcal{B}} f_Z(b)d\mathcal{H}^{d-1}(b)}
\label{eq:limitP}
\end{equation}
for every Borel-measurable set $C\subseteq \mathbb{R}^d\times\mathbb{R}^d\times\{0,1\}$.
\end{lemma}
\begin{proof}
First, we note that Lemma \ref{lem:BoundaryHasProb} implies that $P(\dist(Z_i , \mathcal{B} ) \leq  u)>0$ for all $u>0$. Thus, $P_u$ in \eqref{eq:localLaw} is well defined for any $u>0$.

Next, we show that
\begin{equation}
\int_{\mathcal{B}}f_{Z}(b)d\mathcal{H}^{d-1}(b)~>~0.\label{eq:boundaryLocalLimit_2}
\end{equation}
Let $\delta >0$ and $b_{0}\in \mathcal{B}$ be as in Assumption \ref{ass:assumption}(a), and let $\eta_0>0$ and $\tilde r>0$ be as in part 3 of Lemma \ref{lem:regset_v1}. By Assumption \ref{ass:assumption}(a) and $H_{0}$,
$f_{Z}(b_{0})=2c>0$ for some $c>0$. By $H_{0}$, there exists $\delta _{1}\in (0,\min \{\delta ,\tilde r\})$ such that $\Vert b-b_{0}\Vert \leq \delta _{1}$ implies $|f_{Z}(b)-f_{Z}(b_{0})|\leq c$, which, combined with $f_{Z}(b_{0})=2c$, implies $f_{Z}(b)\geq c$. Then,
\begin{equation*}
\int_{\mathcal{B}}f_{Z}(b)d\mathcal{H}^{d-1}(b)
~\geq~\int_{\{b\in \mathcal{B}:\Vert b-b_{0}\Vert <\delta _{1}\}} f_{Z}(b)d\mathcal{H}^{d-1}(b)
~\overset{(1)}{\geq}~ c\mathcal{H}^{d-1}(\mathcal{B}\cap B(b_{0},\delta _{1}))
~\overset{(2)}{>}~
c\eta_0 \delta _{1}^{d-1} 
~>~0,
\end{equation*}
as desired in \eqref{eq:boundaryLocalLimit_2}, where (1) holds by $f_{Z}(b)\geq c$ for all $\Vert b-b_{0}\Vert \leq \delta _{1}$, and (2) by part 3 of Lemma \ref{lem:regset_v1}.

Also, note that
\begin{equation}
\frac{P(\dist(Z_{i},\mathcal{B})\leq u)}{u}
~\overset{(1)}{\rightarrow}~
2\int_{\mathcal{B}}f_{Z}(b)d\mathcal{H}^{d-1}(b)
~\overset{(2)}{\in}~ (0,\infty),
\label{eq:boundaryLocalLimit_3}
\end{equation}
where (1) holds by Lemma \ref{lem:LimitDen} with $S=\mathcal B$ and $t\in \{0,1\}$, and (2) by  \eqref{eq:boundaryLocalLimit_2} and \eqref{eq:LimitDen_5b} in the proof of Lemma \ref{lem:LimitDen}. Therefore, the RHS of \eqref{eq:limitP} is well defined and defines a probability measure on $\mathbb{R}^d\times\mathbb{R}^d\times\{0,1\}$.

Let $C\subseteq \mathbb{R}^{d}\times \mathbb{R}^{d}\times \{0,1\}$ be an arbitrary $P^{\ast }$-continuity set. By definition, it suffices to show that
\begin{equation}
\lim_{u\downarrow 0}P_{u}(C)=P^{\ast }(C).\label{eq:boundaryLocalLimit_4}
\end{equation}

For any Borel set $D\subseteq \mathbb{R}^{d}\times \mathbb{R}^{d}\times \{0,1\}$ and $t\in \{0,1\}$, define $D_{t}\equiv \{b\in \mathcal{B}:(b,b,t)\in D\}$. Note that
\begin{align}
& \lim_{u\downarrow 0}P\left( \{(B(Z_{i}),B(Z_{i}),T(Z_{i}))\in D\} \cap \{\dist(Z_{i},\mathcal{B})\leq u\}\right)/u\notag \\
& ~\overset{(1)}{=}~\lim_{u\downarrow 0}\sum_{t=0}^{1}
\frac{P\left( \{B(Z_{i})\in D_{t}\}\cap \{T(Z_{i})=t\} \cap \{\dist(Z_{i},\mathcal{B})\leq u\}\right)
}{u}\notag \\
& ~\overset{(2)}{=}~\sum_{t=0}^{1}\int_{D_{t}}f_{Z}(b)d\mathcal{H}^{d-1}(b).
\label{eq:boundaryLocalLimit_5}
\end{align}
where (1) holds because the event
$\{(B(Z_i),B(Z_i),T(Z_i))\in D\}$ can be written as the disjoint union, over $t\in\{0,1\}$, of the events
$\{B(Z_i)\in D_t\}\cap\{T(Z_i)=t\}$, and (2) by Lemma \ref{lem:LimitDen} applied to $S=D_t$, which is a Borel subset of $\mathcal B$ because the map $b\mapsto (b,b,t)$ is continuous.

By \eqref{eq:boundaryLocalLimit_3} and \eqref{eq:boundaryLocalLimit_5}, we obtain
\begin{equation}
\lim_{u\downarrow 0}
P\left( (B(Z_{i}),B(Z_{i}),T(Z_{i}))\in D
\mid \dist(Z_{i},\mathcal{B})\leq u\right)
~=~
P^{\ast }(D).
\label{eq:boundaryLocalLimit_6}
\end{equation}

Next, for any $\eta >0$, define $C^{\eta }\equiv \{x:\dist(x,C)<\eta \}$ and $C^{-\eta }\equiv \{x:\dist(x,C^{c})>\eta \}$. Note that $C^{-\eta}\subseteq C\subseteq C^{\eta }$. To see this, note that if $x\in C$, then $\dist(x,C)<\eta$, and so $x\in C^{\eta}$. In turn, if $x\in C^{-\eta}$, then $\dist(x,C^{c})>\eta$, and so $x\notin C^{c}$, and therefore $x\in C$.

For any $u\in (0,\eta )$, consider the following argument:
\begin{align}
&
\left\{(B(Z_{i}),B(Z_{i}),T(Z_{i}))\in C^{-\eta }\right\}
\cap
\{\dist(Z_{i},\mathcal{B})\leq u\}
\notag \\
&\overset{(1)}{=}
\left\{\dist((B(Z_{i}),B(Z_{i}),T(Z_{i})),C^{c})>\eta\right\}
\cap
\{\dist(Z_{i},\mathcal{B})\leq u\}
\notag \\
&\overset{(2)}{=}
\left\{
\begin{array}{c}
\left\{\dist((B(Z_{i}),B(Z_{i}),T(Z_{i})),C^{c})>\eta\right\}
\cap
\{\dist(Z_{i},\mathcal{B})\leq u\}\\
\cap
\left\{
\Vert (Z_{i},B(Z_{i}),T(Z_{i}))-(B(Z_{i}),B(Z_{i}),T(Z_{i}))\Vert \leq u
\right\}
\end{array}\right\}
\notag \\
&\overset{(3)}{\subseteq}
\left\{(Z_{i},B(Z_{i}),T(Z_{i}))\in C\right\}
\cap
\{\dist(Z_{i},\mathcal{B})\leq u\}
\cap
\left\{
\Vert (Z_{i},B(Z_{i}),T(Z_{i}))-(B(Z_{i}),B(Z_{i}),T(Z_{i}))\Vert \leq u
\right\}
\notag \\
&\overset{(4)}{\subseteq}
\left\{\dist((B(Z_{i}),B(Z_{i}),T(Z_{i})),C)<\eta\right\}
\cap
\{\dist(Z_{i},\mathcal{B})\leq u\}
\notag \\
&\overset{(5)}{=}
\left\{(B(Z_{i}),B(Z_{i}),T(Z_{i}))\in C^{\eta}\right\}
\cap
\{\dist(Z_{i},\mathcal{B})\leq u\},\label{eq:boundaryLocalLimit_7}
\end{align}
where (1) holds by the definition of $C^{-\eta}$, (2) because
$\{\dist(Z_{i},\mathcal{B})\leq u\}$ implies
$\{\dist(Z_{i},B(Z_{i}))\leq u\}$, which, in turn, implies $\{
\Vert (Z_{i},B(Z_{i}),T(Z_{i}))-(B(Z_{i}),B(Z_{i}),T(Z_{i}))\Vert \leq u
\}$, (3) because
$\{\dist((B(Z_{i}),B(Z_{i}),T(Z_{i})),C^{c})>\eta\}$, $\{\Vert (Z_{i},B(Z_{i}),T(Z_{i}))-(B(Z_{i}),B(Z_{i}),T(Z_{i}))\Vert \leq u\}$, and $\eta>u$ imply $(Z_{i},B(Z_{i}),T(Z_{i}))\in C$, (4) because
$\{(Z_{i},B(Z_{i}),T(Z_{i}))\in C\}$,
$\{\Vert (Z_{i},B(Z_{i}),T(Z_{i}))-(B(Z_{i}),B(Z_{i}),T(Z_{i}))\Vert \leq u\}$,
and $\eta>u$ imply
$\{\dist((B(Z_{i}),B(Z_{i}),T(Z_{i})),C)<\eta\}$, and (5) holds by the definition of $C^\eta$.

By \eqref{eq:boundaryLocalLimit_7} and $P(\dist(Z_{i},\mathcal{B})\leq u)>0$, we obtain
\begin{align}
&P\left((B(Z_{i}),B(Z_{i}),T(Z_{i}))\in C^{-\eta} \mid \dist(Z_{i},\mathcal{B})\leq u\right)\notag\\
&\leq~P_{u}(C)\notag\\
&\leq~P\left((B(Z_{i}),B(Z_{i}),T(Z_{i}))\in C^{\eta} \mid \dist(Z_{i},\mathcal{B})\leq u\right).\label{eq:boundaryLocalLimit_8}
\end{align}
Taking limits as $u\downarrow 0$ with \eqref{eq:boundaryLocalLimit_8} and applying
\eqref{eq:boundaryLocalLimit_6} with $D=C^{-\eta}$ and $D=C^\eta$ yields
\begin{equation}
P^{\ast }(C^{-\eta}) ~\leq~ \liminf_{u\downarrow 0}P_{u}(C) ~\leq~ \limsup_{u\downarrow 0}P_{u}(C) ~\leq~ P^{\ast }(C^{\eta}).\label{eq:boundaryLocalLimit_9}
\end{equation}

Then, consider the following argument:
\begin{equation}
\lim_{\eta \downarrow 0}P^{\ast }(C^{-\eta})~\overset{(1)}{=}~ P^{\ast }(\operatorname{int}(C)) ~\overset{(2)}{=}~ P^{\ast }(C) ~\overset{(3)}{=}~ P^{\ast }(\operatorname{cl}(C)) ~\overset{(4)}{=}~ \lim_{\eta \downarrow 0}P^{\ast }(C^{\eta}),\label{eq:boundaryLocalLimit_10}
\end{equation}
where (1) holds by $C^{-\eta}\uparrow \operatorname{int}(C)$ as
$\eta\downarrow 0$ and continuity from below of probability measures, (2) and (3) hold because $C$ is a $P^{\ast}$-continuity set, and so $P^{\ast}(\partial C)=0$, and (4) holds by $C^{\eta}\downarrow \operatorname{cl}(C)$ as $\eta\downarrow 0$ and continuity from above of probability measures. 

Taking limits as $\eta \downarrow 0$ with \eqref{eq:boundaryLocalLimit_9} and applying
\eqref{eq:boundaryLocalLimit_10}, we obtain \eqref{eq:boundaryLocalLimit_4}, as desired.
\end{proof}

\begin{lemma}\label{lem:BoundaryHasProb} 
Under Assumption \ref{ass:assumption}(a) and $H_{0}$ in \eqref{eq:H0}, $P(\dist(Z_i , \mathcal{B} ) \leq  u)>0$ for all $u>0$.
\end{lemma}
\begin{proof}
Let $\delta>0$ and $b_0\in \mathcal{B}$ be as in Assumption \ref{ass:assumption}(a). By this assumption and $H_{0}$, $f_{Z}(b_{0})=2c>0$ for some $c>0$. By $H_{0}$, $\exists \delta _{1}\in (0,\delta )$ such that $\Vert z-b_{0}\Vert \leq \delta _{1}$ implies $|f_{Z}(z)-f_{Z}(b_{0})|\leq c$ which, combined with $f_{Z}(b_{0})=2c$ implies $f_{Z}(z)\geq c$.
For all $u \in (0,\delta_1]$, 
\begin{align*}
    P(\dist(Z_i , \mathcal{B} ) \leq u) ~=~
    \int_{\dist(z, \mathcal{B} ) \leq u} f_{Z}( z) dz~\overset{(1)}{\geq}~ \int_{B(b_0,u)}  f_{Z}( z) dz
    ~\overset{(2)}{\geq}~ c\mathcal{L} ( B(b_0,u))
    ~>~0,
\end{align*}
where (1) holds by $B(b_0,u) \subseteq \{\dist(z, \mathcal{B} )\leq u\}$, and (2) by $u \in (0,\delta_1]$, and so $ f_{Z}(z)\geq c$ for all $\Vert z-b_{0}\Vert \leq\delta _{1}$. Since $P(\dist(Z_i , \mathcal{B} ) \leq u)$ is increasing in $u$, the equation implies the desired result.
\end{proof}

\begin{lemma}\label{lem:regset_v1} 
Assume Assumptions \ref{ass:assumption}(b)-(c). Then:
\begin{enumerate}
\item For $t \in \{0,1\}$, $\{ T =t\} $ has locally finite perimeter in the sense of \citet[page 122]{maggi:2012}.
\item For $t \in \{0,1\}$, 
$\mathcal{H}^{d-1}(\mathcal{B}\setminus \partial ^{\ast }\{T =t\})~=~0$, where $\partial ^{\ast }A$ denotes the reduced boundary of $A$.
\item There exist constants $\eta >0$ and $\tilde{r}>0$ such that for every $b\in \mathcal{B}$ and $r\in (0,\tilde{r})$, 
\begin{equation}
\mathcal{H}^{d-1}(\mathcal{B}\cap B(b,r) )~\geq~ \eta \times r^{d-1}.\label{eq:lowerperdensity}
\end{equation}
\end{enumerate}
\end{lemma}
\begin{proof}
We fix $t\in \{ 0,1\} $ arbitrarily throughout this proof. 

\medskip
\noindent \underline{Part 1.} By Federer's criterion for finite perimeter (\citet[Theorem 4.5.11, page 506]{federer:1996}), it suffices to show that $\mathcal{H}^{d-1}(\partial ^{e}\{ T=t\}  \cap B({\bf{0}}_d, r) )<\infty $ for any $r>0$, where $\partial^e A$ denotes the essential boundary of a set $A$. This, in turn, follows from $\partial ^{e}\{ T=t\} \subseteq bd\{ T=t\} =\mathcal{B}$ and Assumption \ref{ass:assumption}(b). 


\medskip
\noindent \underline{Part 2.} First, we show that 
\begin{equation}
\mathcal{B}~=~\partial ^{e}\{T=t\}.
\label{eq:regset_econ_1}
\end{equation}
By $\partial ^{e}\{ T=t\} \subseteq bd\{ T=t\} =:\mathcal{B}$, it suffices to show that $\mathcal{B}\subseteq\partial ^{e}\{ T=t\} $. To see this, pick $b\in \mathcal{B}$ arbitrarily. By Assumption \ref{ass:assumption}(c),
\begin{align}
    \underset{r\downarrow 0}{\lim \inf }~\frac{\mathcal{L}( \{T=t\}\cap B( b,r) ) }{\mathcal{L}( B( b,r) ) } ~\geq~ \gamma ~>~0.
    \label{eq:regset_econ_1b}
\end{align}
By \eqref{eq:regset_econ_1b}, $b$ does not belong to the measure theoretic exterior of $\{T=t\}$, i.e., $b\not\in \{ T=t\}^{(0)}$. In turn, note that
\begin{align}
    \underset{r\downarrow 0}{\lim \sup }~\frac{\mathcal{L}( \{T=t\}\cap B( b,r) ) }{\mathcal{L}( B( b,r) ) } ~\overset{(1)}{=}~1 - \underset{r\downarrow 0}{\lim \inf }~\frac{\mathcal{L}( \{T=1-t\}\cap B( b,r) ) }{\mathcal{L}( B( b,r) ) }  ~\overset{(2)}{\leq }~ 1- \gamma ~<~1,
    \label{eq:regset_econ_1c}
\end{align}
where (1) holds by $\{T=1\}\cap B( b,r)$ and $\{T=0\}\cap B( b,r)$ are a partition of $B( b,r)$, and so $\mathcal{L}( \{T=1\}\cap B( b,r) ) + \mathcal{L}( \{T=0\} \cap B( b,r) ) =\mathcal{L}(  B( b,r) ) $ for all $r>0$, (2) by Assumption \ref{ass:assumption}(c) applied to $\tilde t =1-t$. By \eqref{eq:regset_econ_1c}, $b$ does not belong to the measure theoretic interior of $\{T=t\}$, i.e., $b\not\in \{ T=t\}^{(1)}$. Since $b\not\in \{ T=t\}^{(0)} \cup \{ T=t\}^{(1)}$, we have that $b$ belongs to the measure-theoretic boundary of $\{T=t\}$, i.e., $b\in \partial ^{e}\{T=t\}$. Since the choice of $b\in \mathcal{B}$ was arbitrary, $\mathcal{B} \subseteq \partial ^{e}\{T=t\}$, as desired.

By part 1, $\{ T=t\} $ has locally finite perimeter. Then, Federer's Theorem (see \cite[Theorem 16.2]{maggi:2012}) implies that $\partial ^{\ast }\{T=t\} \subseteq \partial ^{e}\{T=t\}$ and
\begin{equation}
\mathcal{H}^{d-1}(\partial ^{e}\{T=t\}\setminus \partial ^{\ast }\{T=t\})~=~0.
\label{eq:regset_econ_2}
\end{equation}

The desired result follows immediately from combining \eqref{eq:regset_econ_1} and \eqref{eq:regset_econ_2}. To see this, note that
\begin{equation*}
\mathcal{H}^{d-1}( \mathcal{B}\setminus \partial ^{\ast }\{T=t\})  ~\overset{(1)}{=}~\mathcal{H}^{d-1}( \partial ^{e}\{T=t\}\setminus \partial ^{\ast }\{T=t\}) ~\overset{(2)}{=}~0, 
\end{equation*}
where (1) holds by \eqref{eq:regset_econ_1} and (2) by \eqref{eq:regset_econ_2}.

\noindent \underline{Part 3.} By Assumption \ref{ass:assumption}(c), there exist constants
$\gamma>0$ and $\bar r>0$ such that, for all $b\in\mathcal B$ and all
$r\in(0,\bar r)$,
\begin{equation}
\min \Big\{ \frac{\mathcal{L}(\{T=1\}\cap B(b,r))}{\mathcal{L}(B(b,r))},\frac{\mathcal{L}(\{T=0\}\cap B(b,r))}{\mathcal{L}(B(b,r))}\Big\} ~>~\gamma/2.
\label{eq:regset_econ_3}
\end{equation}

Then, consider the following derivation for any $r\in ( 0,\tilde{r}) $, 
\begin{align*}
\mathcal{H}^{d-1}(\mathcal{B}\cap B(b,r))~&\overset{(1)}{\geq }~\max\{\mathcal{H}^{d-1}(\partial ^{\ast }\{T =0\}\cap B(b,r)), \mathcal{H}^{d-1}(\partial ^{\ast }\{T =1\}\cap B(b,r))\} \\
&\overset{(2)}{=}~\max\{\mathcal{P}(\{T=0\};B(b,r)),\mathcal{P}(\{T=1\};B(b,r))\} \\
&\overset{(3)}{\geq}~
\left(\begin{array}{c}
\mathcal{P}(\{T=1\};B(b,r)) \times {\bf 1}\{ \mathcal{L}(\{T=1\}\cap B(b,r))\leq \mathcal{L}(B(b,r))/2\}\\
+\mathcal{P}(\{T=0\};B(b,r)) \times {\bf 1}\{ \mathcal{L}(\{T=0\}\cap B(b,r))\leq \mathcal{L}(B(b,r))/2\}
\end{array}\right) \\
&\overset{(4)}{\geq }~c( d,1/2) \times \min \{ \mathcal{L}(\{T=1\}\cap B(b,r)),\mathcal{L}(\{T=0\}\cap B(b,r))\} ^{(d-1) /d} \\
&\overset{(5)}{\geq }~
\eta \times r^{d-1},
\end{align*}
as desired, where (1) holds by $\partial ^{\ast }\{T=t\}\subseteq \partial ^{e }\{T=t\}$, which, in turn, follows from \citet[Corollary 15.8]{maggi:2012}, and $\partial ^{e }\{T=t\}\subseteq bd\{T=t\} = \mathcal{B}$, (2) by part 1 and \citet[Remark 12.2 and Theorem 15.9]{maggi:2012} to $E=\{ T=t\}$ and $F=B(b,r):=\{ z\in \mathbb{R} ^{d}:\Vert z-b\Vert \leq r \} $,
where $\mathcal{P}( E;F) $ denotes the relative perimeter of $E$ in $F\subseteq \mathbb{R} ^{d}$, (3) by the fact that either $\mathcal{L}(\{T=0\}\cap B(b,r))\leq \mathcal{L}(B(b,r))/2$ or $\mathcal{L}(\{T=1\}\cap B(b,r))\leq \mathcal{L}(B(b,r))/2$, which, in turn, follows from $\mathcal{L}( \{T=1\}\cap B( b,r) ) + \mathcal{L}( \{T=0\} \cap B( b,r) ) =\mathcal{L}(  B( b,r) ) $, (4) by 
part 1 and \citet[Proposition 12.37]{maggi:2012} applied to $E=\{ T=0\}$ and $E=\{ T=1\}$, and (5) by \eqref{eq:regset_econ_3}, $\mathcal{L}(B(b,r))=\mathcal{L}(B({\bf 0}_{d},1))r^{d}$, and by setting $\eta \equiv c( d,1/2) ( \gamma /2) ^{( d-1) /d}( \mathcal{L}(B({\bf 0}_{d},1))) ^{( d-1) /d}>0$.
\end{proof}

\begin{lemma}\label{lem:LimitDen}
Let $Bo(\mathcal{B})$ denote the class of Borel-measurable subsets of $\mathcal{B}$. Under Assumption \ref{ass:assumption} and $H_0$ in \eqref{eq:H0}, we have that
\begin{equation}\label{eq:LimitDen_0}
\lim_{u \downarrow 0 } \sup_{S \in Bo(\mathcal{B}), t\in\{0,1\}} \left|{P(\{B(Z_{i})\in S\}\cap \{\dist(Z_i,\mathcal{B}) \leq  u\}\cap  \{T(Z_{i}) = t\})}/{u} - \int_{S}f_{Z}(b)d\mathcal{H}^{d-1}(b)\right| ~=~0.
\end{equation}
\end{lemma}
\begin{proof}
Fix $\varepsilon >0$ arbitrarily. It suffices to show that for all sufficiently small $u>0$,
\begin{equation}
\sup_{S \in Bo(\mathcal{B}), t\in\{0,1\}}\Big\vert P(\{B(Z_{i})\in S\}\cap \{\dist(Z_{i},\mathcal{B} ) \leq  u\}\cap \{T(Z_{i})=t\})/u-\int_{S}f_{Z}( b) d\mathcal{H}^{d-1}( b) \Big\vert ~\leq~ \varepsilon .
\label{eq:LimitDen_1}
\end{equation}
We divide the argument into five parts. 

\medskip\noindent
\underline{Part 1.} This part introduces constants and sets, and bounds the probability for several ``bad'' events.

First, by Assumption \ref{ass:assumption}(e), $\exists r_{1}=r_{1}( \varepsilon )>0 $ and $\exists u_{1}=u_{1}( \varepsilon ,r_{1}) \in (0,r_1)$ such that for all $u\in ( 0,u_{1}) $,
\begin{equation}
P(\{\dist(Z_{i},\mathcal{B})\leq u\}\cap \{Z_{i} \not\in B({\bf 0}_{d},r_{1}) \})/u~\leq~ \varepsilon/12. \label{eq:LimitDen_2}
\end{equation}
Then, for all $u\in ( 0,u_{1}) $,
\begin{equation}
P(\{\dist(Z_{i},\mathcal{B}) \leq u\} \cap \{B(Z_{i}) \not\in B({\bf 0}_{d},2r_{1}) \})/u~\overset{(1)}{\leq }~P(\{\dist(Z_{i},\mathcal{B}) \leq u\}\cap \{Z_{i} \not\in B({\bf 0}_{d},r_{1}) \})/u~\overset{(2)}{\leq}~  \varepsilon/12 ,
\label{eq:LimitDen_3}
\end{equation}
where (1) holds by $u <r_1$ and so $\{ \{ \dist(Z_{i},\mathcal{B}) \leq u\}\cap \{B(Z_{i})\not\in  B({\bf 0}_{d},2r_{1}) \}\} \subseteq  \{Z_{i}\not\in  B({\bf 0}_{d},r_{1}) \} $, and (2) by \eqref{eq:LimitDen_2}.
Moreover, for all $u\in ( 0,u_{1}) $,
\begin{equation}
P(\{\dist(Z_{i},\mathcal{B}\setminus B({\bf 0}_{d},2r_{1})) \leq u\})/u~\overset{(1)}{\leq }~P(\{\dist(Z_{i},\mathcal{B}) \leq u\}\cap \{Z_{i}\not\in  B({\bf 0}_{d},r_{1})\})/u~\overset{(2)}{\leq }~ \varepsilon/12, \label{eq:LimitDen_3b}
\end{equation}
where (1) holds by $u <r_1$ and so $\{\dist(Z_{i},\mathcal{B}\setminus B({\bf 0}_{d},2r_{1}))\leq u\}\subseteq \{\dist(Z_{i},\mathcal{B}) \leq u\}\cap \{Z_{i}\not\in  B({\bf 0}_{d},r_{1}) \}$ and (2) by \eqref{eq:LimitDen_2}.
%
%
Also, by Assumption \ref{ass:assumption}(e), $\exists r_{2}>0$ s.t.
\begin{equation}
\int_{\mathcal{B}\setminus  B({\bf 0}_{d},r_{2})}f_{Z}( b) d\mathcal{H}^{d-1}( b) ~\leq~ \varepsilon/12 .
\label{eq:LimitDen_5}
\end{equation}
For the rest of the proof, set $r_{3}\equiv  \max \{ 2r_{1}+2,r_{2}\} $ and $B_{3} \equiv B({\bf 0}_{d},r_{3})$.

Next, consider the following auxiliary result: 
\begin{align}
\int_{\mathcal B} f_Z(b)~d\mathcal H^{d-1}(b)
~\leq~\sup_{b\in\mathcal B} f_Z(b)\mathcal H^{d-1}(\{ b \in \mathcal B : \Vert b\Vert\leq {r}\})+\int_{{b\in\mathcal B:\Vert b\Vert>{r}}} f_Z(b)~d\mathcal H^{d-1}(b) ~\overset{(1)}{<}~\infty ,\label{eq:LimitDen_5b}
\end{align}
where (1) holds for a sufficiently large $r>0$ by Assumptions \ref{ass:assumption}(a), (b), and (e).

We now construct a sequence of sets $\{\mathcal M_j\}_{j\in\mathbb N}$ with certain desirable properties. Let $k:=\max\{2,d-1\}$. By Assumption \ref{ass:assumption}(d), $\mathcal B\setminus\Sigma$ is an embedded $C^k$, $(d-1)$-dimensional manifold without boundary. Hence, $\mathcal B\setminus\Sigma$ admits a proper $C^k$ exhaustion function $\rho:\mathcal B\setminus\Sigma\to[0,\infty)$. Since $\dim(\mathcal B\setminus\Sigma)=d-1$ and $k>d-2$, the finite-differentiability version of the Morse--Sard theorem \citep[p.~67]{hirsch:1976} implies that the set of critical values of $\rho$ has Lebesgue measure zero. We can therefore choose a strictly increasing sequence of regular values $\{c_j\}_{j\in\mathbb N}$ such that $c_j\to\infty$. For each $j\in\mathbb N$, define
\begin{equation*}
\mathcal M_j:=\{b\in\mathcal B\setminus\Sigma:\rho(b)\leq c_j\}.
\end{equation*}
Because $\rho$ is proper, $\mathcal M_j$ is compact in $\mathcal B\setminus\Sigma$. Since the inclusion $\mathcal B\setminus\Sigma\hookrightarrow\mathbb R^d$ is continuous, $\mathcal M_j$ is also compact in $\mathbb R^d$.

Because $c_j$ is a regular value of $\rho$, the regular sublevel set theorem implies that $\mathcal M_j$ is an embedded $C^k$, $(d-1)$-dimensional manifold with manifold boundary, with $\partial_{\mathrm{man}}\mathcal M_j=\{b\in\mathcal B\setminus\Sigma:\rho(b)=c_j\}$ and $\operatorname{Int}_{\mathrm{man}}(\mathcal M_j)=\{b\in\mathcal B\setminus\Sigma:\rho(b)<c_j\}$. 
In particular, because $k\geq2$, $\mathcal M_j$ is an embedded $C^2$ manifold with boundary. Since $c_j<c_{j+1}$, we have $\mathcal M_j\subseteq\operatorname{Int}_{\mathrm{man}}(\mathcal M_{j+1})$. Furthermore, because $c_j\to\infty$, the sequence $\{\mathcal M_j\}_{j\in\mathbb N}$ exhausts $\mathcal B\setminus\Sigma$. Therefore, $\bigcup_{j=1}^{\infty}\mathcal M_j=\mathcal B\setminus\Sigma$ and $(\mathcal B\setminus\mathcal M_j)\downarrow\Sigma$ as $j\to\infty.$

Note that
\begin{equation*}
    0~\leq~ \int_{\Sigma}f_{Z}(b)d\mathcal{H}^{d-1}(b) ~\overset{(1)}{\leq }~ \Big( \sup_{b\in \mathcal{B} }f_{Z}(b)\Big) \times  \mathcal{H}^{d-1}(\Sigma)~\overset{(2)}{=}~0,
\end{equation*}
where (1) holds by $\Sigma \subseteq \mathcal{B}$ and (2) by Assumptions \ref{ass:assumption}(b) and (d). Then,
\begin{equation}
    \int_{\Sigma}f_{Z}(b)d\mathcal{H}^{d-1}(b)~=~0.
    \label{eq:LimitDen_6}
\end{equation}
Also,
\begin{equation}
   \int_{\mathcal{B} \setminus \mathcal{M}_j}f_{Z}(b)d\mathcal{H}^{d-1}(b)~\leq~ \int_{\mathcal{B}}f_{Z}(b)d\mathcal{H}^{d-1}(b)~\overset{(1)}{<}~\infty,
   \label{eq:LimitDen_7}
\end{equation}
where (1) holds by \eqref{eq:LimitDen_5b}. Then,
\begin{equation}
\lim_{j \to \infty }\int_{\mathcal{B} \setminus \mathcal{M}_j}f_{Z}(b)d\mathcal{H}^{d-1}(b)~\overset{(1)}{=}~\int_{\Sigma}f_{Z}(b)d\mathcal{H}^{d-1}(b)~\overset{(2)}{=}~0,
  \label{eq:LimitDen_8}
\end{equation}
where (1) holds by $(\mathcal{B} \setminus \mathcal{M}_j) \downarrow \Sigma$ as $j\to \infty$, \eqref{eq:LimitDen_7}, and dominated convergence, and (2) by \eqref{eq:LimitDen_6}. 

Also, note that
\begin{equation}
\mathcal{H}^{d-1}(( \mathcal{B}\cap B({\bf 0}_{d},r_{3}+1))\setminus \mathcal{M}_j)  ~\leq~ \mathcal{H}^{d-1}(\mathcal{B} \cap B({\bf 0}_{d},r_{3}+1) )~\overset{(1)}{<}~\infty.
   \label{eq:LimitDen_9}
\end{equation}
where (1) holds by Assumption \ref{ass:assumption}(b). Then,
\begin{equation}
\lim_{j \to \infty }\mathcal{H}^{d-1}(( \mathcal{B}\cap B({\bf 0}_{d},r_{3}+1))\setminus \mathcal{M}_j)  ~\overset{(1)}{=}~\mathcal{H}^{d-1}(\Sigma \cap B({\bf 0}_{d},r_{3}+1) )~\overset{(2)}{=}~0,
  \label{eq:LimitDen_10}
\end{equation}
where (1) holds by $\mathcal{B} \setminus \mathcal{M}_j \downarrow \Sigma$ as $j\to \infty$, \eqref{eq:LimitDen_9}, and dominated convergence, and (2) by Assumption \ref{ass:assumption}(d).

By \eqref{eq:LimitDen_5b}, we can choose $\lambda _{1}\in ( 0,1) $ small enough to get
\begin{equation}
2 \lambda _{1}\int_{\mathcal{B}}f_{Z}( b) d\mathcal{H}^{d-1}( b) ~\leq~ \varepsilon /12.
\label{eq:LimitDen_32}
\end{equation}

By \eqref{eq:LimitDen_8} and \eqref{eq:LimitDen_10}, we can find $j=j(\varepsilon )$ and a set $\mathcal{M}=\mathcal{M}_{j(\varepsilon )}$ so that
\begin{align}
\int_{\mathcal{B} \setminus \mathcal{M}}f_{Z}(b)d\mathcal{H}^{d-1}(b)~&\leq~ \varepsilon/12,  \label{eq:LimitDen_11}\\
\mathcal{H}^{d-1}(( \mathcal{B}\cap B({\bf 0}_{d},r_{3}+1))\setminus \mathcal{M}) ~&\leq~  \min\left\{ \frac{\eta  \lambda_1^{d-1} }{\xi ( d) \mathcal{L}( B({\bf 0}_{d},1) )3^{d}}, \frac{9}{40}\right\}\frac{\varepsilon}{12\sup_{b\in \mathcal{B}^{\delta }}f_{Z}(b)}.
  \label{eq:LimitDen_12}
\end{align}
where $\eta>0$ is the constant in part 3 of Lemma \ref{lem:regset_v1}, $\lambda_1$ as in \eqref{eq:LimitDen_32}, and $\xi (d)$ is the (positive and finite) Besicovich covering number in $d$ dimensions (see \citet[Theorem 5.1]{maggi:2012}). 

\medskip\noindent
\underline{Part 2.} This part establishes the tubular neighborhood representation used in the remainder of the proof. 

Part 1 defines the compact manifold with boundary $\mathcal M=\mathcal M_j\subseteq\mathcal B\setminus\Sigma$. Since $\mathcal B$ separates the regions $\{T=0\}$ and $\{T=1\}$, its regular part $\mathcal B\setminus\Sigma$ is two-sided. Without loss of generality, suppose that $T(b)=1$ for every $b\in\mathcal B$, and let $\nu$ be the $C^1$ unit normal field on $\mathcal B\setminus\Sigma$ oriented toward the side on which $T=1$. (The argument in the other case is completely analogous.)

By construction, $\mathcal M=\{b\in\mathcal B\setminus\Sigma:\rho(b)\leq c_j\} \subseteq\operatorname{Int}_{\mathrm{man}}(\mathcal M_{j+1}) =\{b\in\mathcal B\setminus\Sigma:\rho(b)<c_{j+1}\}$. Choose a relatively open neighborhood $V$ of $\mathcal M$ in
$\mathcal B\setminus\Sigma$ such that 
\begin{equation*}
\mathcal M ~\subseteq~ V ~\subseteq~\operatorname{cl}(V) ~\subseteq~ \operatorname{Int}_{\mathrm{man}}(\mathcal M_{j+1}).
\end{equation*}
For example, we may take $V = \{ b\in\mathcal B\setminus\Sigma: \rho(b)<{(c_j+c_{j+1})}/{2}\}$. The closure $\operatorname{cl}(V)$ is compact because $\rho$ is proper. Since $D\nu$ is continuous, $\exists C<\infty$ s.t.
\begin{equation}
\sup_{b\in\operatorname{cl}(V)}\|D\nu(b)\| ~\leq~ C.
\label{eq:LimitDen_19a}
\end{equation}

By the tubular neighborhood theorem (see \citet[Theorem 5.1]{hirsch2012differential}), after reducing $V$ if necessary while preserving $\mathcal M\subseteq V$, there exists $w_0\in(0,\delta)$ s.t.\
\begin{equation*}
\Phi:V\times(-w_0,w_0)\longrightarrow\mathbb R^d, \qquad \Phi(b,w):=b+w\nu(b),
\end{equation*}
is a $C^1$ diffeomorphism onto its image $U:= \{b+w\nu(b):b\in V,\ |w|<w_0\}$. The set $U$ is open in $\mathbb R^d$ because $\Phi$ is a $C^1$ local diffeomorphism and hence an open map. Also, note that $\Phi(b,0)=b \in U$ for every $b\in V$ and so $\mathcal M\subseteq V \subset U$. After reducing $w_0$ if necessary, the chosen orientation implies that $T(\Phi(b,w)) = {\bf 1}[w\geq0]$ for every $b\in V$ and $|w|<w_0$. 

Since $U$ is open and contains the compact set $\mathcal M$, $d_1:=\dist(\mathcal M,U^c)>0$. Moreover, since $\mathcal M$ is compactly contained in the relatively open set $V\subseteq\mathcal B$, $d_2:=\dist(\mathcal M,\mathcal B\setminus V)>0.$ Set $u_2 \in(0,\min\{w_0,d_1,{d_2}/{2}\})$ for the rest of the proof. Then, for every $u\in(0,u_2)$, 
\begin{equation}
\{z\in\mathbb R^d:\dist(z,\mathcal M)\leq u\} ~\overset{(1)}{\subseteq}~U, \label{eq:intcont}
\end{equation}
where (1) holds because $\dist(z,\mathcal M)\leq u$ and $z\in U^c$ would imply $d_1 = \dist(\mathcal M,U^c) \leq \dist(z,\mathcal M) \leq u < u_2 \leq d_1$, which is a contradiction.

Fix $u\in(0,u_2)$ and consider any $z$ s.t.\ $\dist(z,\mathcal M)\leq u$. By \eqref{eq:intcont} and the fact that $\Phi$ is a diffeomorphism onto $U$, $z$ has a unique representation 
\begin{equation}
z~=~\Phi(b,w)\quad\text{for some } b\in V\text{ and }|w|<w_0.
\label{eq:LimitDen_19aa}
\end{equation}
The nearest point property of the tubular neighborhood implies that $b$ is the unique closest point to $z$ in $V$. Next, we show that $b$ is the unique closest point to $z$ in $\mathcal B$. To this end, note that for any $\widetilde b\in\mathcal B\setminus V$, we have
\begin{align}
\|z-\widetilde b\|~\overset{(1)}{\geq}~ \dist(\mathcal M,\mathcal B\setminus V) - \dist(z,\mathcal M) ~\overset{(2)}{\geq}~ d_2-u ~\overset{(3)}{>}~ u \geq \dist(z,\mathcal M) ~\overset{(4)}{\geq}~ \dist(z,V) ~\overset{(5)}{=}~ \|z-b\|, \label{eq:LimitDen_19b}
\end{align}
where (1) holds by $\dist(\mathcal M,\mathcal B\setminus V) \leq \dist(z,\mathcal M)+\dist(z,\mathcal B\setminus V)$ and $\dist(z,\mathcal B\setminus V)\leq\|z-\widetilde b\|$, (2) by $d_2=\dist(\mathcal M,\mathcal B\setminus V)$ and $\dist(z,\mathcal M)\leq u$, (3) by $u<u_2\leq d_2/2$, (4) by $\mathcal M\subseteq V$, and (5) because $b$ is the unique closest point to $z$ in $V$. Therefore, every point of $\mathcal B\setminus V$ is strictly farther from $z$ than $b$. Since $b$ is the unique closest point to $z$ in $V$, it follows that $b$ is the unique closest point to $z$ in $\mathcal B$. In particular,
\begin{equation}
B(z)~=~b\qquad\text{and}\qquad \dist(z,\mathcal B)~=~|w|. \label{eq:LimitDen_19c}
\end{equation}

For $u>0$ and $t\in\{0,1\}$, define
\begin{equation*}
I_t(u) ~:=~ 
\begin{cases}
[0,u], & t=1,\\
[-u,0), & t=0.
\end{cases}
\end{equation*}

We now prove the desired tubular neighborhood representation. Specifically, for any $A\in Bo(\mathcal B)$ with $A\subseteq V$, $t\in\{0,1\}$, and $u\in(0,u_2)$, we argue that 
\begin{align}
&\left\{ z\in\mathbb R^d: B(z)\in A,\  \dist(z,\mathcal M\cap B_3)\leq u,\ T(z)=t \right\} \notag\\
&=~ \left\{ \Phi(b,w): b\in A,\ w\in I_t(u),\ \dist(\Phi(b,w),\mathcal M\cap B_3)\leq u \right\}.
\label{eq:LimitDen_20a}
\end{align}

We now show \eqref{eq:LimitDen_20a}. First, we show that the LHS of \eqref{eq:LimitDen_20a} is included in the RHS of \eqref{eq:LimitDen_20a}. To this end, consider any $z$ in the LHS of \eqref{eq:LimitDen_20a}, i.e., $z\in\mathbb R^d$ s.t.\ $B(z)\in A$, $\dist(z,\mathcal M\cap B_3)\leq u$, and $T(z)=t$. Since $\mathcal M\cap B_3\subseteq\mathcal M$, we have $\dist(z,\mathcal M)\leq\dist(z,\mathcal M\cap B_3)\leq u$. Hence, by \eqref{eq:intcont}, $z\in U$. Thus, by \eqref{eq:LimitDen_19aa}, $z=\Phi(b,w)$ for a unique $b\in V$ and a unique $w\in(-w_0,w_0)$. By \eqref{eq:LimitDen_19c} and $B(z)\in A$, we obtain $b=B(z)\in A$ and $|w|=\dist(z,\mathcal B)$. Since $\mathcal M\cap B_3\subseteq\mathcal B$ and $z=\Phi(b,w)$, we have $|w|=\dist(z,\mathcal B)\leq\dist(z,\mathcal M\cap B_3) =\dist(\Phi(b,w),\mathcal M\cap B_3)\leq u$. Finally, $|w|\leq u$, $T(z)=t$, and the orientation of $\nu$ imply that $w\in I_t(u)$. This shows that $z$ is in the RHS of \eqref{eq:LimitDen_20a}.

Second, we show that the RHS of \eqref{eq:LimitDen_20a} is included in the LHS of \eqref{eq:LimitDen_20a}. To this end, consider any $z$ in the RHS of \eqref{eq:LimitDen_20a}, i.e., $z=\Phi(b,w)$, $b\in A$, $w\in I_t(u)$, and $\dist(\Phi(b,w),\mathcal M\cap B_3)\leq u$. Since $z=\Phi(b,w)$ and $\mathcal M\cap B_3\subseteq\mathcal M$, we have $\dist(z,\mathcal M)\leq\dist(z,\mathcal M\cap B_3) =\dist(\Phi(b,w),\mathcal M\cap B_3)\leq u<u_2$. Hence, \eqref{eq:intcont} implies that $z\in U$. By \eqref{eq:LimitDen_19aa}, $z$ has a unique representation under $\Phi$. Therefore, the base point in this representation is the $b$ appearing in $z=\Phi(b,w)$, and \eqref{eq:LimitDen_19c} gives $B(z)=b\in A$. Moreover, $\dist(z,\mathcal M\cap B_3) =\dist(\Phi(b,w),\mathcal M\cap B_3)\leq u$. Finally, the definition of $I_t(u)$ and the orientation of $\nu$ imply that $T(z)=t$. This shows that $z$ is in the LHS of \eqref{eq:LimitDen_20a} and concludes the proof of part 2.

\medskip\noindent
\underline{Part 3.} This part derives a convergence result for the probability of a ``good'' event. Specifically, our goal is to find $u_3=u_3(\varepsilon)>0$ such that, for all $u\in(0,u_3)$,
\begin{align}
&\sup_{S\in Bo(\mathcal B),\,t\in\{0,1\}}
\left|
\begin{array}{c}
P(\{B(Z_i)\in S\cap B_3\cap\mathcal M\} \cap \{\dist(Z_i,B_3\cap\mathcal M)\leq u\} \cap \{T(Z_i)=t\})/u \\
- \int_{S\cap B_3\cap\mathcal M} f_Z(b)\,d\mathcal H^{d-1}(b)
\end{array}
\right| ~\leq~ \varepsilon/60. \label{eq:LimitDen_18}
\end{align}
To this end, it suffices to show that, for $t\in\{0,1\}$,
\begin{equation}
\lim_{u\downarrow0} \sup_{S\in Bo(\mathcal B)}
\left|
\begin{array}{c}
P(\{B(Z_i)\in S\cap B_3\cap\mathcal M\} \cap \{\dist(Z_i,B_3\cap\mathcal M)\leq u\} \cap \{T(Z_i)=t\})/u \\
- \int_{S\cap B_3\cap\mathcal M} f_Z(b)\,d\mathcal H^{d-1}(b)
\end{array}
\right| ~=~0.
\label{eq:LimitDen_19}
\end{equation}

For any $u\in(0,u_2)$, consider the following argument:
\begin{align}
&P(\{B(Z_i)\in S\cap B_3\cap\mathcal M\} \cap \{\dist(Z_i,B_3\cap\mathcal M)\leq u\} \cap \{T(Z_i)=t\}) \notag\\
&\overset{(1)}{=}~ \int_{\left\{ z\in\mathbb R^d:~ B(z)\in S\cap B_3\cap\mathcal M,\ \dist(z,B_3\cap\mathcal M)\leq u,\ T(z)=t \right\}} f_Z(z)\,dz \notag\\
&\overset{(2)}{=}~ \int_{\{\Phi(b,w):~ b\in S\cap B_3\cap\mathcal M,\ w\in I_t(u),\ \dist(\Phi(b,w),\mathcal M\cap B_3)\leq u\}} f_Z(z)\,dz \notag\\
&\overset{(3)}{=}~ \int_{\{\Phi(b,w):~ b\in S\cap B_3\cap\mathcal M,\ w\in I_t(u)\}} f_Z(z)\,dz \notag\\
&\overset{(4)}{=}~ \int_{I_t(u)} \int_{S\cap B_3\cap\mathcal M} f_Z(\Phi(b,w))J\Phi(b,w)\, d\mathcal H^{d-1}(b)\,dw,
\label{eq:LimitDen_21}
\end{align}
where (1) holds by Assumption \ref{ass:assumption}(a) and $u<u_2 \leq \delta$, (2) by \eqref{eq:LimitDen_20a} applied with $A=S\cap B_3\cap\mathcal M$, (3) because $b\in S\cap B_3\cap\mathcal M$ and $w\in I_t(u)$ imply $\dist(\Phi(b,w),\mathcal M\cap B_3) \leq\|\Phi(b,w)-b\|=|w|\leq u$, and (4) by the change of variables according to the coarea formula, where $J\Phi(b,w)=\left|\det\left({\bf I}_{T_b(\mathcal B\setminus\Sigma)}+wD\nu(b)\right)\right|$ is the Jacobian determinant of $\Phi $ restricted to the tangent space $T_b(\mathcal B\setminus\Sigma)$, and ${\bf I}_{T_b(\mathcal B\setminus\Sigma)}+wD\nu(b)$ is a linear operator defined on $T_b(\mathcal B\setminus\Sigma)$ for each $b \in \mathcal B\setminus\Sigma$.

For every $b\in\mathcal B\cap B_3\cap\mathcal M$, $\Phi(b,w)=b+w\nu(b)\to b$ as $w\to0$. Since $\mathcal M\cap\operatorname{cl}(B_3)$ is compact and $f_Z$ is continuous at $b$ by $H_0$, it follows that
\begin{equation}
\lim_{u\downarrow0} \sup_{|w|\leq u} \sup_{b\in\mathcal B\cap B_3\cap\mathcal M} |f_Z(\Phi(b,w))-f_Z(b)| ~=~ 0.
\label{eq:LimitDen_22b}
\end{equation}

Moreover, since $\sup_{b\in\mathcal M}\|D\nu(b)\|\leq C<\infty$,
\begin{equation}
\sup_{|w|\leq u} \sup_{b\in\mathcal M} \|wD\nu(b)\| ~\leq~ uC.
\label{eq:LimitDen_23a}
\end{equation}
By \eqref{eq:LimitDen_23a} and the continuity of the determinant at the identity matrix,
\begin{equation}
\lim_{u\downarrow0} \sup_{|w|\leq u} \sup_{b\in\mathcal M} |J\Phi(b,w)-1| ~=~0. \label{eq:LimitDen_23}
\end{equation}

Therefore, for each $t\in\{0,1\}$ and every $u\in(0,u_2)$,
\begin{align}
&\sup_{S\in Bo(\mathcal B)} \left|
\begin{array}{c}
P(\{B(Z_i)\in S\cap B_3\cap\mathcal M\} \cap \{\dist(Z_i,B_3\cap\mathcal M)\leq u\} \cap \{T(Z_i)=t\})/u \\
- \int_{S\cap B_3\cap\mathcal M} f_Z(b)\,d\mathcal H^{d-1}(b)
\end{array}
\right| \notag\\
&\overset{(1)}{=}~ \sup_{S\in Bo(\mathcal B)} \frac{1}{u} \left| \int_{I_t(u)} \int_{S\cap B_3\cap\mathcal M} \big( f_Z(\Phi(b,w))J\Phi(b,w)-f_Z(b) \big) \,d\mathcal H^{d-1}(b)\,dw \right| \notag\\
&\overset{(2)}{\leq}~ \sup_{|w|\leq u} \sup_{b\in\mathcal B\cap B_3\cap\mathcal M} |f_Z(\Phi(b,w))J\Phi(b,w)-f_Z(b)| \, \mathcal H^{d-1}(\mathcal B\cap B_3) \notag\\
&\leq~ \left(
\begin{array}{c}
(\sup_{|w|\leq u} \sup_{b\in\mathcal M} |J\Phi(b,w)-1| +1) \sup_{|w|\leq u} \sup_{b\in\mathcal B\cap B_3\cap\mathcal M} |f_Z(\Phi(b,w))-f_Z(b)| \\
+ (\sup_{b\in\mathcal B}f_Z(b)) \sup_{|w|\leq u} \sup_{b\in\mathcal M} |J\Phi(b,w)-1|
\end{array}
\right)
\mathcal H^{d-1}(\mathcal B\cap B_3),
\label{eq:LimitDen_24}
\end{align}
where (1) holds by \eqref{eq:LimitDen_21} and $\mathcal L(I_t(u))=u$, and (2) by $S\cap B_3\cap\mathcal M\subseteq\mathcal B\cap B_3$. By Assumptions \ref{ass:assumption}(a)--(b), \eqref{eq:LimitDen_22b}, and \eqref{eq:LimitDen_23}, the RHS of \eqref{eq:LimitDen_24} converges to zero as $u\downarrow0$. This proves \eqref{eq:LimitDen_19} for $t \in \{0,1\}$. Therefore, $\exists u_3=u_3(\varepsilon)\in(0,u_2)$ so that \eqref{eq:LimitDen_18} holds for every $u\in(0,u_3)$.

\medskip\noindent
\underline{Part 4.} This part establishes bounds for the probability of two ``bad'' events. Specifically, we find $u_4=u_4(\varepsilon)>0$ so that for all $u\in(0,u_4)$,
\begin{align}
P(\{B(Z_i)\in\mathcal B\cap B_3\} \cap \{\dist(Z_i,\mathcal B\cap B_3)\leq u\} \cap \{\dist(Z_i,\mathcal M)>u\})/u ~&\leq~ \varepsilon/2,\label{eq:LimitDen_30}\\
P(\{B(Z_i)\in(\mathcal B\cap B_3)\setminus\mathcal M\} \cap \{\dist(Z_i,\mathcal M\cap B_3)\leq u\})/u~&\leq~ 3\varepsilon/20.
\label{eq:LimitDen_30_B}
\end{align}
We focus on proving \eqref{eq:LimitDen_30}, and the proof of \eqref{eq:LimitDen_30_B} is obtained as a by-product.

For any $v>0$, let $C(v):=\{b\in\mathcal B\cap B_3:\dist(b,\mathcal M)>v\}$. Let $\lambda_1\in(0,1)$ be as in \eqref{eq:LimitDen_32}. For any $u\in(0,1)$,
\begin{align}
&\{\dist(z,\mathcal B\cap B_3)\leq u\} \cap \{\dist(z,\mathcal M)>u\}\notag\\
&\overset{(1)}{\subseteq}~ \{\dist(z,(\mathcal B\cap B_3)\setminus\mathcal M)\leq u\} \cap \{\dist(z,\mathcal M)>u\}\notag\\
&{=}~
\left\{\begin{array}{l}
\big(\{\dist(z,(\mathcal B\cap B_3)\setminus\mathcal M)\leq u\} \cap \{\dist(z,\mathcal M) \in (u,(1+\lambda_1)u]\}\big) ~\cup \\
\big(\{\dist(z,(\mathcal B\cap B_3)\setminus\mathcal M)\leq u\} \cap \{\dist(z,\mathcal M)>(1+\lambda_1)u\}\big)
\end{array}\right\}
\notag\\
&\overset{(2)}{\subseteq}~
\big( \{\dist(z,\mathcal B)\leq u\} \cap \{\dist(z,\mathcal M)\in(u,(1+\lambda_1)u]\} \big) ~\cup~ \{\dist(z,C(\lambda_1u))\leq(1+\lambda_1)u\}\notag\\
&\overset{(3)}{\subseteq}~
\left\{
\begin{array}{c}
\big(\{\dist(z,\mathcal B)\leq u\}\cap\{z\notin B({\bf 0}_d,r_1)\}\big)~\cup ~\{\dist(z,\mathcal M\cap B_3) \in(u,(1+\lambda_1)u]\}\\
\cup~~ \{\dist(z,C(\lambda_1u))\leq(1+\lambda_1)u\}
\end{array}
\right\},
\label{eq:LimitDen_31}
\end{align}
where (1) holds because $\{\dist(z,\mathcal B\cap B_3)\leq u\}$ and $\{\dist(z,\mathcal M)>u\}$ implies
$\{\dist(z,(\mathcal B\cap B_3)\setminus\mathcal M)\leq u\}$, (2) holds because $\dist(z,(\mathcal B\cap B_3)\setminus\mathcal M)\leq u$ implies that, for every $\gamma>0$, there exists $b_\gamma\in(\mathcal B\cap B_3)\setminus\mathcal M$ such that $\|z-b_\gamma\|<u+\gamma$. Choose $\gamma = \min\{\lambda_1u,\dist(z,\mathcal M)-(1+\lambda_1)u\}>0$. Then, by the triangle inequality, $\dist(b_\gamma,\mathcal M)\geq \dist(z,\mathcal M)-\|z-b_\gamma\|>\lambda_1u$, so $b_\gamma\in C(\lambda_1u)$. Moreover, $\|z-b_\gamma\|<u+\gamma<(1+\lambda_1)u$, and hence $\dist(z,C(\lambda_1u))\leq(1+\lambda_1)u$, and  (3) holds by $\mathcal M=(\mathcal M\cap B_3)\cup(\mathcal M\cap B_3^c)$,
while $(1+\lambda_1)u<2$ implies that $\dist(z,\mathcal M\cap B_3^c)\leq(1+\lambda_1)u$ only if $z\notin B({\bf 0}_d,r_3-2)$, which, in turn, implies $z\notin B({\bf 0}_d,r_1)$ by $r_1\leq r_3-2$.

For any $u \in (0,1)$, we then have
\begin{align}
&P(\{B(Z_i)\in\mathcal B\cap B_3\} \cap \{\dist(Z_i,\mathcal B\cap B_3)\leq u\} \cap \{\dist(Z_i,\mathcal M)>u\})/u \notag \\
&\overset{(1)}{\leq}~
\left\{
\begin{array}{c}
P( \{\dist(Z_i,\mathcal B)\leq u\}\cap\{Z_i\notin B({\bf 0}_d,r_1)\} )/u\\
+ P(\{B(Z_i)\in\mathcal M\cap B_3\} \cap \{\dist(Z_i,\mathcal M\cap B_3) \in(u,(1+\lambda_1)u]\})/u\\
+ P(\{B(Z_i)\in(\mathcal B\cap B_3) \setminus \mathcal M\} \cap \{\dist(Z_i,\mathcal M\cap B_3) \leq (1+\lambda_1)u\})/u\\
+ P(\{\dist(Z_i,C(\lambda_1u))\leq(1+\lambda_1)u\})/u 
\end{array}
\right\},\label{eq:LimitDen_60}
\end{align}
where (1) holds by \eqref{eq:LimitDen_31} and $(\mathcal B\cap B_3)=(\mathcal M\cap B_3)\cup((\mathcal B\cap B_3)\setminus \mathcal M)$. To complete the proof, it suffices to find $u_4 \in (0,1)$, s.t.\  for all $u \in (0,u_4)$ we can bound the RHS of \eqref{eq:LimitDen_60} by $\varepsilon/2$.

We begin with the first term on the RHS of \eqref{eq:LimitDen_60}. For every $u\in(0,u_1)$,
\begin{equation}
P(\{\dist(Z_i,\mathcal B)\leq u\} \cap \{Z_i\notin B({\bf 0}_d,r_1)\}) /u~ \overset{(1)}{\leq}~ {\varepsilon}/{12}, \label{eq:LimitDen_60_T2}
\end{equation}
where (1) holds by \eqref{eq:LimitDen_2}.

We next deal with the second term on the RHS of \eqref{eq:LimitDen_60}. For every $u\in(0,u_3/(1+\lambda_1))$,
\begin{align}
& P(\{B(Z_i)\in\mathcal M\cap B_3\} \cap \{\dist(Z_i,\mathcal M\cap B_3) \in(u,(1+\lambda_1)u]\})/u \notag\\
& = ~\left\{
\begin{array}{c}
P(\{B(Z_i)\in\mathcal M\cap B_3\} \cap \{\dist(Z_i,\mathcal M\cap B_3) \leq (1+\lambda_1)u\})/u\\
- P(\{B(Z_i)\in\mathcal M\cap B_3\} \cap \{\dist(Z_i,\mathcal M\cap B_3) \leq u\})/u
\end{array}
\right\}
\notag\\
&{\leq }~
\left\{
\begin{array}{c}
(1+\lambda_1) 
\left|
\begin{array}{c}
P\big(\{B(Z_i)\in\mathcal M\cap B_3\} \cap \{\dist(Z_i,\mathcal M\cap B_3) \leq(1+\lambda_1)u\}\big)/((1+\lambda_1)u)\\
-2\int_{\mathcal M\cap B_3}f_Z(b)\,d\mathcal H^{d-1}(b) 
\end{array}
\right|+\\
\left|
\begin{array}{c}
P\big(\{B(Z_i)\in\mathcal M\cap B_3\} \cap \{\dist(Z_i,\mathcal M\cap B_3) \leq u\}\big)/u\\
- 2\int_{\mathcal M\cap B_3}f_Z(b)\,d\mathcal H^{d-1}(b) 
\end{array}
\right|+2\lambda_1 \int_{\mathcal M\cap B_3}f_Z(b)\,d\mathcal H^{d-1}(b) 
\end{array}
\right\} 
\notag\\
&~\overset{(1)}{\leq}~ (1+\lambda_1)\frac{\varepsilon}{30} + \frac{\varepsilon}{30} + \frac{\varepsilon}{12} ~\overset{(2)}{\leq}~ \frac{11\varepsilon}{60},
\label{eq:LimitDen_60_T3}
\end{align}
where (1) holds by $\mathcal{M} \cap B_3 \subseteq \mathcal{B}$ and \eqref{eq:LimitDen_32}, and \eqref{eq:LimitDen_18}, applied at $t\in\{0,1\}$, $u$, and $(1+\lambda_1)u$, and (2) by $\lambda_1<1$. 

We next deal with the third term on the RHS of \eqref{eq:LimitDen_60}. Since $\sup_{b\in\operatorname{cl}(V)}\|D\nu(b)\|<\infty$, the argument used in \eqref{eq:LimitDen_23a}--\eqref{eq:LimitDen_23} gives $\lim_{v\downarrow0}\sup_{|w|\leq v}\sup_{b\in\operatorname{cl}(V)} |J\Phi(b,w)-1|=0$. Hence, $\exists d_3 \in(0,w_0)$ so that
\begin{equation}
    \sup_{|w|\leq d_3}\sup_{b\in\operatorname{cl}(V)}J\Phi(b,w)~\leq~2.
    \label{eq:LimitDen_32a}
\end{equation}

Fix $u\in(0,u_2/(1+\lambda_1))$ and consider any $z$ such that $\dist(z,\mathcal M\cap B_3)\leq(1+\lambda_1)u$. Since $\mathcal M\cap B_3\subseteq\mathcal M$, we have $\dist(z,\mathcal M)\leq\dist(z,\mathcal M\cap B_3) \leq(1+\lambda_1)u<u_2$. Thus, \eqref{eq:LimitDen_19aa} gives a unique representation $z=\Phi(b,w)$ with $b\in V$, and \eqref{eq:LimitDen_19c} gives $B(z)=b\in V$. Therefore, on the set under consideration, $\{B(z)\in(\mathcal B\cap B_3)\setminus\mathcal M\} = \{B(z)\in((\mathcal B\cap B_3)\setminus\mathcal M)\cap V\}$. We can consequently apply \eqref{eq:LimitDen_20a} with $A=((\mathcal B\cap B_3)\setminus\mathcal M)\cap V$ and $v=(1+\lambda_1)u$.

Consequently, for any $u \in (0,\min\{d_3/(1+\lambda_1), u_2/(1+\lambda_1),\delta/(1+\lambda_1)\})$,
\begin{align}
&P(\{B(Z_i)\in(\mathcal B\cap B_3)\setminus\mathcal M\} \cap \{\dist(Z_i,\mathcal M\cap B_3)\leq(1+\lambda_1)u\})/u \notag\\
&\overset{(1)}{=}~ \frac{1}{u} \int_{\left\{ z\in\mathbb R^d: B(z)\in(\mathcal B\cap B_3)\setminus\mathcal M,\ \dist(z,\mathcal M\cap B_3)\leq(1+\lambda_1)u \right\}} f_Z(z)\,dz \notag\\
&\overset{(2)}{\leq}~ \frac{1}{u} \int_{\left\{ \Phi(b,w): b\in((\mathcal B\cap B_3)\setminus\mathcal M)\cap V,\ |w|\leq(1+\lambda_1)u \right\}} f_Z(z)\,dz \notag\\
&\overset{(3)}{=}~\frac{1}{u} \int_{-(1+\lambda_1)u}^{(1+\lambda_1)u} \int_{((\mathcal B\cap B_3)\setminus\mathcal M)\cap V} f_Z(\Phi(b,w))J\Phi(b,w)\, d\mathcal H^{d-1}(b)\,dw\notag\\
&\overset{(4)}{\leq}~ 4(1+\lambda_1) \Big(\sup_{z\in\mathcal B^\delta}f_Z(z)\Big) \mathcal H^{d-1} \bigl((\mathcal B\cap B_3)\setminus\mathcal M\bigr) \notag\\
&\overset{(5)}{\leq}~ 3{\varepsilon}/{20},
\label{eq:LimitDen_60_T4}
\end{align}
where (1) holds by Assumption \ref{ass:assumption}(a) and $u < \delta/(1+\lambda_1)$, (2) by \eqref{eq:LimitDen_20a}, applied for each $t\in\{0,1\}$ with $A=((\mathcal B\cap B_3)\setminus\mathcal M)\cap V$ and $v=(1+\lambda_1)u$, and by $I_0(v)\cup I_1(v)=[-v,v]$, (3) by the change of variables according to the coarea formula, 
(4) by $u<d_3/(1+\lambda_1)$, \eqref{eq:LimitDen_32a}, $\mathcal L([-(1+\lambda_1)u,(1+\lambda_1)u])=2(1+\lambda_1)u$, and $\Phi(b,w)\in\mathcal B^\delta$, which follows from $b\in\mathcal B$ and $|w|\leq(1+\lambda_1)u<\delta$, and (5) by \eqref{eq:LimitDen_12} and $\lambda_1<1$.

Finally, we deal with the last term on the RHS of \eqref{eq:LimitDen_60}. 
For any $u\in ( 0,1) $ we can cover $C( \lambda _{1} u)$ with a collection of balls $\{B(x,\lambda _{1}u):x\in C( \lambda _{1} u) \}$, whose centers lie in $C( \lambda _{1} u) $.
Note that $C(\lambda _{1}u)\subseteq B_{3}$ is a bounded set. By the Besicovitch Covering Theorem (see \cite[Theorem 5.1]{maggi:2012}), we can find a countable subset of centers $\tilde{C}( \lambda _{1} u) \equiv \{ \{ x_{j,l}\in C( \lambda _{1} u) :j\in I_{l}(u)\} :l=1,2,\dots,\xi ( d) \} $ so that $\{ B(x_{j,l},\lambda _{1}u):j\in I_{l}(u)\} $ are pairwise-disjoint for each $l=1,\dots,\xi ( d) $, and $\{ B(x,\lambda _{1}u):x\in \tilde{C}( \lambda _{1} u) \} $ covers $C( \lambda _{1} u)$, i.e.,
\begin{equation}
C( \lambda _{1} u) ~\subseteq~ \cup_{x\in \tilde{C}( \lambda _{1} u) }B(x,\lambda _{1}u)~=~\bigcup_{l=1,2,\dots,\xi ( d) }\bigcup_{j\in I_{l}(u)}B(x_{j,l},\lambda _{1}u).
\label{eq:LimitDen_35}
\end{equation}
By this, the triangle inequality, and $\lambda_1<1$, we obtain that
\begin{equation}
\{ z\in \mathbb{R} ^{d}:\dist(z,C( \lambda _{1} u) )\leq (1+\lambda _{1})u\} ~\subseteq~ \cup_{l=1,2,\dots,\xi ( d) }\cup_{j\in I_{l}(u)}B(x_{j,l},3u).
\label{eq:LimitDen_36}
\end{equation}

We now provide an argument to control the cardinality of $\tilde{C}( \lambda _{1} u) $. Since $\{ B(x_{j,l},\lambda _{1}u):j\in I_{l}(u)\} $ are pairwise-disjoint for each $l=1,\dots,\xi ( d) $, we have that, for any $b\in \mathbb{R}^{d}$,
\begin{equation}
\sum_{x\in \tilde{C}( \lambda _{1} u) }{ \bf 1}\{ b\in B(x,\lambda _{1}u)\} ~=~\sum_{l=1,2,\dots,\xi ( d) }\sum_{j\in I_{l}(u)}{ \bf 1}\{ b\in B(x_{j,l},\lambda _{1}u)\} ~\leq~ \xi ( d) .
\label{eq:LimitDen_37}
\end{equation}
Furthermore, for any $x_{j,l}\in \tilde{C}( \lambda _{1} u)$, one can show that
\begin{equation}
(B(x_{j,l},\lambda _{1} u)\cap \mathcal{B})~\subseteq ~( \mathcal{B}\cap B( {\bf 0}_{d},r_{3}+1) ) \setminus \mathcal{M}.
\label{eq:LimitDen_38}
\end{equation}

Let $\tilde{r}>0$ be as in part 3 of Lemma \ref{lem:regset_v1}. Then, for any $u\in ( 0,\tilde{r}/\lambda _{1}) $,
\begin{align}
\sum_{l=1,2,\dots,\xi ( d) }\sum_{j\in I_{l}(u)}\eta (\lambda_1 u)^{d-1}~&\overset{(1)}{\leq }~\sum_{l=1,2,\dots,\xi ( d) }\sum_{j\in I_{l}(u)}\mathcal{H}^{d-1}(B(x_{j,l},\lambda _{1}u)\cap \mathcal{B})\notag \\
&~\overset{(2)}{=}~\sum_{l=1,2,\dots,\xi ( d) }\sum_{j\in I_{l}(u)}\int_{( \mathcal{B}\cap B( 0,r_{3}+1) ) \setminus \mathcal{M}}{ \bf 1}\{ b\in B(x_{j,l},\lambda_{1} u)\} d\mathcal{H}^{d-1}(b) \notag \\
&~\overset{(3)}{=}~\int_{( \mathcal{B}\cap B( 0,r_{3}+1) ) \setminus \mathcal{M}}\Big( \sum_{l=1,2,\dots,\xi ( d) }\sum_{j\in I_{l}(u)}{ \bf 1}\{ b\in B(x_{j,l},\lambda_ 1 u)\}\Big) d\mathcal{H}^{d-1}(b)  \notag \\
&~\overset{(4)}{\leq }~\xi ( d) \mathcal{H}^{d-1}(( \mathcal{B }\cap B( {\bf 0}_{d},r_{3}+1) ) \setminus \mathcal{M})  \notag \\
&~\overset{(5)}{\leq }~\Big( \frac{\eta \lambda_1^{d-1} }{\mathcal{L}\left( B({\bf 0}_{d},1\right) )3^{d}(\sup_{b\in \mathcal{B}^{\delta }}f_{Z}(b))}\Big)\varepsilon /12,\label{eq:LimitDen_39}
\end{align}
where (1) holds by $x_{j,l}\in \mathcal{B}$, $\lambda _{1}u<\tilde{r}$, and part 3 of Lemma \ref{lem:regset_v1}, (2) by \eqref{eq:LimitDen_38}, (3) by using Tonelli's Theorem to justify interchanging sum and integrals,
(4) by \eqref{eq:LimitDen_37}, and (5) by \eqref{eq:LimitDen_12}. For any $u\in ( 0,\tilde{r}/\lambda _{1}) $, \eqref{eq:LimitDen_39} implies that 
\begin{equation}
\sum_{l=1,2,\dots,\xi ( d) }\sum_{j\in I_{l}(u)}1~\leq~ \Big(\frac{1}{u ^{d-1}\mathcal{L}( B({\bf 0}_{d},1) )3^{d}(\sup_{b\in \mathcal{B}^{\delta }}f_{Z}(b)) }\Big)\varepsilon /12.
\label{eq:LimitDen_39b}
\end{equation}

Then, for all $u\in ( 0,\min \{1, \delta/(1+\lambda_1),\tilde{r}/\lambda _{1}\}) $, we get
\begin{align}
P(\dist(Z_{i},C( \lambda _{1} u) )\leq (1+\lambda _{1})u)/u &~\overset{(1)}{\leq }~\big(\sup_{b\in \mathcal{B}^{\delta}}f_{Z}(b)\big) \mathcal{L}(\{ z \in \mathbb{R}^d: \dist(z,C( \lambda _{1} u) ) \leq (1+\lambda _{1})u\})/u \notag\\
&\overset{(2)}{\leq }~\big(\sup_{b\in \mathcal{B}^{\delta}}f_{Z}(b)\big)\mathcal{L }\big( \cup_{l=1,2,\dots,\xi ( d) }\cup_{j\in I_{l}(u)}B(x_{j,l},3u)\big) /u \notag\\
&\overset{(3)}{\leq }~\big(\sup_{b\in \mathcal{B}^{\delta}}f_{Z}(b)\big)u^{d-1} \mathcal{L }( B({\bf 0}_{d},1))3^d  \sum_{l=1,2,\dots,\xi ( d) }\sum_{j\in I_{l}(u)}1 \notag\\
&\overset{(4)}{\leq } \varepsilon /12,\label{eq:LimitDen_60_T1}
\end{align}
where (1) holds by $u<\delta/(1+\lambda_1) $ and Assumption
\ref{ass:assumption}(a), (2) by \eqref{eq:LimitDen_36}, (3) by subadditivity of Lebesgue measure and $\mathcal{L}(B(x,3u))=u^d3^d\mathcal{L}(B({\bf 0}_{d},1))$ for all $x\in\mathbb{R}^d$, and (4) by $u< \tilde{r}/\lambda _{1}$ and \eqref{eq:LimitDen_39b}.

To conclude the proof, set $u_4 = \min \{u_1,1,\tilde{r}/\lambda _{1}, \min\{u_2,u_3,\delta,d_3\}
/(1+\lambda_1)\}$. For all $u \in (0,u_4)$, \eqref{eq:LimitDen_30} then follows from combining \eqref{eq:LimitDen_60}, \eqref{eq:LimitDen_60_T2}, \eqref{eq:LimitDen_60_T3}, \eqref{eq:LimitDen_60_T4}, and \eqref{eq:LimitDen_60_T1}. Since $\lambda_1>0$, we also note that \eqref{eq:LimitDen_60_T4} implies \eqref{eq:LimitDen_30_B}.

\medskip\noindent
\underline{Part 5.} This part combines the previous parts to complete the proof.

For any $u\in (0,\min \{u_{1},u_{2},u_{3},u_{4},1\})$, we have 
\begin{align*}
& \sup_{S\in Bo(\mathcal{B}),\,t\in \{0,1\}}\left\vert P(\{B(Z_{i})\in S\}\cap \{\dist(Z_{i},\mathcal{B})\leq u\}\cap \{T(Z_{i})=t\})/u -\int_{S}f_{Z}(b)\,d\mathcal{H}^{d-1}(b)
\right\vert  \\
& \overset{(1)}{=}~\sup_{S\in Bo(\mathcal{B}),\,t\in \{0,1\}}\left\vert 
\begin{array}{c}
P(\{B(Z_{i})\in (S\cap B_{3}\cap \mathcal{M})\}\cap \{\dist(Z_{i},B_{3}\cap \mathcal{M})\leq u\}\cap \{T(Z_{i})=t\})/u \\ 
-\int_{S\cap B_{3}\cap \mathcal{M}}f_{Z}(b)\,d\mathcal{H}^{d-1}(b)-\int_{S \setminus (B_{3}\cap \mathcal{M}}f_{Z}(b)\,d\mathcal{H}^{d-1}(b)  \\ 
+P\left(
\begin{array}{c}
\{B(Z_{i})\in (S\cap B_{3}\cap \mathcal{M})\}\cap \{\dist(Z_{i},\mathcal{B})\leq u\}\cap \\
\{\dist(Z_{i},B_{3}\cap \mathcal{M})>u\}\cap \{T(Z_{i})=t\}
\end{array}
\right)/u \\ 
+P(\{B(Z_{i})\in (S\cap B_{3})\setminus \mathcal{M}\}\cap \{\dist(Z_{i}, \mathcal{B})\leq u\}\cap \{T(Z_{i})=t\})/u \\ 
+P(\{B(Z_{i})\in (S\setminus B_{3})\}\cap \{\dist(Z_{i},\mathcal{B})\leq u\}\cap \{T(Z_{i})=t\})/u
\end{array}
\right\vert  \\
& \leq ~\left\{ 
\begin{array}{c}
\sup_{S\in Bo(\mathcal{B}),\,t\in \{0,1\}}\left\vert 
\begin{array}{c}
P(\{B(Z_{i})\in (S\cap B_{3}\cap \mathcal{M})\}\cap \{\dist(Z_{i},B_{3}\cap \mathcal{M})\leq u\}\cap \{T(Z_{i})=t\})/u \\ 
-\int_{S\cap B_{3}\cap \mathcal{M}}f_{Z}(b)\,d\mathcal{H}^{d-1}(b)
\end{array}
\right\vert  \\ 
+P(\{B(Z_{i})\in (B_{3}\cap \mathcal{M})\}\cap \{\dist(Z_{i},\mathcal{B})\leq u\}\cap \{\dist(Z_{i},B_{3}\cap \mathcal{M})>u\})/u \\ 
+P(\{B(Z_{i})\in (\mathcal{B}\cap B_{3})\setminus \mathcal{M}\}\cap \{\dist(Z_{i},\mathcal{B})\leq u\})/u \\ 
+P(\{B(Z_{i})\not\in B_{3}\}\cap \{\dist(Z_{i},\mathcal{B})\leq u\})/u +\int_{\mathcal{B}\setminus B_{3}}f_{Z}(b)\,d\mathcal{H}^{d-1}(b)+\int_{\mathcal{B}\setminus \mathcal{M}}f_{Z}(b)\,d\mathcal{H}^{d-1}(b)
\end{array}
\right\}  \\
& \overset{(2)}{\leq }~\left\{ 
\begin{array}{c}
\underset{S\in Bo(\mathcal{B}),\,t\in \{0,1\}}{\sup }\left\vert 
\begin{array}{c}
P(\{B(Z_{i})\in S\cap B_{3}\cap \mathcal{M}\}\cap \{\dist(Z_{i},B_{3}\cap \mathcal{M})\leq u\}\cap \{T(Z_{i})=t\})/u \\ 
-\int_{S\cap B_{3}\cap \mathcal{M}}f_{Z}(b)\,d\mathcal{H}^{d-1}(b)
\end{array}
\right\vert  \\ 
+P(\{B(Z_{i})\in (\mathcal{B}\cap B_{3})\setminus \mathcal{M}\}\cap \{\dist(Z_{i},B_{3}\cap \mathcal{M})\leq u\})/u \\ 
+P(\{\dist(Z_{i},\mathcal{B})\leq u\}\cap \{Z_{i}\notin B(\mathbf{0}_{d},r_{3}-1)\})/u \\ 
+P(\{B(Z_{i})\in \mathcal{B}\cap B_{3}\}\cap \{\dist(Z_{i},\mathcal{B}\cap B_{3})\leq u\}\cap \{\dist(Z_{i},\mathcal{M})>u\})/u \\ 
+P(\{B(Z_{i})\notin B_{3}\}\cap \{\dist(Z_{i},\mathcal{B})\leq u\})/u +\int_{\mathcal{B}\setminus B_{3}}f_{Z}(b)\,d\mathcal{H}^{d-1}(b)+ \int_{\mathcal{B}\setminus \mathcal{M}}f_{Z}(b)\,d\mathcal{H}^{d-1}(b)
\end{array}
\right\}  \\
& \overset{(3)}{\leq}~\frac{\varepsilon}{60}+\frac{3\varepsilon}{20}+\frac{\varepsilon}{12}+\frac{\varepsilon}{2}+\frac{\varepsilon}{12}+\frac{\varepsilon}{12}+\frac{\varepsilon}{12}~=~\varepsilon ,
\end{align*}
as desired in \eqref{eq:LimitDen_1}, where (1) holds because $(S\cap B_{3})$ and $(S\setminus B_{3})$ partition $S$, $(S\cap B_{3}\cap \mathcal{M})$ and $S\setminus (B_{3}\cap \mathcal{M})$ partition $S$, $(S\cap B_{3}\cap \mathcal{M})$ and $(S\cap B_{3})\setminus \mathcal{M}$ partition $(S\cap B_{3})$, (2) holds by $\{B(Z_{i})\in (B_{3}\cap \mathcal{M})\}\cap \{\dist(Z_{i},\mathcal{B})\leq u\}\cap \{\dist(Z_{i},B_{3}\cap \mathcal{M} )>u\}=\emptyset $, $\{B(Z_{i})\in (\mathcal{B}\cap B_{3})\setminus \mathcal{M}\}\cap \{\dist(Z_{i},\mathcal{B})\leq u\}\cap \{\dist(Z_{i},B_{3}\cap \mathcal{M})>u\}\subseteq \{\dist(Z_{i},\mathcal{B}\cap B_{3})\leq u\}$, and, for $u<1$, $\{\dist(Z_{i},B_{3}\cap \mathcal{M})>u\}\cap \{\dist(Z_{i}, \mathcal{M})\leq u\}\subseteq \{Z_{i}\in B(\mathbf{0}_{d},r_{3}-1)^{c}\}$, and (3) holds by $r_{3}-1\geq r_{1}$, \eqref{eq:LimitDen_2}, \eqref{eq:LimitDen_3}, \eqref{eq:LimitDen_5}, \eqref{eq:LimitDen_11}, \eqref{eq:LimitDen_18}, \eqref{eq:LimitDen_30}, and \eqref{eq:LimitDen_30_B}.
\end{proof}

\section{Definitions from Geometric Measure Theory}\label{sec:GMTstuff}

The paper uses the following standard definitions from geometric measure theory. See \cite{federer:1996,evans/gariepy:1992,morgan:1998,maggi:2012} for references on this topic.

For any set $A \subset \mathbb{R}^d$, the diameter of $A$ is defined as
\[
\mathrm{diam}(A)~\equiv~\sup_{x,y \in A} \Vert x-y \Vert.
\]
Let $\omega_m$ denote the volume of the unit ball in $\mathbb R^m$. For any set $A\subset\mathbb{R}^d$, $m>0$, and $\delta>0$, define
\begin{equation*}
    \mathcal{H}^m_\delta(A)
    ~\equiv~
    \inf\Bigg\{
    \sum_{j \in \mathbb{N}} \omega_m
    \left( \frac{\mathrm{diam}(D_j)}{2} \right)^m
    ~:~
    A\subset\cup_{j\geq 1} D_j,\;
    \mathrm{diam}(D_j)\le\delta
    \Bigg\}.
\end{equation*}
The $m$-dimensional Hausdorff measure of $A$ is defined by
\begin{equation*}
    \mathcal{H}^m(A)~\equiv~\lim_{\delta\downarrow 0}\mathcal{H}^m_\delta(A).
\end{equation*}
With this normalization, $\mathcal H^d$ coincides with $d$-dimensional Lebesgue measure on $\mathbb R^d$. More generally, when $m<d$, $\mathcal H^m$ measures the size of lower-dimensional subsets of $\mathbb R^d$. For example, when $d=2$, line segments have zero two-dimensional Lebesgue measure, but their lengths are measured by $\mathcal H^1$. For smooth hypersurfaces in $\mathbb R^d$, $\mathcal H^{d-1}$ coincides with the usual notion of surface measure.

The Hausdorff dimension of a set $A\subset\mathbb R^d$ is defined as
\begin{equation*}
    \dim_H(A)
    ~\equiv~
    \inf\{m\geq 0:\mathcal H^m(A)=0\}
    ~=~
    \sup\{m\geq 0:\mathcal H^m(A)=\infty\},
\end{equation*}
with the usual convention that $\inf\{\emptyset\}=\infty$. If $\mathcal H^m(A)\in(0,\infty)$, then $\dim_H(A)=m$. 

We now collect several definitions concerning the boundaries of sets. The {\em topological boundary} of a set $A\subset\mathbb R^d$ is the intersection of the closure of $A$ and the closure of its complement:
\begin{equation*}
    bd(A)~=~cl(A)\cap cl(A^c).
\end{equation*}
This is the notion of boundary used in the definition of $\mathcal B$ in Section \ref{sec:setup}.

For any set $A\subseteq\mathbb{R}^{d}$ and any point $x\in\mathbb{R}^{d}$, the {\em measure theoretic density} of $A$ at $x$ is defined as
\begin{equation*}
    \theta(A,x)
    ~\equiv~
    \lim_{r\downarrow 0}
    \frac{\mathcal{L}^d(A\cap B(x,r))}
    {\mathcal{L}^d(B(x,r))},
\end{equation*}
whenever this limit exists. The {\em measure theoretic interior} and {\em measure theoretic exterior} of $A$ are defined, respectively, as
\begin{equation*}
    A^{(1)}
    ~\equiv~
    \{x\in\mathbb R^d:\theta(A,x)\text{ exists and equals }1\}
\end{equation*}
and
\begin{equation*}
    A^{(0)}
    ~\equiv~
    \{x\in\mathbb R^d:\theta(A,x)\text{ exists and equals }0\}.
\end{equation*}
The {\em measure theoretic boundary} of $A$ is
\begin{equation*}
    \partial^e A
    ~\equiv~
    \mathbb R^d\setminus\big(A^{(1)}\cup A^{(0)}\big).
\end{equation*}

For any set $A$ of locally finite perimeter, in the sense of \citet[p.\ 117]{maggi:2012}, the {\em reduced boundary} of $A$, denoted by $\partial^\ast A$, is the set of points at which the distributional gradient of the indicator function of $A$ admits a measure-theoretic unit normal. By \citet[Corollary 15.8]{maggi:2012}, any point in the reduced boundary of $A$ has measure theoretic density equal to $1/2$:
\begin{equation*}
    x\in\partial^\ast A
    \quad\Longrightarrow\quad
    \lim_{r\downarrow 0}
    \frac{\mathcal{L}(A\cap B(x,r))}
    {\mathcal{L}(B(x,r))}
    ~=~
    \frac{1}{2}.
\end{equation*}
This implies that the reduced boundary is contained in the measure-theoretic boundary, that is,
\begin{equation*}
    \partial^\ast A \subseteq \partial^e A.
\end{equation*}

We also use two standard notions from differential topology. First, a subset $M \subset \mathbb{R}^d$ is a $C^r$ manifold of dimension $m$ if, near every point $x\in M$, the set $M$ can be described by $m$ local coordinates in a way that is $r$ times continuously differentiable. Equivalently, for every $x\in M$, there exists a neighborhood $U\subset\mathbb{R}^d$ of $x$ and a $C^r$ change of coordinates on $U$ such that $M\cap U$ is mapped onto an open subset of $\mathbb{R}^m\times\{0\}^{d-m}$; see \citet[p.~12]{hirsch2012differential}.

Second, the $C^r$ norm measures the size of a function, or of a local parametrization, together with the size of its derivatives up to order $r$. For example, if $\psi:U\subset\mathbb{R}^m\to\mathbb{R}^d$ is a local parametrization and $K\subset U$ is compact, then
\begin{equation*}
    \Vert \psi \Vert_{C^r(K)}
    ~=~
    \max_{|a|\le r}
    \sup_{u\in K}\Vert D^a \psi(u)\Vert,
\end{equation*}
where $a$ is a multi-index. Thus, a bound on the $C^r$ norm controls the parametrization itself and its derivatives up to order $r$; see \citet[p.~35]{hirsch2012differential}. In Assumption \ref{ass:assumption}(d), we apply these notions to the boundary $\mathcal B$. We use $\Vert M\Vert_{C^2}$ as shorthand for the corresponding $C^2$ size of the local parametrizations of a manifold $M$. The assumption requires that there exists a countable collection of disjoint $(d-1)$-dimensional $C^2$ manifolds $\{\mathcal M_j\}_{j\in\mathbb N}$ with uniformly bounded $C^2$ size,
\begin{equation*}
\sup_{j \in \mathbb{N}} \Vert \mathcal M_j\Vert_{C^{2}} < \infty
\end{equation*}
such that these manifolds cover the boundary up to an $\mathcal H^{d-1}$-negligible set:
\begin{equation*}
\mathcal{H}^{d-1}\big(\mathcal{B} \setminus \cup_{j=1}^\infty \mathcal M_j\big) = 0 .
\end{equation*}
Thus, apart from a set that is negligible in $(d-1)$-dimensional Hausdorff measure, the boundary is covered by countably many smooth hypersurface pieces. The $C^2$ condition ensures that each piece has well-defined tangent planes and locally bounded curvature, while the uniform bound on the $C^2$ norms gives common control over slopes and curvature across pieces. The assumption, therefore, allows $\mathcal B$ to be globally irregular and permits lower-dimensional singularities such as corners and edges, but rules out pathological local behavior on the parts of the boundary that matter for the $\mathcal H^{d-1}$ measure.

\section{Extension of \cite{lee:2008} and \cite{mccrary:2008} to BDD}\label{sec:leeExtension}

This section considers a natural extension of the statistical framework developed by \cite{lee:2008} to the BDD framework introduced in Section \ref{sec:setting}. The framework introduces an unobservable variable $W$, interpreted as the unit's type, and studies the conditional distribution of the observable running variable $Z$ given $W$. The purpose of this section is to show that $H_0$ in \eqref{eq:H0} is the natural BDD analog of the density-continuity implication underlying the manipulation test in \cite{mccrary:2008}.

Aside from notational changes, the following assumption is a BDD counterpart to the framework proposed by \citet{lee:2008}. The theorem below uses only the smoothness conditions on the conditional distribution of $Z$ given $W$. The potential outcome models are included to clarify how the setup maps to the framework of \cite{lee:2008}.

\begin{assumption}\label{ass:Lee}
Assume the following:
\begin{enumerate}[(a)]

\item (Model 1) Let $(W,Z)$ be a pair of random variables, where $W$ is unobservable and $Z \in \mathbb{R}^{d}$ is observable. Let
\begin{gather*}
	Y_1~=~y_1(W), 
	\qquad
	Y_0~=~y_0(W),
	\qquad
	X~=~x(W),
\end{gather*}
where $y_1(\cdot)$, $y_0(\cdot)$, and $x(\cdot)$ are real-valued functions. Assignment is determined by the running variable according to
\begin{gather*}
	D~=~\mathbf {\bf 1}\{T(Z)=1\}.
\end{gather*}
Let $G(\cdot)$ denote the marginal distribution of $W$.

\item (Model 2) Let $(W,Z)$ be a pair of random variables, where $W$ is unobservable and $Z \in \mathbb{R}^{d}$ is observable. Let
\begin{gather*}
	Y~=~y(W,Z),
	\qquad
	X~=~x(W).
\end{gather*}
For each $w$, the function $z\mapsto y(w,z)$ is continuous away from $\mathcal B$ and admits one-sided limits at each point of $\mathcal B$. Specifically, for every $b\in\mathcal B$, define
\begin{gather*}
	y^{+}(w,b)
	~=~
	\lim_{z \to b,~ T(z)=1} y(w,z),
	\qquad
	y^{-}(w,b)
	~=~
	\lim_{z \to b,~ T(z)=0} y(w,z),
\end{gather*}
whenever these limits exist.

\item (Smooth conditional density) There exists $\delta>0$ such that, for $G$-almost every $w$, the conditional distribution of $Z$ given $W=w$ admits a density $f_{Z|W}(\cdot|w)$ on
\begin{gather*}
	\mathcal{B}^{\delta}
	\equiv
	\{z \in \mathbb{R}^{d}: \inf_{a \in \mathcal{B}}\Vert z-a\Vert \leq \delta\}.
\end{gather*}
Moreover, for every $b\in\mathcal B$ and $G$-almost every $w$,
\begin{gather*}
	\lim_{z\to b} f_{Z|W}(z|w)
	~=~
	f_{Z|W}(b|w).
\end{gather*}

\item (Dominating function) There exists a function $C(W)$ such that
\begin{gather*}
	\int C(w)dG(w)~<~\infty
\end{gather*}
and
\begin{gather*}
	\sup_{z \in \mathcal{B}^{\delta}} f_{Z|W}(z|w)
	~\leq~
	C(w)
\end{gather*}
for $G$-almost every $w$.

\end{enumerate}
\end{assumption}

Assumption \ref{ass:Lee} represents a BDD analog of \citet{lee:2008}'s framework. Part (a) is the multidimensional counterpart of \citet{lee:2008}'s Model 1, in which the running variable determines treatment assignment but does not directly enter the response functions for potential outcomes. Part (b) is the counterpart of Model 2, which allows the response function to depend directly on the running variable and to have different one-sided limits at the boundary. In both cases, the scalar assignment rule is replaced by the assignment rule $D=\mathbf {\bf 1}\{T(Z)=1\}$.

Parts (c) and (d) contain the conditions needed to derive the density continuity implication. Part (c) requires the conditional density of the vector running variable, given the unobservable type $W$, to be continuous at the boundary. Part (d) is a domination condition that justifies interchanging limits and integration over the distribution of unobserved types, a step that is implicit in \citet{lee:2008}'s derivations.

The following theorem shows that Assumption \ref{ass:Lee}(c)-(d) implies $H_0$ in \eqref{eq:H0}. This is the $d$-dimensional analog of the density-continuity implication underlying the manipulation test in \citet[page 701, second paragraph]{mccrary:2008}, within the framework proposed by \cite{lee:2008}.

\begin{theorem}\label{thm:Lee}
Under Assumption \ref{ass:Lee}(c)-(d), $H_0$ in \eqref{eq:H0} holds.
\end{theorem}
\begin{proof}
Consider any sequence $\{z_m\}_{m\geq 1}$ with $z_m\in \mathbb R^d$ and $z_m\to b\in\mathcal B$. Since $z_m\to b$, there exists $M$ such that $z_m\in\mathcal B^\delta$ for all $m\geq M$.

By Assumption \ref{ass:Lee}(c), for $G$-almost every $w$,
\begin{gather}
	\lim_{m\to\infty} f_{Z|W}(z_m|w)
	~=~
	f_{Z|W}(b|w).
	\label{eq:lee1}
\end{gather}
By Assumption \ref{ass:Lee}(d), for all $m\geq M$ and $G$-almost every $w$,
\begin{gather}
	f_{Z|W}(z_m|w)
	~\leq~
	C(w),
	\qquad
	\int C(w)dG(w)<\infty .
	\label{eq:lee2}
\end{gather}
Then,
\begin{align*}
	\lim_{m\to\infty} f_Z(z_m)
	&~=~
	\lim_{m\to\infty}
	\int f_{Z|W}(z_m|w)dG(w) \\
	&~\overset{(1)}{=}~
	\int \lim_{m\to\infty} f_{Z|W}(z_m|w)dG(w) \\
	&~\overset{(2)}{=}~
	\int f_{Z|W}(b|w)dG(w) \\
	&~=~
	f_Z(b),
\end{align*}
where (1) holds by \eqref{eq:lee2} and the dominated convergence theorem, and (2) by \eqref{eq:lee1}.

Since the sequence $\{z_m\}_{m\geq 1}$ was arbitrary, the marginal density $f_Z$ is continuous at $b$. In particular, this conclusion applies to sequences satisfying $T(z_m)=1$ and to sequences satisfying $T(z_m)=0$. Hence,
\begin{gather*}
	\lim_{z\to b,~T(z)=1} f_Z(z)
	~=~
	\lim_{z\to b,~T(z)=0} f_Z(z)
	~=~
	f_Z(b).
\end{gather*}
Since $b\in\mathcal B$ was arbitrary, $H_0$ in \eqref{eq:H0} holds.
\end{proof}




The theorem provides a formal motivation for $H_0$ in \eqref{eq:H0}. If the distribution of the vector running variable is smooth conditional on unobserved types, then integrating over types preserves smoothness of the marginal density at the boundary. Thus, a discontinuity of the marginal density at the boundary is inconsistent with this smooth sorting framework. As in the one-dimensional RDD, this implication does not test the full set of identification assumptions of the design. It provides a testable implication that can be used as a diagnostic for the validity of the BDD.

\begin{remark}
Assumptions \ref{ass:Lee}(a)-(b) are not required for Theorem \ref{thm:Lee}. We include them to show how the framework of \cite{lee:2008} can be extended to the BDD setting and to clarify which aspects of that framework motivate the joint density continuity null. The density continuity implication follows from the smoothness of the conditional distribution of the vector running variable given the unobserved type, together with the domination condition in Assumption \ref{ass:Lee}(d).
\end{remark}

\section{Behavior of the signed distance test}\label{sec:signed_distance_test}

This appendix studies the properties of the signed distance test. We present two examples illustrating its potential limitations. First, under certain alternatives to $H_0$ in \eqref{eq:H0}, the density of signed distance can remain continuous at zero, rendering the signed distance test powerless against those alternatives. Second, under $H_0$, the density of signed distance can be discontinuous at zero, causing the signed distance test to reject even though $H_0$ holds.

\subsection{The signed distance test may be powerless against violations of $H_0$}
\label{app:counterexample_averaging}

This section illustrates why a failure of $H_0$ in \eqref{eq:H0} does not imply that the signed-distance density is discontinuous at zero. We illustrate this point with a simple example.

Let $d=2$, where $Z=(Z_1,Z_2)$ has density
\begin{equation}
    f_Z(z_1,z_2)
    ~=~
    \phi(z_1)\phi(z_2)
    \left[
        1+\frac{1\{z_1<0\}}{2}\frac{z_2}{1+z_2^2}
    \right],
    \label{eq:counterexample_averaging_density}
\end{equation}
where $\phi$ denotes the standard normal density. 

Assume that the assignment rule is
\begin{equation*}
    T(z)~=~1\{z_1\geq0\}.
\end{equation*}
Thus, the boundary is
\begin{equation*}
   \mathcal B~=~\{(0,z_2):z_2\in\mathbb R\}.
\end{equation*}

Note that $H_0$ in \eqref{eq:H0} fails at almost every point of the boundary. Indeed, for any $z_2\in\mathbb R$, the two one-sided limits are
\begin{align*}
    \lim_{z_1\downarrow0}f_Z(z_1,z_2)
    ~=~
    \phi(0)\phi(z_2)\qquad\text{and}\qquad
    \lim_{z_1\uparrow0}f_Z(z_1,z_2)
    ~=~
    \phi(0)\phi(z_2)
    \left[
        1+\frac{1}{2}\frac{z_2}{1+z_2^2}
    \right].
\end{align*}
Therefore, the two one-sided limits differ for every $z_2\neq0$. In particular, the density jumps upward when the boundary is crossed from the treatment region to the control region for $z_2>0$ and downward for $z_2<0$. These discontinuities have opposite signs and average out along the boundary.

Figure \ref{fig:counterexample-averaging} illustrates the construction. Panel (a) displays the treatment and control regions and the behavior of the joint density along the boundary. Panel (b) plots the density of signed distance.

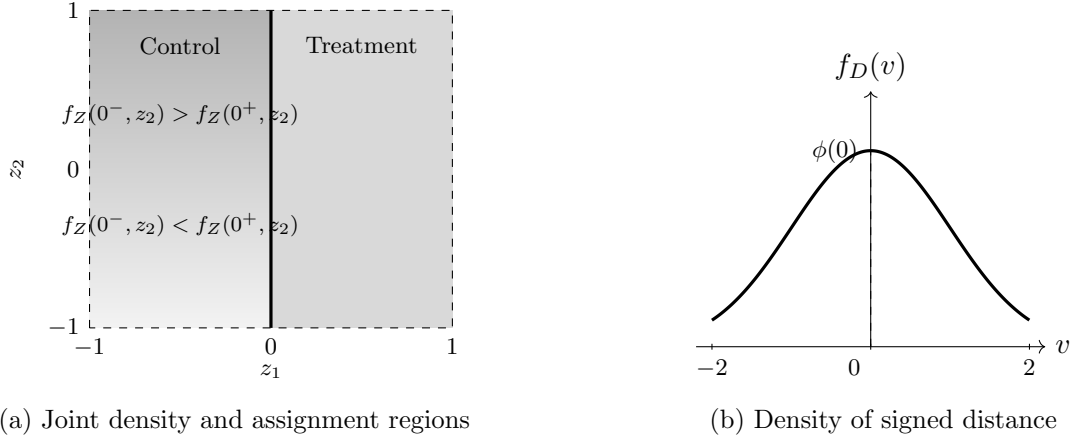
\begin{figure}[htbp]
\centering

\begin{minipage}[t]{0.48\textwidth}
\centering
\begin{tikzpicture}[xscale=2.4,yscale=2.1]
    \shade[top color=gray!60,bottom color=gray!10]
        (-1,-1) rectangle (0,1);

    \fill[gray!30] (0,-1) rectangle (1,1);

    \draw[dashed] (-1,-1) rectangle (1,1);

    \draw[very thick] (0,-1) -- (0,1);

    \node at (-0.5,0.78) {\footnotesize Control};
    \node at (0.5,0.78) {\footnotesize Treatment};

    \node[align=center] at (-0.5,0.35)
        {\scriptsize $f_Z(0^-,z_2)>f_Z(0^+,z_2)$};
    \node[align=center] at (-0.5,-0.35)
        {\scriptsize $f_Z(0^-,z_2)<f_Z(0^+,z_2)$};

    \node[below] at (-1,-1) {\footnotesize $-1$};
    \node[below] at (0,-1) {\footnotesize $0$};
    \node[below] at (1,-1) {\footnotesize $1$};

    \node[left] at (-1,-1) {\footnotesize $-1$};
    \node[left] at (-1,0) {\footnotesize $0$};
    \node[left] at (-1,1) {\footnotesize $1$};

    \node[below=0.35cm] at (0,-1) {\footnotesize $z_1$};
    \node[rotate=90] at (-1.4,0) {\footnotesize $z_2$};
\end{tikzpicture}

\vspace{0.3em}
\small (a) Joint density and assignment regions
\end{minipage}
\hfill
\begin{minipage}[t]{0.48\textwidth}
\centering
\begin{tikzpicture}[xscale=1.05,yscale=6.5]
    \draw[->] (-2.2,0) -- (2.2,0) node[right] {$v$};
    \draw[->] (0,0) -- (0,0.52) node[above] {$f_D(v)$};

    \draw[very thick,domain=-2:2,samples=100,smooth]
        plot (\x,{0.39894228*exp(-\x*\x/2)});

    \draw[dashed] (0,0) -- (0,0.3989);

    \draw (-2,0.006) -- (-2,-0.006)
        node[below] {\footnotesize $-2$};
    \draw (0,0.006) -- (0,-0.006)
        node[below left] {\footnotesize $0$};
    \draw (2,0.006) -- (2,-0.006)
        node[below] {\footnotesize $2$};

    \draw (0.03,0.3989) -- (-0.03,0.3989)
        node[left] {\footnotesize $\phi(0)$};
\end{tikzpicture}

\vspace{0.3em}
\small (b) Density of signed distance
\end{minipage}

\caption{\small Panel (a) shows the treatment and control regions and the behavior of the joint density along the boundary within a portion of $\mathbb R^2$. The density on the control side of the boundary is greater than that on the treatment side for $z_2>0$ and smaller for $z_2<0$. Panel (b) shows that the density of signed distance is nevertheless continuous at zero.}
\label{fig:counterexample-averaging}
\end{figure}

Let $D(z)$ denote signed distance from the boundary. For any $z\in\mathbb R^2$,
\begin{equation*}
    D(z)
    ~=~
    \dist(z,\mathcal B)1\{T(z)=1\}
    -
    \dist(z,\mathcal B)1\{T(z)=0\}
    ~\overset{(1)}{=}~
    z_1,
\end{equation*}
where (1) follows from the definitions of $\mathcal B$ and $T(z)$. Consequently,
\begin{equation*}
    D~=~D(Z)~=~Z_1.
\end{equation*}

It is not hard to verify that $D$ has density
\begin{align*}
    f_D(v)
    ~=~
    \int f_Z(v,z_2)\,dz_2
    ~=~
    \phi(v)
    \left[
        1+\frac{1\{v<0\}}{2}
        \int_{\mathbb R}\phi(z_2)\frac{z_2}{1+z_2^2}\,dz_2
    \right]
    ~=~
    \phi(v),
\end{align*}
which is continuous at every point, including at $v=0$.

Thus, although the joint density of the vector running variable is discontinuous across almost every point of the boundary, the signed-distance density is continuous at zero. The discontinuities below and above $z_2=0$ exactly average out when $Z_2$ is integrated out. Consequently, a one-dimensional manipulation test applied to signed distance would not detect this violation of $H_0$.

\subsection{The signed distance test may be invalid under $H_0$}
\label{app:counterexamples}

This section provides an example in which a signed-distance test may be invalid under $H_0$ in \eqref{eq:H0}. The example involves an admittedly unusual boundary configuration that violates Assumption \ref{ass:assumption}. Our purpose is not to argue that such boundary configurations are empirically common, but rather to underscore the importance of formally establishing the conditions under which signed-distance tests are valid. In general, continuity of the joint density across the boundary would not, by itself, justify the conclusion that a signed-distance test is asymptotically valid.

Let $d=2$, where $Z=(Z_1,Z_2)$ has density
\begin{equation*}
    f_Z(z_1,z_2)
    ~=~
    \phi(z_1-1/2)\phi(z_2-1/2),
\end{equation*}
where $\phi$ denotes the standard normal density. Thus, $f_Z$ is continuous on $\mathbb R^2$, and $H_0$ in \eqref{eq:H0} holds.

Define the assignment rule by
\begin{equation*}
    T(z)~=~1\{z_2\geq1/2,\ z_1\neq1/2\}.
\end{equation*}
Thus, the treatment region is the upper half-plane excluding the vertical line at $z_1=1/2$. The boundary is
\begin{equation*}
    \mathcal B
    ~=~
    \{(z_1,1/2):z_1\in\mathbb R\}
    \cup
    \{(1/2,z_2):z_2\geq1/2\}.
\end{equation*}
The boundary is a finite union of line segments or rays, but Assumption \ref{ass:assumption}(c) fails. Indeed, at every point $(1/2,z_2)$ with $z_2>1/2$, the control region is locally contained in the line $\{z_1=1/2\}$ and therefore occupies a negligible fraction of every sufficiently small ball centered at that point. Figure \ref{fig:boundary-and-density} displays the boundary and the treatment and control regions.

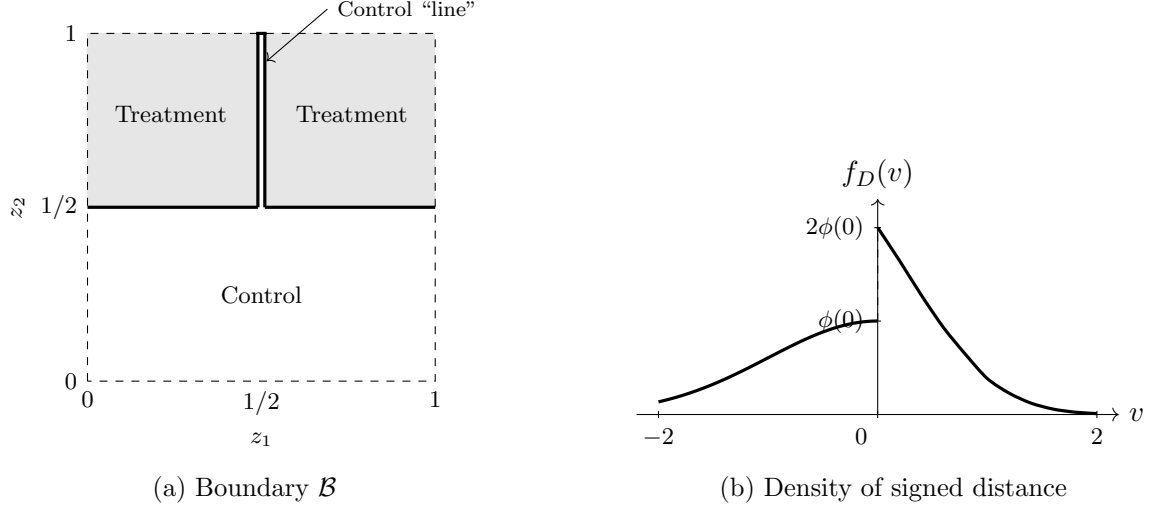
\begin{figure}[htbp]
\centering

\begin{minipage}[t]{0.48\textwidth}
\centering
\begin{tikzpicture}[scale=4.6]
    \fill[gray!20] (0,0.5) rectangle (1,1);

    \fill[white] (0.49,0.5) rectangle (0.51,1);

    \draw[dashed] (0,0) rectangle (1,1);

    \draw[very thick] (0,0.5) -- (0.49,0.5);
    \draw[very thick] (0.51,0.5) -- (1,0.5);
    \draw[very thick] (0.485,1) -- (0.515,1);

    \draw[very thick] (0.49,0.5) -- (0.49,1);
    \draw[very thick] (0.51,0.5) -- (0.51,1);

    \node at (0.24,0.77) {\footnotesize Treatment};
    \node at (0.76,0.77) {\footnotesize Treatment};
    \node at (0.5,0.25) {\footnotesize Control};

    \draw[->,thin] (0.69,1.07) -- (0.515,0.92);
    \node[anchor=west] at (0.69,1.07)
        {\scriptsize Control ``line''};

    \node[below] at (0,0) {\footnotesize $0$};
    \node[below] at (0.5,0) {\footnotesize $1/2$};
    \node[below] at (1,0) {\footnotesize $1$};

    \node[left] at (0,0) {\footnotesize $0$};
    \node[left] at (0,0.5) {\footnotesize $1/2$};
    \node[left] at (0,1) {\footnotesize $1$};

    \node[below=0.32cm] at (0.5,-0.05) {\footnotesize $z_1$};
    \node[rotate=90] at (-0.2,0.5) {\footnotesize $z_2$};
\end{tikzpicture}

\vspace{0.3em}
\small (a) Boundary $\mathcal B$
\end{minipage}
\hfill
\begin{minipage}[t]{0.48\textwidth}
\centering
\begin{tikzpicture}[xscale=1.45,yscale=3.1]
    \draw[->] (-2.2,0) -- (2.2,0) node[right] {$v$};
    \draw[->] (0,0) -- (0,0.92) node[above] {$f_D(v)$};

    \draw[very thick,domain=-2:0,samples=100,smooth]
        plot (\x,{0.39894228*exp(-\x*\x/2)});

    \draw[very thick]
        plot[smooth] coordinates {
            (0,0.7979)
            (0.2,0.6560)
            (0.4,0.5081)
            (0.6,0.3714)
            (0.8,0.2573)
            (1.0,0.1536)
            (1.2,0.0910)
            (1.4,0.0484)
            (1.6,0.0232)
            (1.8,0.0100)
            (2.0,0.0039)
        };

    \draw[dashed] (0,0.3989) -- (0,0.7979);

    \draw (-2,0.015) -- (-2,-0.015)
        node[below] {\footnotesize $-2$};
    \draw (2,0.015) -- (2,-0.015)
        node[below] {\footnotesize $2$};
    \draw (0,0.015) -- (0,-0.015)
        node[below left] {\footnotesize $0$};

    \draw (0.03,0.3989) -- (-0.03,0.3989)
        node[left] {\footnotesize $\phi(0)$};
    \draw (0.03,0.7979) -- (-0.03,0.7979)
        node[left] {\footnotesize $2\phi(0)$};
\end{tikzpicture}

\vspace{0.3em}
\small (b) Density of signed distance
\end{minipage}

\caption{\small Panel (a) shows the boundary $\mathcal B$ and the treatment and control regions within a portion of $\mathbb R^2$. The zero-width control segment at $z_1=1/2$ is displayed as a narrow white region for visibility, with its width exaggerated. Panel (b) plots the density of signed distance, whose left and right limits at zero are $\phi(0)$ and $2\phi(0)$, respectively.}
\label{fig:boundary-and-density}
\end{figure}

Let $D(z)$ denote signed distance from the boundary. For any $z\in\mathbb R^2$,
\begin{align*}
    D(z)~&=~\dist(z,\mathcal B)1\{T(z)=1\}
    -\dist(z,\mathcal B)1\{T(z)=0\}\\
    &\overset{(1)}{=}~
    -(1/2-z_2)1\{z_2<1/2\}
    +
    \min\{z_2-1/2,\lvert z_1-1/2\rvert\}1\{z_2\geq1/2\},
\end{align*}
where (1) follows from the definitions of $\mathcal B$ and $T(z)$. Consequently,
\begin{equation*}
    D~=~-(1/2-Z_2)1\{Z_2<1/2\}
    +\min\{Z_2-1/2,\lvert Z_1-1/2\rvert\}1\{Z_2\geq1/2\}.
\end{equation*}
From this, we obtain
\begin{equation*}
    F_D(v)~=~\Phi(v) 1\{v<0\} + (1-2\{1-\Phi(v)\}^2)  1\{v\geq 0\},
\end{equation*}
where $\Phi$ denotes the standard normal cumulative distribution function.
Differentiating yields
\begin{equation*}
    f_D(v)~=~
    \phi(v)1\{v<0\}
    +4\phi(v)\{1-\Phi(v)\}1\{v>0\}.
\end{equation*}
Since
\begin{equation*}
    \lim_{v\uparrow0}f_D(v)=\phi(0)
    \qquad\text{and}\qquad
    \lim_{v\downarrow0}f_D(v)=2\phi(0),
\end{equation*}
the density of signed distance is discontinuous at zero. This example establishes that continuity of $f_Z$ does not, by itself, imply continuity of $f_D$. The additional boundary regularity imposed in Assumption \ref{ass:assumption}(c) is necessary for this implication.

The treatment region in this example is admittedly unusual because the control region collapses to a line within the upper half-plane. Nevertheless, this configuration can be obtained as the limit of a simple family of treatment regions. In particular, for any $h>0$, consider the assignment rule
\begin{equation*}
    T_h(z)~=~1\{z_2\geq1/2,\ z_1<1/2-h\}
    +1\{z_2\geq1/2,\ z_1>1/2+h\}.
\end{equation*}
The corresponding control region contains a vertical strip of width $2h$ centered at $z_1=1/2$. For every fixed $h>0$, both the treatment and control regions are locally nonnegligible along the boundary. As $h\downarrow0$, however, the control strip collapses to the line $\{(1/2,z_2):z_2\geq1/2\}$, yielding the treatment region in our example.
\end{small}

\bibliography{BIBLIOGRAPHY}
\end{document}